\documentclass[article]{winnower}

\renewcommand{\thesubsection}{\thesection.\arabic{subsection}}
\usepackage{sectsty}
\sectionfont{\centering \large \normalfont \textsc  }
\subsectionfont{\centering \normalfont  \textit }
\subsubsectionfont{\centering}
\RequirePackage[noblocks]{authblk}
\usepackage[square,numbers]{natbib}

\usepackage{amssymb,amsmath,amsthm,mathtools,tikz-cd,listings,array,bbm,xcolor,mathrsfs}
\usepackage{soul,floatrow}
\newfloatcommand{capbtabbox}{table}[][\FBwidth]
\usepackage{graphicx,tocbibind}
\usepackage[left = 3cm,bottom=2.8cm, top = 2.85cm]{geometry}
\usepackage{multirow}
\usepackage{subcaption}
\usepackage{graphicx}

\newcommand{\tr}{\ensuremath{\text{tr}\,}}
\theoremstyle{plain}
\newtheorem{theorem}{Theorem}
\newtheorem{corollary}{Corollary}
\newtheorem{remark}{Remark}
\newtheorem{lemma}{Lemma}
\newtheorem{algo}{Algorithm}
\newtheorem{proposition}{Proposition}

\begin{document}
	
	\title{\textbf{Inference for partially observed Riemannian Ornstein--Uhlenbeck diffusions of covariance matrices}}
	
	\author[1]{\normalsize Mai Ngoc Bui \thanks{mai.bui@ucl.ac.uk}}
	\author[2]{Yvo Pokern \thanks{y.pokern@ucl.ac.uk}}
	\author[2,3]{Petros Dellaportas \thanks{p.dellaportas@ucl.ac.uk}}
	\affil[1]{Department of Applied Health Research, University College London, Gower Street, London, WC1E 6BT, U.K.}
	\affil[2]{Department of Statistical Science, University College London, Gower Street, London WC1E 6BT, U.K.}
	\affil[3]{Department of Statistics, Athens University of Economics and Business, Athens 10434, Greece.}


	\date{}
	
    \maketitle
	
\begin{abstract}
    \normalsize We construct a generalization of the Ornstein--Uhlenbeck processes  on the cone of covariance matrices endowed with the Log-Euclidean and the Affine-Invariant metrics. Our development exploits the Riemannian geometric structure of symmetric positive definite matrices viewed as a differential manifold. We then provide Bayesian inference for discretely observed diffusion processes of covariance matrices based on an MCMC algorithm built with the help of a novel diffusion bridge sampler accounting for the geometric structure. Our proposed algorithm is illustrated with a real data financial application.
\end{abstract}

\vspace{5pt}

	
\section{Introduction}

We are interested in Bayesian inference for diffusion processes of covariance matrices when only a discrete set of observations is available over some finite time period.  Our motivation stems from financial applications where diffusions have been often adopted to describe continuous-time processes
\citep{kalogeropoulos2010inference, stramer2010bayesian}, 
but an inferential framework to model realized covariances of asset log-returns is not available. This task is challenging not only because  the marginal likelihood of the data obtained in partially observed diffusions  is generally intractable, but also because  dealing with covariance matrices requires models that preserve their positive definiteness. The resulting complexity of Bayesian inference via MCMC sampling algorithms necessitates the development of dynamics in $\mathcal{SP}(n)$, the curved space of symmetric positive definite matrices in $\mathbb{R}^{n \times n}$, together with sampling algorithms for diffusion bridges in $\mathcal{SP}(n)$.

There is considerable work on stochastic differential equations (SDE's) defined on positive semidefinite matrices based on Wishart processes introduced by M.-F. Bru \citep{bru1991wishart} as a matrix generalisation of squared Bessel processes;  see for example, \citep{gourieroux2006continuous, gourieroux2009wishart, gourieroux2010derivative, buraschi2010correlation, barndorff2007positive}.  While Wishart processes seem natural candidates for Bayesian inference on $\mathcal{SP}(n)$, they lack geometric structure which, as will become evident in our model development, is a highly desirable property. As an example, O. Pfaffel \citep{pfaffel2012wishart} notes that simulation of Wishart processes via the Euler--Maruyama method fails to always generate points on $\mathcal{SP}(n)$ so a time-adjustment algorithm is necessary.

We construct a generalisation of Ornstein--Uhlenbeck (OU) processes on $\mathcal{SP}(n)$ by noting that their dynamics are naturally specified by the Riemannian geometric structure of $\mathcal{SP}(n)$ viewed as a  differential manifold endowed with the Euclidean, the Log-Euclidean (LE) and the Affine-Invariant (AI) metrics.  We emphasize the need to adopt the LE and the AI metrics which, unlike the Euclidean metric, achieves efficient sampling even from points lying near the boundary of $\mathcal{SP}(n)$. We then propose an MCMC algorithm that operates on $\mathcal{SP}(n)$ and alternates between imputing diffusivity-independent Brownian motions driving diffusion bridges between consecutive observations and approximating the likelihood with the Euler--Maruyama method adapted to the Riemannian setting via the exponential map. In particular, the construction for the AI metric required the development of a novel bridge sampling algorithm.  We demonstrate our methodology with simulated and real financial data.

The essence of our construction is based on the following key points.  We adopt an intrinsic point of view  of $\mathcal{SP}(n)$ equipped with either the LE or the AI metric \citep{elworthy1982stochastic}. We construct a Riemannian Brownian motion as the limit of a random walk along geodesic segments  using the exponential map for both metrics \citep{gangolli1964construction}.  We then proceed to construct a $d=n(n+1)/2$-dimensional OU process on $\mathcal{SP}(n)$ by adding a mean reverting drift and prove that its solution is equivalent to a solution of an SDE in $\mathbb{R}^d$ in the case of the LE metric. For the AI metric we establish existence and non-explosion of the newly proposed process. Finally, armed with the mathematical constructions, we proceed to the Bayesian estimation through a Bayesian data augmentation strategy in which the key required ingredient is the ability to sample from a diffusion bridge on $\mathcal{SP}(n)$; see \citep{roberts2001inference}. 

Sampling from Brownian bridges has played an important role in Bayesian inference for diffusions.  When the transition function is analytically unavailable, MCMC data augmentation sampling strategies that impute partial trajectories via bridge samplers have been used to numerically approximate the transition functions, see \citep{elerian2001likelihood, eraker2001mcmc, roberts2001inference}. The use of bridge sampling has a long history in the inference for diffusions starting from \citep{pedersen1995consistency}. Recent advances include the modified diffusion bridge by G. B. Durham \& A. R. Gallant \citep{durham2002numerical} and its modifications, see \citep{golightly2008bayesian, stramer2010bayesian, lindstrom2012regularized} and ideas based on sequential Monte Carlo \citep{del2015sequential, lin2010generating}. There has been a line of research based on ideas of  B. Delyon \& Y. Hu \citep{delyon2006simulation} that uses guided and residual proposal densities, see \citep{van2017bayesian, schauer2017guided, whitaker2017improved}.  Finally, a recent promising approach is based on \citep{bladt2005statistical},  see \citep{mider2019simulating}.  We contribute to this literature by proposing a sampling strategy to sample from a diffusion bridge on $\mathcal{SP}(n)$ with AI metric which can be viewed as a guided proposal density for our MCMC sampling according to the ideas in \citep{delyon2006simulation}.

We investigate with both simulated and real data the performance of our proposed diffusion processes with the three metrics. We demonstrate that LE and AI metrics should be preferred to the Euclidean metric and we illustrate that both LE and AI metrics, unlike the Euclidean metric which neglects the geometric structure of $\mathcal{SP}(n)$,  do not have the problem of the swelling effect \citep{arsigny2007geometric} or the difficulties when sampling near the boundary of $\mathcal{SP}(n)$.  We have also found that the diffusion based on the LE metric, compared with the AI metric,  leads to greater anisotropy which is more evident when conditioning on matrices with eigenvalues close to zero.  Our financial data example is chosen to illustrate this exact point: one can use diffusions on $\mathcal{SP}(n)$ with LE and AI metrics for pricing or portfolio construction even when the dynamics on $\mathcal{SP}(n)$ operate near the boundary.

\section{Riemannian geometry for covariance matrices}
\subsection{Preliminary of Riemannian geometry}\label{Preliminary of Riemannian geometry}
	 
	Smooth manifolds are motivated by the desire to extend the differentiation property to curved spaces that are more general and complicated than $\mathbb{R}^d$. This is achieved by considering coordinate charts, i.e. functions that map small patches of the given manifold $\mathcal{M}$ to open sets in Euclidean space. It is then possible to define smooth curves $\gamma: [0,T] \rightarrow \mathcal{M}$ which pass through some point $\gamma(0)=P \in \mathcal{M}$, and whose velocity vectors $\gamma'(0)$ at $P$ are known as tangent vectors constituting a vector space $T_P\mathcal{M}$, the tangent space at $P$. A Riemannian metric tensor $g$  assigns to each point $P$ on $\mathcal{M}$ a bilinear function $g_P$ on $T_P\mathcal{M} \times T_P\mathcal{M} $ which is symmetric and positive definite. Smooth manifolds equipped with Riemannian metric tensors are called Riemannian manifolds and are characterised by their corresponding Riemannian metrics.
	 
	As differentiability is so special with a smooth manifold, one initially considers first order derivatives: Firstly, these include vector fields $X$ which assign to each point $P\in\mathcal{M}$ a tangent vector $v\in T_P\mathcal{M}$ and give rise to the tangent bundle $T\mathcal{M}$, the disjoint union of all points' tangent spaces. The set of all smooth vector fields is denoted $\Gamma(T\mathcal{M})$. Secondly, differentials of smooth maps from one manifold to another which give rise to linear maps from one tangent space to another are also first order derivatives. Then, consideration moves on to second order derivatives such as the derivative of a vector field with respect to another vector field. Suppose $x =\{x^{(i)}\}_{i=1}^d $ is a local chart on an open neighbourhood $\mathscr{U}$ of some point $P$ on the manifold $\mathcal{M}$ of dimension $d$, the vector fields $X_i = \partial/\partial x^{(i)}$ span the tangent space $T_P\mathcal{M}$ at each $P \in \mathscr{U}$. The covariant derivative, denoted as $\nabla_XY$, explores how a vector field $Y$ varies along another vector field $X$ and the Christoffel symbols $\Gamma_{ij}^k$ are functions on $\mathscr{U}$ defined uniquely by the relation
	$
		\big(\nabla_{X_i}X_j\big)_P = \sum_{k=1}^d \Gamma_{ij}^k(P) \,X_k
	$ for all $P \in \mathscr{U}$,
	see \citep{boothby1986introduction}. Moreover, using an orthornormal basis $\{E_i(P)\}_{i=1}^d$ with respect to the metric tensor $g$, one can simply compute the Riemannian gradient of any smooth function, i.e. $f \in C^{\infty}(\mathcal{M})$ as
	$
		(\nabla f)_P =  \sum_{i=1}^n (E_if)_P \, E_i(P).
	$
	Here $(E_if)_P$ can be understood as the  differential of $f$ at $P$ in the direction of $E_i(P)$. 

	The connection $\nabla$ allows us to transport a tangent vector from one tangent space to another on $\mathcal{M}$ in a parallel manner. A vector field $V$ along the curve $\gamma$ on $\mathcal{M}$ is said to be parallel along the curve if $\nabla_{\gamma'(t)}V = 0$ at every point on the curve \citep{lovett2010differential}.
	Furthermore, any curve $\gamma(t)$ on $\mathcal{M}$ that satisfies $\nabla_{\gamma'(t)}\gamma'(t) = 0$ at all points  on the curve is called a geodesic. They are locally defined as minimum length curves over all possible smooth curves that connect two given points on the Riemannian manifold \citep{caseiro2012nonparametric}. The exponential map, $\text{Exp}_P : T_P\mathcal{M} \rightarrow \mathcal{M}$ computes the point at which a geodesic starting from $P$ in the direction $\nu \in T_P\mathcal{M}$ ends after one time unit.
	In general, $\text{Exp}_P$ is bijective only from a small neighbourhood $\mathscr{V} \subset T_P\mathcal{M}$ to a neighbourhood  $\mathscr{U} \subset \mathcal{M}$ of $P$ on which the inverse map of $\text{Exp}_P$ can be defined uniquely: this is called the logarithm map $\text{Log}_P = \text{Exp}_P^{-1}$.

	We focus on the space  of $ n \times n$ symmetric positive definite matrices $\mathcal{SP}(n)$ with dimension $d = n(n+1)/2$, which is a sub-manifold of the space of symmetric matrices $\mathcal{S}(n)$. Any metric on the space of $n \times n$ invertible matrices $\mathcal{GL}(n)$ induces a metric on $\mathcal{SP}(n)$. For example, the Frobenius inner product induces the so-called Euclidean metric $g^\text{E}$ which, by noting that the tangent space at any point on $\mathcal{SP}(n)$ is simply $\mathcal{S}(n)$, is given by
	\begin{equation}
		g^{\text{E}}(S_1,S_2) = \langle S_1, S_2 \rangle_F = \text{tr}(S_1^TS_2) \text{ for } S_1, S_2 \in \mathcal{S}(n),
		\label{eq:Euclidean metric}
	\end{equation}
	where tr stands for the trace operator on $\mathcal{GL}(n)$. 
	
	Since the symmetry property is not preserved under the usual matrix multiplication, Arsigny et al.  \citep{arsigny2007geometric} proposed the use of the matrix exponential/logarithm functions:
	\[
    	P \odot Q =\exp(\log P + \log Q),
    	\,\,\lambda \ast P = \exp(\lambda \log P)  \text{ for } P, Q  \in \mathcal{SP}(n) \text { and } \lambda \in \mathbb{R}.
	\]
	Equipping  $\mathcal{SP}(n)$ with $\odot$, $\mathcal{SP}(n)$ becomes an Abelian group as matrix addition is commutative. Since both matrix exponential and logarithm are diffeomorphisms, $(\mathcal{SP}(n),\odot)$ is in fact a Lie group. Moreover, we can get a vector space structure with $(\mathcal{SP}(n),\odot,\ast)$ since $(\mathcal{SP}(n),\odot)$ is isomorphic and diffeomorphic to $(\mathcal{S}(n),+)$ via the matrix logarithm function $\log$. Therefore, even though $\mathcal{SP}(n)$ is not a vector space, we can identify $\mathcal{SP}(n)$ with a vector space by considering its image under the matrix logarithm. To obtain a metric, the Frobenius inner product on the Lie algebra (i.e. $T_{I_n}\mathcal{SP}(n) = \mathcal{S}(n)$ for an $n\times n$ identity matrix $I_n$) can be extended by left-translation and becomes a bi-invariant metric $g^{\text{LE}}$ on $\mathcal{SP}(n)$. This metric is called the Log-Euclidean (LE) metric,
	\begin{equation}
	    g_P^{\text{LE}}(S_1,S_2) = \big\langle d\log_P(S_1), d \log_P(S_2)\big\rangle_F \,\,\, \text{ for } S_1, S_2 \in \mathcal{S}(n),
	    \label{eq:Log-Euclidean metric}
	\end{equation}
	where $d\log_P(S)$ is the differential of the matrix logarithm function at $P$ acts on $S$. In fact, $d\log_P(S)$ is identical to the derivative of matrix logarithm function at $P$ in direction $S$, denoted by $D_P\log.S$,  for any $P \in \mathcal{SP}(n)$ and $S \in \mathcal{S}(n)$ \cite{arsigny2007geometric}. As the name suggests, the Log-Euclidean metric is simply the Euclidean metric in the logarithmic domain. Equipping $\mathcal{SP}(n)$ with $g^{\text{LE}}$, we gain invariance with respect to inversion, $g^{\text{LE}}_P(A,B)=g^{\text{LE}}_P(A^{-1},B^{-1})$; and 
	invariance under similarity transform $\hat{A}=R^{-1} A R$ (where $R$ is an $n \times n$ invertible matrix): $g^{\text{LE}}_P(A,B)=g^{\text{LE}}_P(\hat{A},\hat{B})$.  Finally, matrices having non-positive eigenvalues are infinitely far away from any covariance matrix.

	Besides the LE metric, another metric on $\mathcal{SP}(n)$, namely the Affine-Invariant (AI) metric, has been studied intensively \citep{caseiro2012nonparametric,moakher2005differential,pennec2006statistical}. There are many ways of defining this metric whose name arises from the group action $(\star)$ on $\mathcal{GL}(n)$ that gives rise to Riemannian metrics invariant under this action, where
	\[
	    R  \star S = R S R^T\, \text{ for } S \in \mathcal{S}(n),\, R \in \mathcal{GL}(n).
	\] 
	The AI metric $g^{\text{AI}}$ is thus defined to satisfy
	\begin{equation}
		g_P(S_1,S_2) = g_{R\star P}(R \star S_1, R \star S_2) \text{ for } S_1, S_2 \in \mathcal{S}(n),\,  P \in \mathcal{SP}(n)\text{ and } R \in \mathcal{GL}(n).
	\label{eq:affine-invariance}
	\end{equation}
	Choosing $g_I^{\text{AI}}$ to be the Frobenius inner product  $\left\langle\cdot,\cdot\right\rangle_F$ then defines the metric on  $\mathcal{SP}(n)$:
	\begin{equation}
		g_P^{\text{AI}}(S_1,S_2) = 
		\left\langle P^{-1/2} \star S_1 , P^{-1/2}\star S_2 \right\rangle_F.
		\label{eq:Affine-Invariant metric}
	\end{equation}
	Alternatively, the AI metric can be obtained from the theory of the multivariate normal distribution through the Fisher information  \citep{moakher2011riemannian,skovgaard1984riemannian}. 
	
	A metric tensor can also be expressed in the form of a matrix function $G \in \mathcal{SP}(d)$  with respect to some basis. For example, the AI metric can be expressed explicitly in matrix form $G(P)$ at any $P \in \mathcal{SP}(n)$ with respect to the standard symmetric basis $\mathfrak{B}_d$ on $\mathcal{S}(n)$, defined in equation~\eqref{eq:standard symmetric basis}:
	\begin{equation}
		G(P) = D_n^T\cdot \left(P^{-1} 	\otimes P^{-1}\right)\cdot D_n  \,\,\, \,\,\,\&\,\,\,\,\,\, G^{-1}(P) = D_n^\dag\cdot \left(P	\otimes P\right)\cdot (D_n^\dag)^T,
		\label{eq:AI metric in matrix form}
	\end{equation}
	where $D_n \in \mathbb{R}^{n^2 \times d}$ is a constant matrix (referred to as the duplication matrix), that satisfies $\text{vec}(P) = D_n \,\nu(P)$ with $\nu (P)$ containing all independent entries of $P$ and $D^\dag_n$ is the Moore-Penrose inverse of $D_n$ \citep{moakher2011riemannian}. Similarly to the LE metric, the AI metric is inversion-invariant and any covariance matrix is at infinite distance to any non-positive definite matrix. While the AI metric attains full affine-invariance, i.e. equation~\eqref{eq:affine-invariance} holds for any invertible matrices, the LE metric only achieves similarity invariance, i.e. equation~\eqref{eq:affine-invariance} only holds for orthogonal matrices.
	
    Finally, we summarize some results about the Euclidean metric in equation~\eqref{eq:Euclidean metric}, the Log-Euclidean metric in equation~\eqref{eq:Log-Euclidean metric} and the Affine-Invariant metric in equation~\eqref{eq:Affine-Invariant metric} into Table~\ref{tab:metric summary}, which includes explicit formulae  of the exponential/logarithm maps, geodesics and distance square  \citep{arsigny2007geometric,pennec2006statistical}. The Frobenius norm $||A||_F$ is defined by the Frobenius inner product, i.e. $||A||_F = \sqrt{\langle A,A\rangle_F} = \sqrt{\tr(A^TA)}$ for any $A \in \mathcal{GL}(n)$.
    
	\begin{table}[t]
		\begin{center}
			\begin{tabular}{c|c|c|c} 
				&  Euclidean & Log-Euclidean & Affine-Invariant\\
				\hline
				$\text{Exp}_P (S)$& $S + P$ & $\exp(\log P + D_P\log.S)$ & $P^{1/2}\star \exp(P^{-1/2} \star S )$\\
				$\text{Log}_P(Q)$ & $Q - P$ &$D_{\log P}\exp.(\log Q - \log P)$  &$P^{1/2}\star \log(P^{-1/2}\star Q)$  \\
				$\gamma_{(P,Q)}(t)$&$P + t(Q-P)$& $\exp(\log P + t( \log Q - \log P))$&$P^{1/2}\star \exp(P^{-1/2}\star t Q)$\\
				$\text{d}^2(P,Q)$& $||Q-P||_F^2$& $||\log Q - \log P||_F^2$&$\| \log(P^{-1/2} \star Q)||_F^2$ 		
			\end{tabular}
		\end{center}
		\caption{Explicit formulae of exponential map, logarithm map, geodesic and distance square for the Euclidean, Log-Euclidean and Affine-Invariant metrics. For the Log-Euclidean and Affine-Invariant cases, these will be denoted as $\text{Exp}^\text{LE}$, $\text{Exp}^\text{AI}$, $\text{Log}^\text{LE}$, $\text{Log}^\text{AI}$ and $d_\text{LE}$, $d_\text{AI}$, respectively.}
		\label{tab:metric summary}
	\end{table}
	
    Let us fix an orthonormal basis $\mathfrak{B}_d = \{S_i\}_{i=1}^d$ with respect to the Frobenius inner product on the tangent space $\mathcal{S}(n)$ of $\mathcal{SP}(n)$, where $d = n(n+1)/2$: 
	\begin{align}
	    & S_i = e_{ii}^{(n)}  \hspace{1.5cm} \text{ for } 1 \leq i \leq n\label{eq:standard symmetric basis}\\ 
	    & S_{n+1} = \left(e_{21}^{(n)} + e_{12}^{(n)}\right)/\sqrt{2} ,\, S_{n+2} = \left(e_{31}^{(n)} + e_{13}^{(n)}\right)/\sqrt{2},\, S_{n+3} = \left(e_{32}^{(n)} + e_{23}^{(n)}\right)/\sqrt{2}  ,\, \ldots \nonumber
	\end{align}
    Here, $\{S_i\}_{i=1}^n$ has all entries zero except the $j\text{th}$ entry on the diagonal being one. The remaining $\{S_i\}_{i=n+1}^d$ are obtained by adding, with $i > j$, the single-entry matrix $e_{ij}^{(n)}$ with one at the $(i,j)$th entry and zero elsewhere to its transpose and dividing to $\sqrt{2}$ so that it has unit Frobenius norm. We call $\mathfrak{B}_d$ the standard symmetric basis of $\mathcal{S}(n)$.

\subsection{The importance of Riemannian geometry to $\mathcal{SP}(n)$}\label{The importance of Riemannian geometry to SP(n)}
	We discuss two major reasons that necessitate the use of Riemannian geometry: easier sampling close to the boundaries of $\mathcal{SP}(n)$ and no swelling effects. One may additionally argue that other properties, such as inversion-invariance and similarity-invariance for the LE and AI metrics, and affine-invariance for the AI metric, may be useful in calculations for complicated computational algorithms.

	Although the Frobenius inner product on $\mathcal{SP}( n)\subset \mathcal{GL}(n)$ is simple, it is problematic because non-covariance matrices are only a finite distance away from covariance matrices. As Table~\ref{tab:metric summary} illustrates,  LE and AI metrics do not suffer from this problem as the involvement of the matrix logarithm guarantees that non-covariance matrices are at infinite distance from any point on $\mathcal{SP}(n)$.  Therefore, they avoid the undesirable inequality constraints that are required in the Frobenius induced geometry to ensure positive definiteness and whose number grows quadratically with $n$. As will become evident in our simulation experiments, this turns out to be a highly desirable property because it facilitates sampling close to the boundary of $\mathcal{SP}(n)$.
	
	\begin{figure}[t]
    	\includegraphics[width = \textwidth]{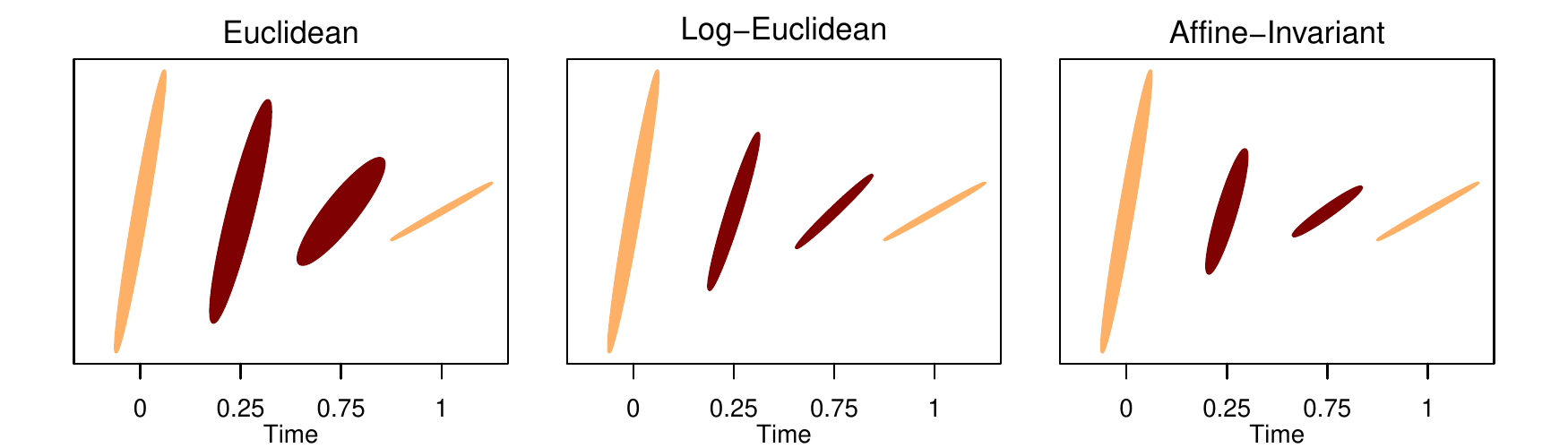}
        \caption{Comparison of three metric tensors on $\mathcal{SP}(2)$: Dark ellipses represent points at time $t  = 0.25,\, 0.75$ on the geodesic connecting $P_0$ at  $t= 0$ and $P_1$ at $t= 1$. The length of the axes are the square root of the eigenvalues.}
	\label{fig:AI_LE_ellipse}
	\end{figure}	
	
    The determinant of a covariance matrix measures the dispersion of the data points from a multivariate normal distribution. For the Euclidean metric, the geodesic connecting two fixed points often contains points with a larger determinant than the two fixed points, and the difference can get extremely large whenever the fixed points lie near the boundary of $\mathcal{SP}(n)$. This problem is referred to as the swelling effect \citep{arsigny2007geometric,dryden2009non,jung2015scaling}. In many contexts, the swelling effect is described as undesirable because the level of dispersion should remain close to the given information obtained by the observations of covariance matrices \citep{arsigny2007geometric,chefd2004regularizing,fletcher2007riemannian,tschumperle2001diffusion}.  The LE and AI metrics avoid this swelling effect.
    Moreover, points on LE and AI geodesics at corresponding times have the same determinants. These determinants are the result of linear interpolation in the logarithmic domain; this can be proved by following similar lines as in \citep{arsigny2007geometric}.

    A visual illustration of the swelling effect is provided in Figure~\ref{fig:AI_LE_ellipse} where two intermediate points on the geodesic connecting 
	    \[P_0 = \begin{pmatrix}
				0.4&0.3\\
				0.3&0.4
			\end{pmatrix}
	 \text{ at } t=0\,\,\,\, \text{ and }\,\,\,\, 
	  P_1 = \begin{pmatrix}
				1&0.1\\
				0.1&0.02
			\end{pmatrix}
		\label{eq:ellipse example}
     \text{ at } t=1\]
    are shown for each metric.  Notice the swelling effect in the case of the Euclidean metric and also that the geodesic of the LE metric has points with exaggerated anisotropy; for more details on this phenomenon see \citep{arsigny2006log,dryden2009non}. In general, whether anisotropy constitutes a problems depends on the application of interest.

\section{Stochastic processes on $\mathcal{SP}(n)$}
\subsection{Overview}
    SDEs on manifolds present additional complications over the Euclidean setting which we deal with in three stages: firstly, we introduce the general notion of SDEs on manifolds. Secondly, we discuss Brownian Motion (definition, local representation, horizontal lift, construction via Euler-Maruyama with the exponential map and non-explosion) in subsection \ref{Brownian motion class} and thirdly we proceed to the Ornstein--Uhlenbeck class in subsection \ref{Ornstein--Uhlenbeck (OU) class}.
    
    Since the curvature of Riemannian manifolds makes direct use of Euclidean stochastic analysis prohibitively hard, a common strategy in stochastic analysis is to adopt the extrinsic view that involves embedding the manifold in a higher dimensional Euclidean space using the Nash embedding theorem; see for example,  \citep{hsu2002stochastic,elworthy1982stochastic,gangolli1964construction}. 
    The approach benefits from existing theory for SDE's on Euclidean space but a suitable coordinate system is often not explicitly available or too inconvenient, so this approach has limited practical use. 
    
    In this paper, we work with a probability space $(\Omega,\mathfrak{F},\mathbb{P})$ equipped with the filtration $\mathfrak{F}_*= \{\mathfrak{F}_t, \, t \geq 0\}$ of $\sigma$-fields contained in $\mathfrak{F}$. Here, we assume that $\mathfrak{F} = \underset{t \uparrow \infty}{\lim} \mathfrak{F}_t$, while $\mathfrak{F}_*$ is right continuous and every null set (i.e. a subset of a set having measure zero) is contained in $\mathfrak{F}_t$.

	On a Riemannian manifold $\mathcal{M}$ of dimension $d$, E. P. Hsu \citep{hsu2002stochastic} writes the SDE driven by smooth vector fields $V_0, \ldots, V_l$ by an $\mathbb{R}^{l+1}$-valued semi-martingale $Z_t$ (with initial condition $P \in \mathfrak{F}_0$) as
	\begin{equation}
		dX_t = \sum_{i=0}^l V_i(X_t) \circ dZ^{(i)}_t \hspace{1.5cm} (X_0 = P).
		\label{eq:SDE vector field}
	\end{equation}
	$Z^{(i)}_t$ could be a deterministic component, such as time, or a stochastic component,  such as a Brownian motion: this corresponds to the usual distinction of some $V_i$ as drift and other $V_i$ as diffusivity. We call $X_t$ an $\mathcal{M}$-valued semi-martingale defined up to a $\mathfrak{F}_*$ stopping time $\tau$ if it satisfies 
	\begin{equation}
		f(X_t) = f(P) + \sum_{i=0}^l\int_0^t V_if(X_s)\circ dZ_s^{(i)} \hspace{1cm} 0 \leq t < \tau\,\,, \,\, f \in C^{\infty}(\mathcal{M}), 
		\label{eq:SDE vector fields Stratonovich}
	\end{equation}
	where the integrals above are in the Stratonovich sense, and converting them to Ito sense yields:
	\begin{equation}
		f(X_t) = f(P) + \sum_{i=0}^l\int_0^t V_if(X_s)\, dZ_s^{(i)} + \frac{1}{2}\sum_{i,j=0}^l\int_0^t(\nabla_{V_j}V_i)f(X_s)\, d[ Z^{(i)},Z^{(j)} ]_s .
		\label{eq:SDE vector fields Ito}
	\end{equation}
	Here, $[Z^{(i)},Z^{(j)}]_t$ stands for the usual quadratic covariation of $Z_t^{(i)}$ and $Z_t^{(j)}$ defined on the Euclidean space. Comparing to equation~\eqref{eq:SDE vector fields Stratonovich}, the additional terms in equation~\eqref{eq:SDE vector fields Ito} arise from the non-trivial chain rule in the Ito case. The Stratonovich representation in equation~\eqref{eq:SDE vector fields Stratonovich} brings simplicity and is invariantly defined whence it is usually preferred for SDE's on manifolds \citep{elworthy1982stochastic}.

\subsection{Brownian motion class}\label{Brownian motion class}

	Since the infinitesimal generator of Brownian motion on Euclidean space is $\Delta/2$, with $\Delta$ the usual Laplace operator, Brownian motion on a Riemannian manifold $\mathcal{M}$ can also be defined as a diffusion process generated by $\Delta_{\mathcal{M}}/2$ where $\Delta_{\mathcal{M}}$ denotes the Laplace-Beltrami operator. The resulting Brownian Motion $X_t$ can be expressed in local coordinates
    using the standard Brownian motion $B_t$ on $\mathbb{R}^d$ by writing $\Delta_{\mathcal{M}}$ in local coordinates: 
	\begin{equation}
		dX_t^i = \sum_{j=1}^d (G^{-1/2})_{ij}(X_t)\,dB_t^{(j)} - \frac{1}{2} \sum_{k,l=1}^d (G^{-1})_{lk}(X_t)\,\Gamma_{kl}^i(X_t)dt \hspace{1cm} 1 \leq i \leq d,
		\label{eq:BM in local}
	\end{equation} 
	where $(G^{-1})_{ij}$ is the $(i,j)$-entry of $G^{-1}$ which is the matrix form of the metric tensor $g$ of the manifold $\mathcal{M}$, see \citep{elworthy1982stochastic,hsu2002stochastic}. 

	We choose to instead adopt the intrinsic viewpoint that studies Riemannian manifolds via their metric or connection which enables us to write down less cumbersome SDEs with more readily interpretable parameters. Thus, let us introduce a frame $u : \mathbb{R}^d \rightarrow T_P\mathcal{M}$ at $P \in \mathcal{M}$ (this is an isomorphism of vector spaces with inner product) and the frame bundle $\mathscr{F}(\mathcal{M})$ which is the collection of all frames for all $P \in \mathcal{M}$. Let us fix the standard basis $\{e_i\}_{i=1}^d$ on $\mathbb{R}^d$ and express $u$ in local coordinates $\{E_1,\ldots,E_d,e_1,\ldots,e_d\}$ in some neighbourhood $U$ covering $P \in \mathcal{M}$ as $u = (P,\zeta)$ with $\zeta = (\zeta^i_j) \in \mathbb{R}^{d\times d}$ the coefficients with respect to the orthonormal basis $\{E_i(Q)\}_{i=1}^d$ on $T_P\mathcal{M}$ for all $Q \in U$. This means that for any vector $e \in\mathbb{R}^d$ with coordinates $\epsilon_i$, i.e. $e = \sum_{i=1}^d \epsilon_i e_i \in \mathbb{R}^d$: 
	\[
	    u(e) = \sum_{i,j=1}^d \epsilon_i \,\zeta_j^i\, E_i(P) \in T_P \mathcal{M}.
	\]

	Moreover, $\mathscr{F}(\mathcal{M})$ is again a smooth manifold, and the canonical projection map $\pi : \mathscr{F}(\mathcal{M}) \rightarrow \mathcal{M}$ is smooth \citep{hsu2002stochastic}. If $u_t$ is a smooth curve on $\mathscr{F}(\mathcal{M})$ and for each $e \in \mathbb{R}^d$ the vector field $u_t(e)$ of $\mathcal{M}$ is parallel along the curve $\pi(u_t)$, $u_t$ is called horizontal curve on $\mathscr{F}(\mathcal{M})$. For any smooth curve $\gamma$ on $\mathcal{M}$, there is a corresponding horizontal curve on $\mathscr{F}(\mathcal{M})$ (unique up to choice of initial condition $u_0$) which is referred to as the horizontal lift of $\gamma$. This definition carries over to the horizontal lift of a tangent vector on $\mathcal{M}$. We define the anti-development of $\gamma$ as
	\[
	    w_t= \int_0^tu_s^{-1}(\gamma'(s)) \, ds,
	\]
	where $u_t$ is the horizontal lift  of $\gamma(t)$ on $\mathcal{M}$. While this anti-developement $w_t$ is guaranteed to exist and is uniquely defined up to the initial conditions $u_0$ and $\gamma(0)$, its computation is often difficult. To address this issue, we establish the following:
	
	\begin{proposition}
	On $\mathcal{SP}(n)$, in the case of the LE metric, the horizontal lift of a smooth curve can be explicitly expressed in local coordinates. For the AI metric, a first order Euler approximation for the horizontal lift of a geodesic can be explicitly computed.
	\label{AI_LE_stochastic development}
	\end{proposition}

   More details about the stochastic development on $\mathcal{SP}(n)$ are discussed in Subsection~\ref{supp:absolute continuity} of the Supplementary Material. 
    
	By carrying out a similar process to the case of a smooth curve, we obtain a corresponding horizontal semi-martingale $U_t$ on the frame bundle and an anti-development $W_t$ on $\mathbb{R}^d$ to a semi-martingale $X_t$ on $\mathcal{M}$. Up to a choice of initial conditions, this relationship is one-to-one. The process of transforming $W_t$ to $X_t$ is called stochastic development \citep{elworthy1982stochastic}. One particularly important result is that one can define Riemannian Brownian motion on $\mathcal{M}$ with the connection $\nabla$ by having its anti-development $W_t$ be the standard Euclidean Brownian motion. On the sphere $\mathcal{S}^2$, stochastic development is intuitively described as  ``rolling without slipping": if we have a path of Brownian motion $W_t$ on a flat paper (this paper acts as the tangent plane of $\mathcal{S}^2$), rolling the sphere along the path of $W_t$ without slipping results in a trajectory on $\mathcal{S}^2$ which turns out to be a path of the Riemannian Brownian motion on $\mathcal{S}^2$. 
	
    Yet another alternative construction of Riemannian Brownian Motion 
	uses the Euler--Maruyama approximation, which employs the exponential map. We call thus method the exponential adapted Euler--Maruyama method, that is
	\begin{equation}
		X_{t+\delta t} = \text{Exp}_{X_t}\left\{ \sum_{i=1}^d (B^{(i)}_{t+\delta t} - B^{(i)}_t) \,E_i(X_t)\right\} \hspace{0.5cm} \text{ for } \delta_t > 0 ;
		\label{eq:BM via exp}
	\end{equation}
	see for example, \citep{baxendale1976measures,manton2013primer,mckean1960brownian}. As $\delta t \rightarrow 0$, $X_t$ converges to the Riemannian Brownian motion in distribution if there exists a global basis field on the tangent bundle $\mathscr{F}(\mathcal{M})$, i.e. $\mathcal{M}$ is a parallelizable manifold \citep{gangolli1964construction,mckean1960brownian,joergensen1978construction}. Since $\mathcal{SP}(n)$ endowed with either the LE or the AI metric is geodesically complete (i.e. the exponential map is a global diffeomorphism) and parallelizable, the approximation method in  equation~\eqref{eq:BM via exp} becomes more convenient and efficient in our case. 
	
	The transition density function $p_{\mathcal{M}}(s,P;t,Q)$ of the Riemannian Brownian motion exists but usually no explicit expression is available, while on the Euclidean space it is simply the Gaussian distribution \citep{hsu2002stochastic,elworthy1982stochastic}. On the Euclidean space, Brownian motion does not explode in finite time, and if this holds in the Riemannian setting, that is 
	\begin{equation*}
		\int_{\mathcal{M}} p_{\mathcal{M}}(0,P;t,Q)\, dQ= 1 \hspace{1cm} \text{ for all } P \in \mathcal{M} \,\,\, \&\,\,\,\, 0 <t < \infty,
	\end{equation*}
	then the manifold $\mathcal{M}$ is said to be stochastically complete. It turns out that $\mathcal{SP}(n)$ equipped either with the LE or the AI metric is also stochastically complete, see  Subsection~\ref{supp:absolute continuity} of the Supplementary Material.

\subsection{Ornstein--Uhlenbeck (OU) class}\label{Ornstein--Uhlenbeck (OU) class}
    Adopting the intrinsic point of view, we present a construction of an OU class of processes on $\mathcal{SP}(n)$ for both the LE and the AI metrics. In analogy with the Euclidean OU process, we start with Brownian motion and add a mean-reverting drift which pushes the process toward the point of attraction. In the Euclidean setting, this drift is simply given as the gradient of the squared distance between the process and the point of attraction. We translate this idea to manifolds by using the covariant derivative in the place of the Euclidean gradient. This is similar to the treatment of the drift term by  V. Staneva \& L. Younes \citep{staneva2017learning} for shape manifolds.
 
    Let us define the OU process on $\mathcal{SP}(n)$ to be the solution of the following SDE with model parameters $\theta \in \mathbb{R}_{>0}$, $M \in \mathcal{SP}(n)$ and $\sigma \in \mathbb{R}_{>0}$ :
	\begin{equation}
		dX_t = -\frac{\theta}{2}\,\nabla_{X_t} \left\{d^2(X_t,M)\right\} \, dt + F_{X_t}( \sigma \, dB_t) \hspace{1 cm } (X_0 = P).
		\label{eq:SDE of OU}
	\end{equation}
	This uses the smooth function $F: \mathcal{SP}(n) \times \mathbb{R}^d \rightarrow \Gamma(T\,\mathcal{SP}(n))$ defined by:
	\begin{equation}
		F_Q(e) = \sum_{i=1}^d \epsilon_i\,E_i(Q) \,\,\, \text{ with } \,\,\,e = \sum_{i=1}^d \epsilon_i \,e_i \in \mathbb{R}^d,\, Q \in \mathcal{SP}(n) \text{ and } d = \frac{n(n+1)}{2},
		\label{eq:OU F}
	\end{equation}
	where $\{E_i\}_{i=1}^d$ is the orthonormal basis field  on $T\, \mathcal{SP}(n)$ with respect to the given metric tensor. The following Proposition  demonstrates that the covariant derivative chosen for the drift in the  SDE~\eqref{eq:SDE of OU} is explicitly computable; the proof is presented in Proposition~\ref{supp:D_LE_AI} of the Supplementary Material. 
	\begin{proposition}[Riemannian gradient of distance squared on $\mathcal{SP}(n)$] \hspace{1cm}
	\begin{enumerate}
	    \item[(i)](LE metric). 	The set $\mathfrak{B}^{\text{LE}}_d = \{E_i^{\text{LE}}\}_{i=1}^d$ is an orthonormal frame on the tangent bundle $T\mathcal{SP}(n)$, where for any $P \in \mathcal{SP}(n)$:
	    \begin{equation*}
		    E_i^{\text{LE}}(P) =  \big(d_{\log P}\big)^{-1}(S_i) = D_{\log P}\exp. S_i  \,\,\,\, \text{ for } 1 \leq i \leq d.
		\label{eq:LE orthonormal basis}
	    \end{equation*}
        Moreover, the Riemannian gradient of distance squared for any fixed point $Q \in \mathcal{SP}(n)$ is   
    	$
    		\Big(\nabla \text{d}^2_{\text{LE}}(P,Q)\Big)_P =  -2D_{\log P} \exp.(\log Q - \log P) = -2\, \text{Log}_P^{\text{LE}}(Q).
    	$
    	\item[(ii)](AI metric). The set $\mathfrak{B}^{\text{AI}}_d = \{E_i^{\text{AI}}\}_{i=1}^d$ is an orthonormal frame on the tangent bundle $T\mathcal{SP}(n)$, where for any $P \in \mathcal{SP}(n)$:
	    \begin{equation*}
		    E_i^{\text{AI}}(P) =  P^{1/2} \star S_i  \,\,\,\, \text{ for } 1 \leq i \leq d.
		\label{eq:AI orthonormal basis}
	    \end{equation*}
	    Moreover, the Riemannian gradient of distance squared for any fixed point $Q\in \mathcal{SP}(n)$ is 
	    $
		    	\Big(\nabla \text{d}^2_{\text{AI}}(P,Q)\Big)_P  = -2\,\sum_{i=1}^d \left\langle \log (P^{-1/2}\star Q),S_i\right\rangle_F\,  E_i^{\text{AI}}(P)= -2 \,\text{Log}^{\text{AI}}_P(Q).
		$
	\end{enumerate}
	\label{D_LE_AI}
	\end{proposition}
   
    The OU processes on $\mathcal{SP}(n)$ are obtained as the limit $\delta_t$ tends to zero of the exponential adapted Euler-Maruyama method:
	\begin{equation}
		X_{t+\delta_t} = \text{Exp}_{X_t}\Bigg\{-\frac{\theta}{2} \nabla_{X_t}\{d^2(X_t,M)\}\,\delta_t + \sum_{j = 1}^d (B^{(j)}_{t+\delta_t} - B^{(j)}_t)\, \sigma \, E_j(X_t)\Bigg\} \,\,\text{ for } \delta_t >0.
		\label{eq:OU via exp}
	\end{equation}
    The selection of the basis fields plays an important role as they represent the horizontal lift of $X_t$ locally when using the piece-wise approximation method in equation~\eqref{eq:OU via exp}.

    For the LE metric, we define a global isometric diffeomorphism $\mathfrak{h}$, that allows us to constructively identify $\mathcal{SP}(n)$ with $\mathbb{R}^d$. Additionally, it characterizes the OU class of processes on $\mathcal{SP}(n)$ equipped with the LE metric as the image of the standard OU process on  $\mathbb{R}^d$ under $\mathfrak{h}$, see Theorem \ref{SDE LE} in which the proof is presented in Subsection~\ref{supp:SDEs with the LE metric} of the Supplementary Material. In turn, this permits establishing existence, uniqueness and non-explosion of the OU class for the LE metric. 
    
    We define $\mathfrak{h} = (\mathfrak{h}_j) : \mathcal{SP}(n) \rightarrow \mathbb{R}^d$ with   $e = \sum_{j=1}^d \epsilon_j \, e_j \in \mathbb{R}^d$ and $P \in \mathcal{SP}(n)$ as follows:
    \begin{equation}
		 \mathfrak{h}_j(P) = \langle \log P, S_j\rangle_F  \,\,\,(1 \leq j \leq d) \,\,\,\, \text{ and } \,\,\,\,\mathfrak{h}^{-1}(e) = \exp\big(\sum_{j=1}^d \epsilon_j\, S_j\big).
		\label{eq:mathfrak(h)}
	\end{equation}


	\begin{theorem}
    		Suppose the process $X_t$ is the solution of the following SDE on $\mathcal{SP}(n)$ endowed with the LE metric,  for $t \in [0, \tau)$ with  a $\mathfrak{F}_*$-stopping time $\tau$:
    		\begin{equation}\label{eq:solve SDE in LE case}
    			dX_t = A(t, X_t) \, dt + F_{X_t}\big(b(X_t) \, dB_t\big) \hspace{2cm}(X_0 = P),
    		\end{equation}
    		where $A$ assigns smoothly for each $t \in [0,\tau)$ a smooth vector field $A(t,\cdot)$ on $\mathcal{SP}(n)$ and some smooth function $b : \mathcal{SP}(n) \rightarrow \mathbb{R}^{d \times d}$. Moreover, $B_t$ is $\mathbb{R}^d$-valued Brownian motion and the function $F$ is  defined in Equation~\eqref{eq:OU F} associated with the basis $\mathfrak{B}_d^{\text{LE}}$. Then the problem of solving the SDE~\eqref{eq:solve SDE in LE case} on $\mathcal{SP}(n)$ is the same as solving the following SDE on $\mathbb{R}^d$ :
    		\begin{equation}
    			dx_t = a(t,x_t)\,dt + \tilde{b}(x_t)\,dB_t \hspace{2.1cm} (x_0 = p),
    		\label{eq:solve SDE in LE case2}
    		\end{equation}
    		Here, $p = \mathfrak{h}(P)$, $x_t = \mathfrak{h}(X_t)$ hold for all $t \in [0,\tau)$ and smooth function $\tilde{b}$ is given by $\tilde{b} = b\circ \mathfrak{h}^{-1}$. In addition,  smooth function $a = \big(a^{(j)}\big)$  is given by  
    		\[a^{(j)} : [0,\tau) \times \mathbb{R}^d \rightarrow \mathbb{R}, \hspace{1cm} \,(t,x_t) \mapsto \big\langle D_{X_t}\log. A(t,X_t) ,S_j\big\rangle_F \,\,\,\,\,\text{ for all } 1 \leq j \leq d.\]
    	\label{SDE LE}
    	\end{theorem}

	Solutions of \eqref{eq:solve SDE in LE case} and \eqref{eq:solve SDE in LE case2} are in one-to-one correspondence. Therefore, the conditions for existence and uniqueness of the solution for the SDE~\eqref{eq:solve SDE in LE case} depend directly on the requirements that the drift and diffusivity of the SDE~\eqref{eq:solve SDE in LE case2} satisfy on the Euclidean space, e.g. continuity and local Lipschitzness. Indeed, equating the drift and diffusivity of the SDE~\eqref{eq:solve SDE in LE case} with our OU process, the SDE~\eqref{eq:solve SDE in LE case2} turns out to be a standard OU process on $\mathbb{R}^d$. 
	Thus, most favourable properties that the OU process has on the Euclidean space will carry over to $\mathcal{SP}(n)$, such as existence and uniqueness of the solution and ergodicity. The transition probability density is explicitly available up to the Jacobian term involving the derivative of the matrix exponential. 
	
	\begin{corollary}
	 SDE~\eqref{eq:SDE of OU} has a unique solution and gives rise to an ergodic diffusion process on $\mathcal{SP}(n)$ in the LE case.
	 \label{solution of OU (LE)}
	\end{corollary}

    We conclude this section by establishing equivalent results for the AI case in a non-constructive manner. Although  there is no simple diffeomorphism corresponding to $\mathfrak{h}$, $\mathcal{SP}(n)$ equipped with the AI metric is parallelizable. Therefore, equivalent results can be obtained for the AI case:
	
	\begin{proposition}
	The existence and uniqueness theorem in \citep[Theorem 2E, Page 121]{elworthy1982stochastic} is applicable to the OU process on $\mathcal{SP}(n)$ equipped with the AI metric. Moreover, this diffusion process is also non-explosive, see \citep[Corollary 6.1, Page 131]{elworthy1982stochastic}.	
	\end{proposition}
\section{Bayesian parameter estimation } \label{Bayesian parameter estimation}

	We now focus on the Bayesian estimation of the parameters of the OU difussion processes on $\mathcal{SP}(n)$ when observations are collected at low frequency.  We adopt the data augmentation MCMC computational strategy introduced by  G. O. Roberts \&  O. Stramer \citep{roberts2001inference} which requires data imputation through sampling from a diffusion bridge.  We need to build diffusion bridge samplers that operate on $\mathcal{SP}(n)$ which, unlike the Euclidean case, have not been studied before. A common approach used in manifolds is to use embedding or local charts followed by an appropriate Euclidean method; see for example, \citep{ball2008brownian,sommer1705bridge,staneva2017learning}, but this strategy is unsuitable when transitioning between charts is required and charts can be cumbersome to work with.  We therefore develop a diffusion bridge sampler exploiting the exponential map and adopting an intrinsic viewpoint. In fact, by using Corollary~\ref{solution of OU (LE)} we can translate any existing method for the OU process from the Euclidean to the LE setting, so in the remainder of this section we will focus only on the AI metric where no such result exists.  To deal with the data augmentation problem it is either assumed that $X_t$ has constant diffusivity, or that $X_t$ is transformed to a process of constant diffusivity, or existence of a process that is absolutely continuous to $X_t$ and its corresponding transition probability density need to be derived.  In the case of the AI metric, at first glance the SDE seems to have constant diffusivity. However, due to the presence of curvature, the diffusivity part does depend on the position of the process $X_t$, hence straightforward algorithms from literature are not applicable. In fact, by looking at the local coordinate expression of the Riemannian Brownian motion on $\mathcal{SP}(n)$ equipped with the AI metric (substituting $G^{-1}$ in equation~\eqref{eq:AI metric in matrix form} and Christoffel symbols given in Lemma~\ref{AI Chris symbol}  into equation~\eqref{eq:BM in local}),  the dependence of the diffusivity on $X_t$ is evident as is the complexity of the resulting expressions. Furthermore, attempting to sample from this Euclidean SDE using the standard Euler-Maruyama method will lead to symmetric but non-positive definite matrices.

	On the Euclidean space, B. Delyon \& Y. Hu \citep{delyon2006simulation} and M. Schauer \& F. Van Der Meulen et al. \citep{schauer2017guided} suggest adding an extra drift term which guides the SDE solution toward the correct terminal point and leaves its law absolutely continuous with respect to the law of the original conditional diffusion process; the Radon-Nikodym derivative is available explicitly. This results in an easier simulation, much better MCMC mixing rate of convergence and no difficulty of computing acceptance probability when updating the proposal bridge. The additional drift is the gradient of the logarithm of the transition density of an auxiliary process $\tilde{X}_t$ which must have explicitly available transition probability density. In the manifold setting, we are aware of one attempt to use the above approach to sample diffusion bridges through local coordinates by S. Sommer et al. \citep{sommer1705bridge}, but using the exponential map in this context is new. Motivated by these ideas, we construct a methodology that allow us to sample a diffusion bridge on $\mathcal{SP}(n)$ equipped with the AI metric using a guided proposal process.

	We need to choose a proposal diffusion process, which has both an explicit transition probability density and an analytically tractable gradient of the log transition probability density due to the requirement in otaining the additional drift in the guided proposals SDE.  The Wishart process is inappropriate because it does not have a closed form for its Riemannian gradient so we choose a diffusion process $\tilde{X}_t$ on $\mathcal{SP}(n)$ with transition probability the Riemannian Gaussian distribution \citep{said2017riemannian} given as  
	\begin{equation*}
		p(s,\tilde{X}_s; t, \tilde{X}_t) = \frac{1}{K_n(\sigma)} \, \exp \left(-\frac{\text{d}_{\text{AI}}^2(\tilde{X}_s,\tilde{X}_t)}{2\, (t-s)\,\sigma^2}\right),
		\label{eq:proposal p(s,x;t,y)}
	\end{equation*}
	where $K_n(\sigma)$ is the normalising constant that depends only on $\sigma$ and $n$.  We emphasize that explicit availability of the guided proposal SDE is not actually a pre-requisite for the guided proposals algorithm. Indeed, the proposal Markov process $\tilde{X}_t$ exists but its SDE form is not explicitly available \citep{baxendale1976measures}.

	Consider the  OU process $X_t$  on $\mathcal{SP}(n)$ equipped with the AI metric  with its law $\mathbb{P}_t$ and assume that sampling from the target diffusion bridge $X_t^{*} = \{X_t, 0 \leq t \leq T \,| \, X_0 = U, X_T = V\}$ with its corresponding law $\mathbb{P}^{*}_t$ is required. We introduce the guided proposal $X_t^{\diamond}$ which is the solution of the following SDE with its law $\mathbb{P}^{\diamond}_t$:
	\begin{equation}
		dX^{\diamond}_t = \left(\theta \, \text{Log}^{\text{AI}}_{X^{\diamond}_t}\, M + \frac{\text{Log}_{X^{\diamond}_t}^{\text{AI}}\,V}{T-t} \right) \,\, dt +F(X^{\diamond}_t)( \sigma \, dB_t) \hspace{1.5 cm } (X^{\diamond}_0 = U).
		\label{eq:AI guided proposal}
	\end{equation}
	We then show that $\mathbb{P}_t^*$ and $\mathbb{P}_t^{\diamond}$ are equivalent up to time $T$  with the aid of Lemma~\ref{AI Chris symbol} and Lemma~\ref{hess AI} in Theorem~\ref{AI equivalent laws1}; proofs are shown in Subsection~\ref{supp:absolute continuity} of the Supplementary Material.
	
	\begin{lemma}
		The Christoffel symbols at any $P \in \mathcal{SP}(n)$ with respect to the basis $\mathfrak{B}_d^{\text{AI}} $  do not depend on $P$ and are given by
		$
		    \Gamma_{ij}^k(P) = \left\langle -(S_iS_j +S_jS_i)\,,\,S_k\right\rangle_F / 2 
		$.
		\label{AI Chris symbol} 
	\end{lemma}
    \begin{lemma}The Laplace-Beltrami operator $\Delta_{\mathcal{SP}(n)}$ and the Hessian to the squared Riemannian distance  with respect to $\mathfrak{B}_d^{\text{AI}}$ equal to $2\, I_d$, that is 
            \[(\text{Hess}f)_P\big(E_i^{\text{AI}}, E_j^{\text{AI}}\big) = 2\,\delta_{ij} \hspace{0.3cm}\& \hspace{0.3cm} \Delta_{\mathcal{SP}(n)}f\big(E_i^{\text{AI}},E_j^{\text{AI}}\big) = 2\, \delta_{ij} \hspace{1cm} \text{ for all } 1\leq i,j \leq d.\]
            where $f(P) = d^2_{\text{AI}}(P, Q)$ and $\delta_{ij}$ is the Kronecker delta. 
            \label{hess AI}
    \end{lemma}
    
	\begin{theorem} 
	    \begingroup
	    \allowdisplaybreaks
		For $t \in [0,T)$ the laws $\mathbb{P}_t$,$\mathbb{P}_t^{\diamond}$ and $\mathbb{P}_t^{*}$ are absolutely continuous. Let $p(t,X_t;T,V)$ be the true (unknown) transition density of moving from $X_t$ at time $t$ to $V$ at time $T$  and let $X^{\diamond}_{[0:t]}$ be the path of $X_t^{\diamond}$ from time $0$ to $t$. Then
\begin{align}
    			\frac{d\mathbb{P}_t}{d\mathbb{P}_t^{\diamond}}\big(X^{\diamond}_{[0:t]}\big) &= \frac{\exp\big(f(X^{\diamond}_0;\sigma^2)\big)}{\exp\big(f(X^{\diamond}_t;\sigma^2)\big)}  \exp\Big\{\Phi(t,X^{\diamond}_{[0:t]})+ \phi\big(t,X^{\diamond}_{[0:t]}\big)\Big\},
    			\label{eq:derivative1}\\
    			\frac{d\mathbb{P}_t^{*}}{d\mathbb{P}_t^{\diamond}}\big(X^{\diamond}_{[0:t]}\big) &= \frac{p(t,X_t^{\diamond};T,V)}{\exp\big(f(X^{\diamond}_t;\sigma^2)\big)} \, \frac{\exp\big(f(X^{\diamond}_0;\sigma^2)\big)}{p(0,U;T,V)}  \exp\Big\{\Phi\big(t,X^{\diamond}_{[0:t]}\big)+ \phi\big(t,X^{\diamond}_{[0:t]}\big)\Big\},
    			\label{eq:derivative2}
    		\end{align} 
    		where  the functions $f,\, \phi$ and $\Phi$ are defined by
    		\begin{align}
    			&f(X^{\diamond}_t;\sigma^2) =  -\frac{d^2_{\text{AI}}(X^{\diamond}_t,V)}{2\sigma^2(T-t)} = -\frac{\left|\left| \log\left((X^{\diamond}_t)^{-1/2} V (X^{\diamond}_t)^{-1/2}\right)\right|\right|_F^2}{2\sigma^2(T-t)}, \label{eq:AI f}\\
    			 &\phi\big(t, X^{\diamond}_{[0:t]}\big)=  \sum_{i,j=1}^d \int_0^t\frac{ \big(\zeta^i_j(X_s^{\diamond};\Theta)\big)^2\, }{ 2(T-s)} \,ds, \label{eq:AI phi}\\
    			&\Phi \big( t, X^{\diamond}_{[0:t]}\big) = \int\limits_0^t \Bigg( \frac{\theta \,g_{X_s}^{\text{AI}}\left(\text{Log}^{\text{AI}}_{X^{\diamond}_s}M ,\text{Log}^{\text{AI}}_{X_s^{\diamond}}V\right) }{\sigma^2 (T-s)} \nonumber\\
    			 & + \sum_{i,j,r=1}^d  \,  \frac{g_{X_s^{\diamond}}^{\text{AI}} \Big(  \zeta_j^i(X_s^{\diamond};\Theta)\big(\sum_{l=1}^d \zeta_l^i(X_s^{\diamond};\Theta) \Gamma_{jl}^r + \big(E_j^{\text{AI}}\zeta_r^i(\cdot;\Theta)\big)_{X_s^{\diamond}}\big) E_r^{\text{AI}}(X_s^{\diamond})\,,\,\text{Log}^{\text{AI}}_{X_s^{\diamond}}V \Big)}{ 2(T -s)} \Bigg) ds,
    			\label{eq:AI Phi} 
    		\end{align}	
    		with $\Gamma_{jl}^r$ given in Lemma~\ref{AI Chris symbol}. The functions $ \zeta(X_s^{\diamond};\Theta) = \big(\zeta_j^i(X_s^{\diamond};\Theta)\big)$ are the coefficients with respect to the basis $\mathfrak{B}^{\text{AI}}_d$ in the local expression of the horizontal lift, see Proposition~\ref{AI_LE_stochastic development}.  
		\endgroup
		\label{AI equivalent laws1}
	\end{theorem}
     Note that all terms dependent on $\zeta$ are not independent of the model parameters $\Theta$. When the parameters change, the diffusion path $X^{\diamond}_t$ also changes, even when the Brownian motion driving this process remains unchanged. In other words, there are implicit dependencies between $\zeta$ and the model parameters $\Theta$. Hence, for clarification purposes, the function $\zeta$ is written as $\zeta(X_s^{\diamond};\Theta)$. In fact, we should understand $X^{\diamond}_t$ as a function that depends on the driving Brownian motion and model parameters.  The Radon-Nikodym derivatives in equations~\eqref{eq:derivative1}--\eqref{eq:derivative2} can not be computed explicitly due to the presence of $\zeta$. However, we can approximate $\phi$ and $\Phi$ in Theorem~\ref{AI equivalent laws1} based on the approximation of $\zeta$, see Corollary~\ref{supp:AI approximation of zeta}  of the Supplementary Material. As a result, Remark~\ref{AI equivalent law1 approx} is a crucial step to make our proposed algorithm practicable; detailed calculation of the approximation of $\phi$ and $\Phi$ are presented in Corollary~\ref{supp:AI equivalent law1 approx}  of the Supplementary Material.
    
    \begin{remark}

\begingroup
	    \allowdisplaybreaks
	    Suppose that we have an $\mathcal{SP}(n)$-valued path $\{X_{t_k}^{\diamond} = y^{\diamond}_{t_k}\}_{k=0}^{m+1}$ of the OU process $X_t$ when equipping $\mathcal{SP}(n)$ with the AI metric, in which it is simulated from the exponential adapted Euler-Maruyama method, see Equation~\eqref{eq:OU via exp}, where $\max\{t_{k+1}-t_k\}_{k=0}^m$ is sufficiently small and $0 = t_0 < \ldots <t_{m+1} = t$. Then the functions $\phi$ in Equation~\eqref{eq:AI phi} and $\Phi$ in Equation~\eqref{eq:AI Phi} can be approximated as follows:
		\begin{align}
			\phi(t, X^{\diamond}_{[0:t]}) &\approx  \frac{d}{2} \log \frac{T-t}{T}, \label{eq:AI phi approx}\\
			\Phi(t,X^{\diamond}_{[0:t]}) &\approx \sum_{k=0}^m \frac{t_{k+1}- t_k}{T- t_k} \,\Bigg\{  \frac{\theta\,\left\langle \log\big((y_{t_k}^{\diamond})^{-1/2} \star M \big)\,,\, \log\big((y_{t_k}^{\diamond})^{-1/2} \star V  \big)\right\rangle_F}{\sigma^2} \nonumber  \\
    	    & \hspace{1cm} +     \frac{\left\langle  \Gamma\,,\, \log\big((y_{t_k}^{\diamond})^{-1/2} \star V\big) \right\rangle_F}{2}\Bigg\},  \label{eq:AI Phi approx}
		\end{align}	
		with $\Gamma = \sum_{i,r=1}^d\Gamma_{ii}^r  S_r $ and the Christoffel symbols $\Gamma_{ii}^r$ are given in Lemma~\ref{AI Chris symbol}.
		\endgroup
		\label{AI equivalent law1 approx}        
    \end{remark}
	
	We expect that taking the limit $t \uparrow T$ in equation~\eqref{eq:derivative2} will follow along similar lines to those in \cite{delyon2006simulation} and \cite{schauer2017guided} so that the following holds
	\begin{equation*}
		\frac{d\mathbb{P}_T^{*}}{d\mathbb{P}_T^{\diamond}}(X^{\diamond}_{[0:T]}) =  \frac{\mathcal{H}_{n,T}}{p(0,U;T,V)} \, \exp\left\{f(X^{\diamond}_0;\sigma^2) + \Phi(T,X^{\diamond}_{[0:T]})   - \frac{d}{2}\log \sigma^2\right\}\,,
		\label{eq: AI equivalent laws2}
	\end{equation*}
	where  $\mathcal{H}_{n,T}$ is a fixed scalar that depends only on the dimension $n$ and the terminal time $T$. 
    While we do not present the full argument, we do, however, provide a careful numerical validation in Section \ref{Simulation study}.
	
	Suppose we have discretely observed data $\mathfrak{D}= \{X_{t_j} = y_j\}_{j=0}^N$ at observation times $t_0 = 0 < t_1 < \cdots < t_N = T$, where the diffusion process $X_t$ is the OU process $X_t$  on $\mathcal{SP}(n)$. We aim to sample from the posterior distribution of $\Theta = \{\theta,\,M, \sigma^2\}$. We set $\mu = \sum_{i=1}^d\mu^{(i)}\,e_i = \mathfrak{h}(M)$, as defined in equation~\eqref{eq:mathfrak(h)}, and the prior distributions of $\{\theta, \mu, \sigma^2\}$ as $\pi_0^{\theta},\, \pi_0^{\mu}$ and $\pi_0^{\sigma}$ respectively.  The key step involves first imputing suitable $m_j - 1$ data points between the $j\text{th}$ consecutive observations in a way that they are independent of the diffusivity, and then employing the exponential adapted Euler--Maruyama method for the likelihood approximation. We choose to use random walk symmetric proposal distributions with suitable choices of step size  $q(\tilde{\theta}|\theta),\, q(\tilde{\mu}|\mu)$ and $q(\tilde{\sigma}^2|\sigma^2)$ for $\{\theta,\mu,\sigma^2\}$ respectively. 
	
	\begin{algo}[Guided proposals on $\mathcal{SP}(n)$] \hspace{5cm}
		\begin{enumerate}
			\item (Iteration $k =0$). Choose starting values for $\Theta$ and sample standard Brownian motions $W_j$, independently for $1 \leq j \leq N$,  each covering the time interval $t_{j} - t_{j-1}$, and set $B_j^{(0)} = W_j$.
			\item (Iteration $k \geq 1$).
			\begin{enumerate}
				\item[(a)] Update $B_j$ independently $(1 \leq j \leq N)$: sample the proposal  $\tilde{W}_j$ and obtain $\tilde{Y}_{[t_{j-1},t_j]}$ from $\{\tilde{W}_j,\theta_{k-1},M_{k-1},\sigma^2_{k-1}, \mathfrak{D}\}$  and  $Y_{[t_{j-1},t_j]}$ from $\{B_j^{(k-1)},\theta_{k-1},M_{k-1},\sigma^2_{k-1}, \mathfrak{D}\}$ using equation~\eqref{eq:OU via exp} to approximately solve the SDE~\eqref{eq:AI guided proposal}; then accept $\tilde{W}_j$ with probability 
					\[\alpha^{(B)} = \min\left\{1,\exp \left[ \Phi(t_j-t_{j-1},\tilde{Y}_{[t_{j-1},t_{j}]})  - \Phi(t_j-t_{j-1},Y_{[t_{j-1},t_{j}]}) \right] \right\}.\]
				
				\item[(b)] Sample $\tilde{\sigma}^2 $ from  $q(\sigma^2|\sigma_{k-1}^2)$ and obtain $\tilde{Y}_{[t_{j-1},t_j]}$ from $\{B_j^{(k)},\theta_{k-1},M_{k-1},\tilde{\sigma}^2, \mathfrak{D}\}$  and  $Y_{[t_{j-1},t_j]}$ from $\{B_j^{(k)},\theta_{k-1},M_{k-1},\sigma^2_{k-1}, \mathfrak{D}\}$ using equation~\eqref{eq:OU via exp} to approximately solve the SDE~\eqref{eq:AI guided proposal}; then accept $\tilde{\sigma}^2$ with  probability 
					\begin{eqnarray*}
					    \alpha^{(\sigma)} &=& \min\Bigg\{1,\\
					    && \frac{\pi_0^{\sigma}(\tilde{\sigma}^2) \cdot \prod_{j=1}^N\exp \left\{\Phi(t_j-t_{j-1},\tilde{Y}_{[t_{j-1},t_{j}]})+f(\tilde{Y}_{t_{j-1}},\tilde{\sigma}^2) - \frac{d}{2}\tilde{\sigma}^2\right\}}{\pi_0^{\sigma}(\sigma^2_{k-1}) \cdot \prod_{j=1}^N\exp \left\{\Phi(t_j-t_{j-1},Y_{[t_{j-1},t_{j}]})+f(Y_{t_{j-1}},\sigma_{k-1}^2) - \frac{d}{2}\sigma^2_{k-1}\right\}}\Bigg\}.
					\end{eqnarray*}
				\item[(c)] Update $\mu$  and $M$ : sample $\tilde{\mu} $ from  $q(\mu|\mu_{k-1})$, compute the corresponding $\tilde{M} = \mathfrak{h}^{-1}(\tilde{\mu})$ and accept  $\tilde{\mu}$, $\tilde{M}$ with  probability
			    \begin{eqnarray*}
			    \alpha^{(M)} =\min\Bigg\{1, \frac{\pi_0^{\mu}(\tilde{\mu})}{\pi_0^{\mu}(\mu_{k-1})} \, &&\exp\bigg\{  \sum_{j=1}^N\Big[ \Phi\left(t_j-t_{j-1},\tilde{Y}_{[t_{j-1},t_{j}]}\right) \\
			    &&- \Phi\left(t_j-t_{j-1},Y_{[t_{j-1},t_{j}]}\right)\Big]
			    \bigg\}\Bigg\},
			    \end{eqnarray*}
					where $\tilde{Y}_{[t_{j-1},t_j]}$, $Y_{[t_{j-1},t_j]}$ are from $\{B_j^{(k)},\theta_{k-1},\tilde{M},\sigma^2_k, \mathfrak{D}\}$,   $\{B_j^{(k)},\theta_{k-1},M_{k-1},\sigma^2_{k}, \mathfrak{D}\}$ respectively using equation~\eqref{eq:OU via exp} to approximately solve the SDE~\eqref{eq:AI guided proposal}.
				\item[(d)] Update $\theta$ similarly as $\mu$.
			\end{enumerate} 
		\end{enumerate}
		\label{guided proposal algorithm}
	\end{algo}
	
	\begin{remark}[Time change]
		Since $\Phi$ in equation~\eqref{eq:AI Phi} explodes as $t \uparrow T$, M. Schauer \& F. Van Der Meulen et al. \citep{schauer2017guided} suggests time change and scaling to reduce the required number of imputed data points. Scaling will not be as effective here as in the univariate setting because it can only fit one of the directions involved,
		so the effect is less pronounced than in the univariate setting and we expect further lessening as dimension increases. Nonetheless, in this work, we adopt one time change function from \citep{schauer2017guided} which maps $s$ to $s\, (2-s/T)$.
		\label{AI time change}
	\end{remark}
	
\section{Simulation study on $\mathcal{SP}(2)$}\label{Simulation study}
    \begin{table}[H]
        \centering
        \begin{tabular}{c| c| c| c|  c }
        		\multicolumn{5}{c}{{\textbf{ p-values of K-S tests}}}\\
        		\multicolumn{2}{c}{}&\multicolumn{3}{c}{{$\epsilon$}}\\
        		\multicolumn{1}{c}{}&   & $0.2$ &$0.1$& $0.05$  \\
        		\hline
        		\multicolumn{1}{c}{\multirow{2}{5em}{Determinant}} & $m=1000$   &0.00412 & 0.0329 & 0.288\\
        		\multicolumn{1}{c}{}& $m=2000$ &0.00283 & 0.0612 & 0.370\\
        		        		\hline
        		\multicolumn{1}{c}{\multirow{2}{5em}{ Trace}} & $m=1000$  &0.0112 & 0.0293& 0.108\\ 
        		\multicolumn{1}{c}{}&$m=2000$ &0.00123 & 0.0378 & 0.288\\

        \end{tabular}
        \includegraphics[width = 0.7\textwidth, height = 4.5cm]{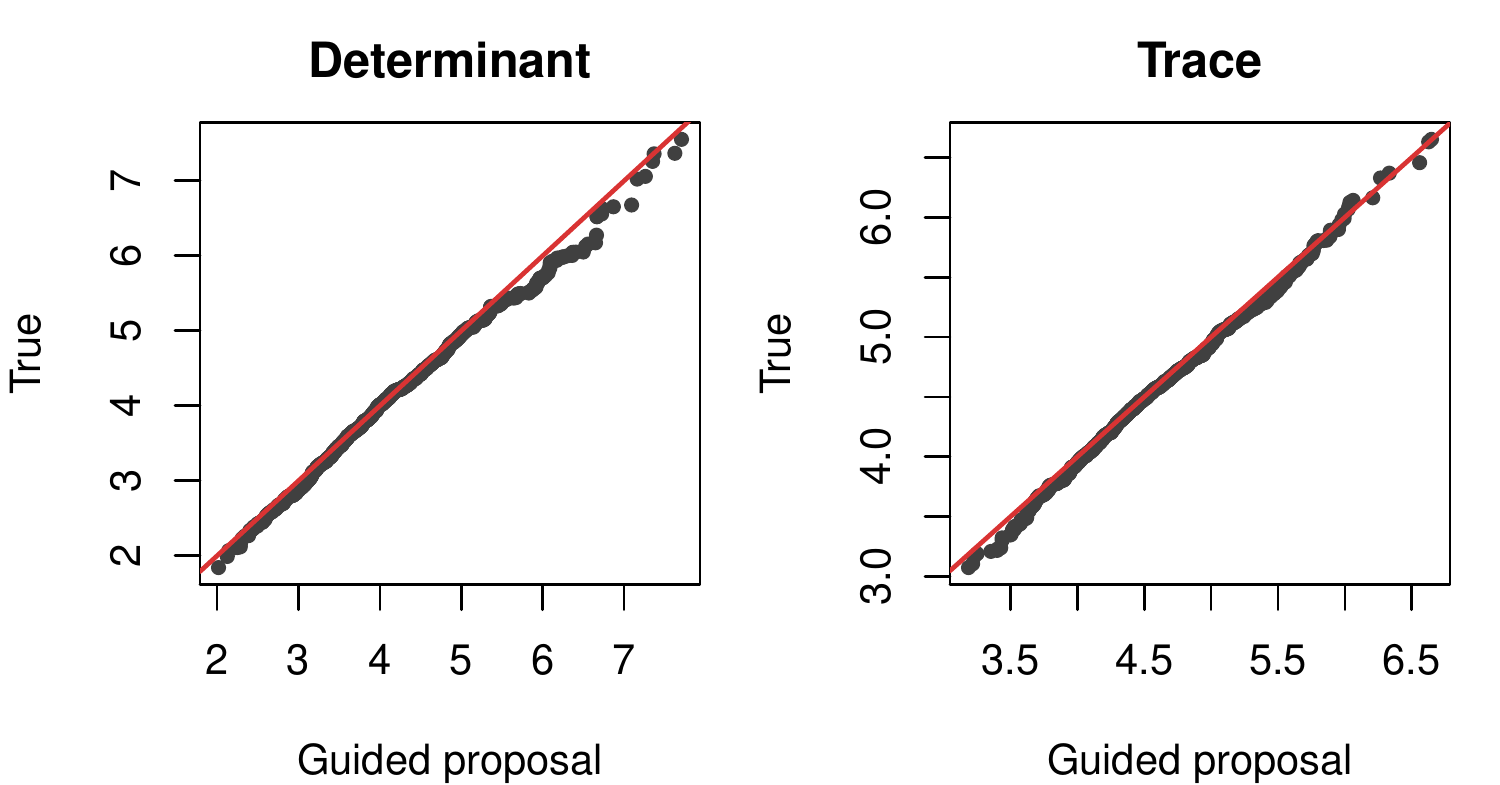}
        \caption{Comparison in distribution at $t =T/2$ of the true bridges and the guided proposal bridges, given that $U$ and $V$ are chosen as in Case 1. Top table: p-values from the Kolmogorov–Smirnov (K-S) tests, where $\epsilon$ and $m$ are varied. Bottom figures: the Q-Q plots for determinant and trace  when $m = 2000$ and $\epsilon = 0.05$.  }
        \label{fig:QQ plot and KS tests}
    \end{table}
    
\subsection{Brownian bridges}\label{Brownian bridges}

     We perform a simulation exercise to illustrate our proposed bridge sampling of Algorithm~\ref{guided proposal algorithm} and to compare the performance for the Euclidean, LE and AI metrics. The simulation scenario involves sampling a standard Brownian bridge $W_t$ conditioned on $\{W_0 = U, W_T = V\}$, $T=0.1$, $n=2$, in two cases: (i) with $U$ and  $V$ lying far away from the boundary,  
     $U = \begin{pmatrix}	2&1\\ 1&2 \end{pmatrix}$, and  $V = \begin{pmatrix} 3&1\\ 1&2 \end{pmatrix}$
     and (ii) with $U$ and $V$ lying close to the boundary,   $U = \begin{pmatrix}
    			2&1.999\\1.999&2 \end{pmatrix}$,  $V = \begin{pmatrix}
    			3&2.435\\ 2.435&2 \end{pmatrix}$.
    For the Euclidean metric we simply embed $\mathcal{SP}(n)$ in $\mathcal{S}(n)$ endowed with the Frobenius inner product and thus the Brownian motion in this case is the solution of the SDE $d\nu(W_t) = dB_t \,\,(W_0 = U)$, where $B_t$ is the standard  Brownian motion on $\mathbb{R}^d$ and $\nu(W_t) \in \mathbb{R}^d$ contains only independent entries of $W_t$. The LE and AI metrics are based on  Theorem~\ref{SDE LE} and  Algorithm \ref{guided proposal algorithm} respectively.



	\begin{figure}[t]
		\centering
		\includegraphics[width = \textwidth]{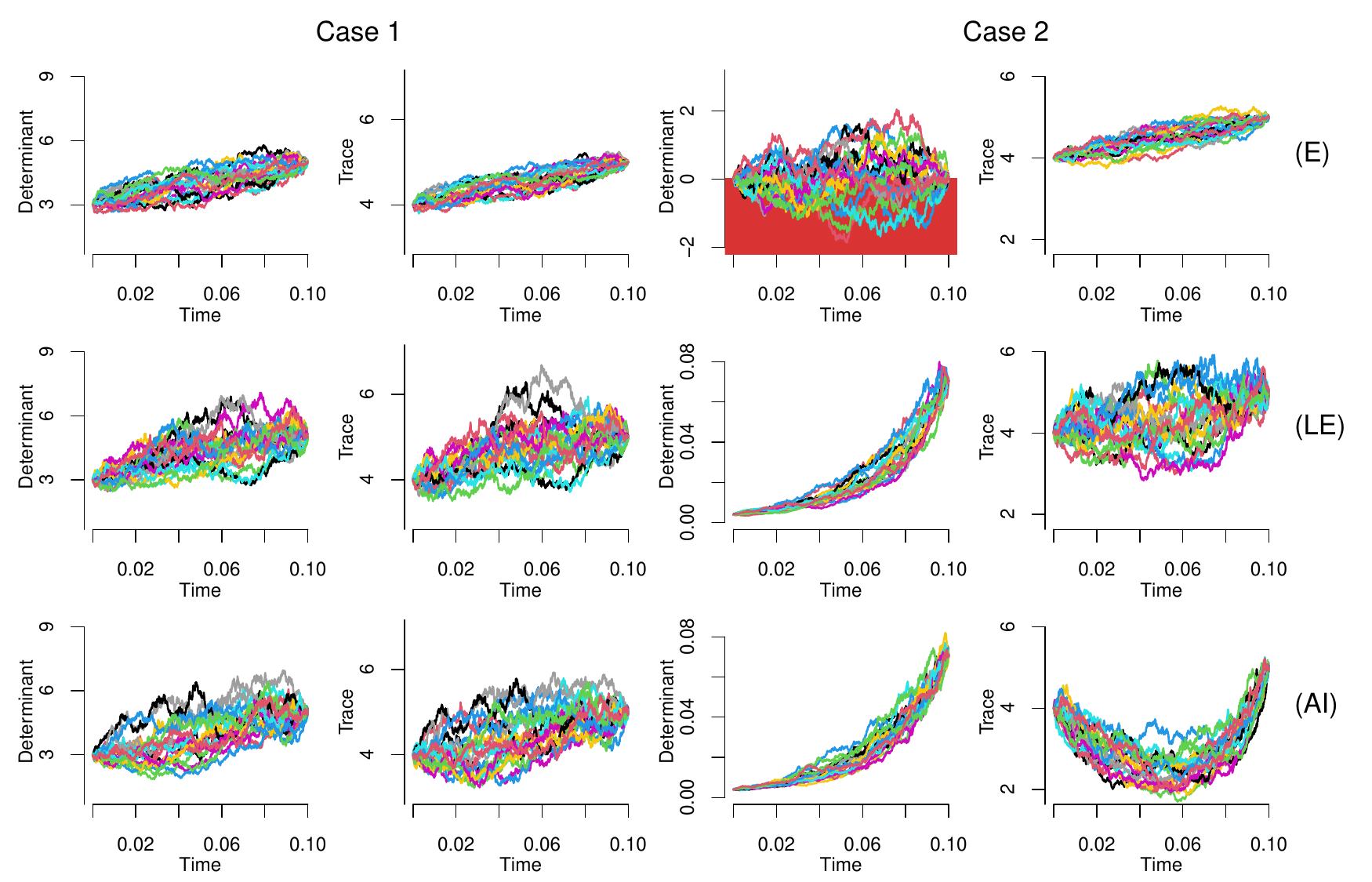}
		\caption{Time series for determinant and trace of 20 simulated Brownian bridges on $\mathcal{SP}(2)$ endowed with three different metrics (Euclidean: top row, Log Euclidean: middle row, Affine Invariant: bottom row) in the $2$ cases (i) and (ii). Red background area shows failure to be positive definite. }
		\label{fig:compare_BB}
	\end{figure}
	First we evaluate how Algorithm~\ref{guided proposal algorithm} performs compared with the naive simulation approach. For case (i) we obtain samples by forward-simulating the guided proposal in equation~\eqref{eq:AI guided proposal} and accepting with probability $\alpha^{(B)}$ in Algorithm~\ref{guided proposal algorithm}. We then compare these bridges with the so-called true bridges,  which are generated from the naive simulation approach by forward-solving the Riemannian Brownian motion $X_t$ and picking only those paths $X_t$  that satisfy $\text{d}_{\text{AI}} (X_{0.1},V) < \epsilon $ for some $\epsilon > 0$. For different values of $\epsilon$ and number of imputed points $m$, we collect $1000$ sampled bridges and carry out a Kolmogorov–-Smirnov (K-S) test to compare the distribution of the true and approximated bridges at $t = T/2 = 0.05$; The Q-Q plots and the K-S p-values shown in Table~\ref{fig:QQ plot and KS tests} indicate that the values $m=2000, \epsilon = 0.05$ provide a good approximation and as $\epsilon \downarrow 0$ the two distributions get closer to each other. Note that for case (ii) we had a difficulty in sampling true bridges because $V$ is too close to the boundary of $\mathcal{SP}(2)$ and points on the boundary are at infinite distance to covariance matrices. This difficulty escalates as dimension $n$ increases. Beside high dependence with the diffusivity, the naive approach of simulating bridges has very low acceptance rate when the conditioned observation is close to being non-positive definite, thus it is clearly not a good approach to use in the data-augmentation algorithm.

	Next, we illustrate the problems that arise when neglecting the geometric structure of $SP(n)$ as stated in Section~\ref{The importance of Riemannian geometry to SP(n)}. Figure~\ref{fig:compare_BB} depicts traces and determinants  of $20$ simulated bridges with each metric.  First note that the Euclidean metric samples do not achieve positive definiteness and have a clear presence of the swelling effect.  The LE and AI metrics sample points that lie in $\mathcal{SP}(n)$ while the determinants behave reasonably with respect to the conditional points. In fact, one can show following the lines of the 
	case of geometric means in \citep{arsigny2007geometric} that the 
	distribution of the determinant is the same for the LE and the AI metrics.  The AI metric leads to less anisotropy that becomes more noticeable when the conditional points have eigenvalues close to zero, see plots about the trace in Fig.~\ref{fig:compare_BB}, indicating that the selection between LE and AI depends on the desirable property that one wishes to achieve. 

\subsection{Parameter estimation for the Affine-Invariant metric}\label{Parameter estimation for the Affine-Invariant metric}
\subsubsection{Dimension  $n = 2$}

	\begin{figure}[H]
		\centering
		\includegraphics[width = \textwidth, height = 12.6cm]{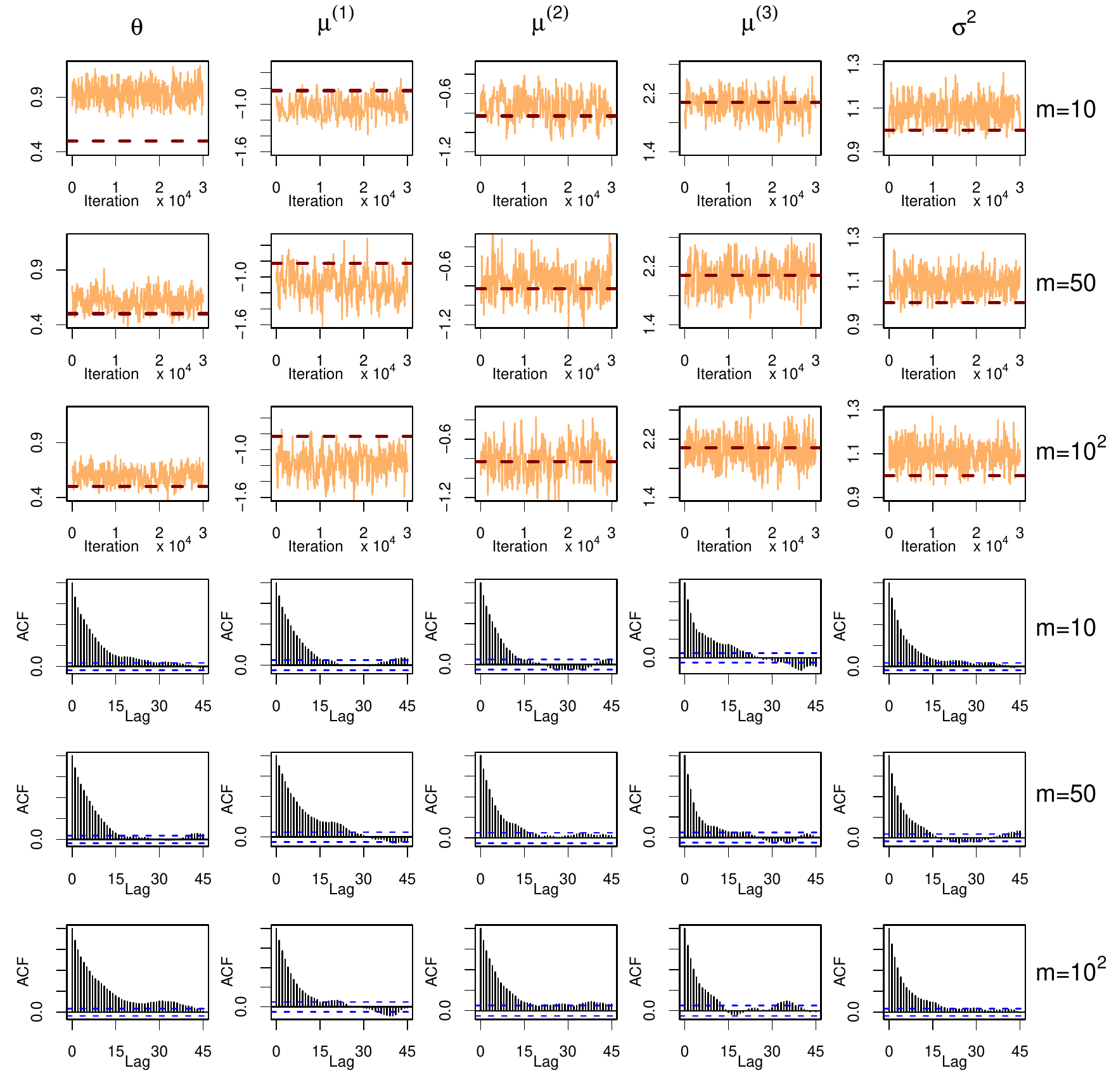}
		\caption{(Simulation study). $m$ is varied over $10,\, 50$ and $100$. Top three rows: Traceplots from $4 \times 10^3$ iterates after discarding $10^3$ iterations of burn-in, where true values are indicated with the red dashed lines. Bottom three rows: ACF plots based on these MCMC chains.}
		\label{fig:AI_convergence_simulation}
	\end{figure}
	
    We simulate $10^6 +1$ equidistant time points $X_t$ of the OU process in the case of the AI metric on $[0,100]$ using equation~\eqref{eq:OU via exp} with model parameters
	$\theta = 0.5$, $M = \left(\begin{array}{cc} 	1 & 0.9 \\0.9& 1
	\end{array}\right)$, $\sigma^2 = 1$ and $X_0 =  I_2$ and take sub-samples at time points $\{0, 0.2, \ldots, 100$\}. We apply the Algorithm \ref{guided proposal algorithm} with time change assuming the prior distributions $\log \sigma^2 ,\, \log \theta \sim N(0,4)$ and $\mu \sim N(0,4 \times I_3)$.

    Figure~\ref{fig:AI_convergence_simulation} is based on 1,000 burn-in and 4,000 MCMC iterations; it indicates that increasing 
    $m$ does not affect the mixing of the chain and improves, in some cases, the approximation to the marginal densities.  We then run a longer MCMC 
    chain of $50000$ iterations while varying the value of $m$ over $10,\,50,\,100$ and $200$. These chains are thinned out  after a burn-in period of $2000$ iterations and  samples of $4000$ points  are collected from the target distributions. Figure~\ref{fig:AI_posterior_simulation_coef} shows that the kernel density estimations of the marginal posterior distributions of $\{\theta,\, \mu,\,\sigma^2 \}$  are approximately the same for $m = 100,\, 200$. Thus, $m = 100$ is considered to provide a sufficiently fine discretization for these data. The average proportions of accepting the bridges after the burn-in period are  $72.7 \%,70.2 \%,68.5 \%$  and $ 67.3\%$ for $m = 10, 50, 100$ and $200$ respectively.

\begin{figure}[t]
		\centering
		\includegraphics[width = \textwidth
		]{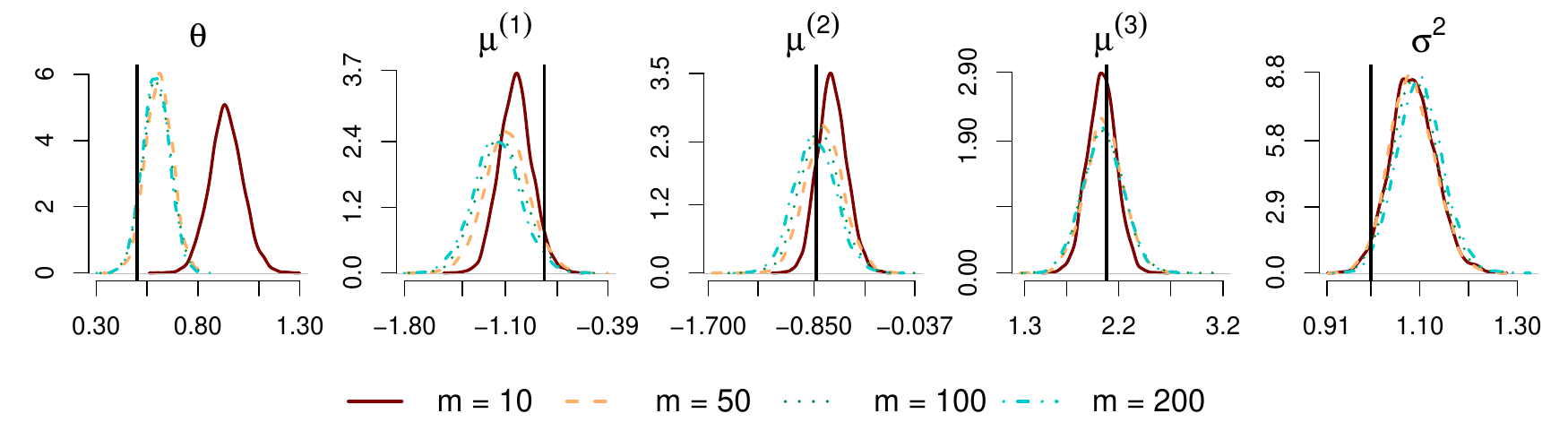}
		\caption{(Simulation study). Estimated posterior distribution of $\{\theta,\mu,\sigma^2\}$ using 
		 $5 \times 10^4$ MCMC iterations ($2\times 10^3$ burn-in discarded, thinned by $12$), and the result for $M = \mathfrak{h}^{-1}(\mu)$ is given in Figure~\ref{fig:supp:AI_posterior_simulation} (Supplementary material). True values are indicated by solid vertical black lines.}
		\label{fig:AI_posterior_simulation_coef}
	\end{figure}    
	\begin{figure}[t]
		\centering
		\includegraphics[width = \textwidth]{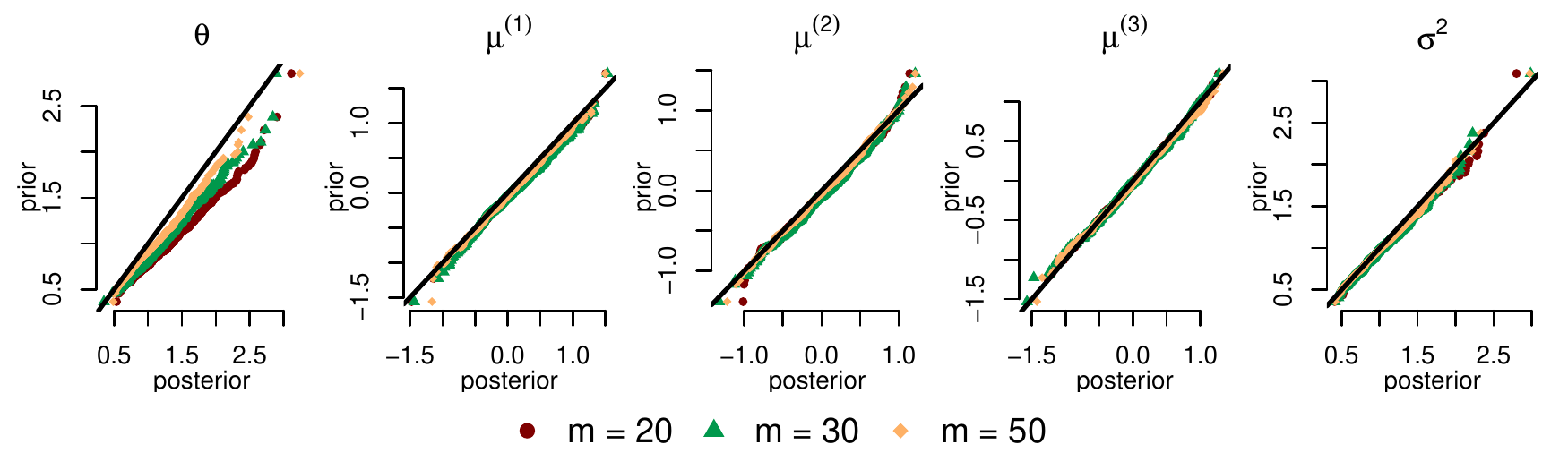}
		\caption{Prior reproduction test  to validate the Algorithm \ref{guided proposal algorithm} with time change: Q-Q plots of priors  against posteriors from three model parameters. 
		}
		\label{fig:AI_prior_reproduction}
	\end{figure}	
	
    Our final investigation for Algorithm \ref{guided proposal algorithm} is the prior reproduction test of S. R. Cook et al. \citep{cook2006validation} to validate the Algorithm \ref{guided proposal algorithm}.  We assume proper priors for $\{\theta, \mu, \sigma^2\}$ : $\log \theta,\, \log \sigma^2 \sim N(0,0.3^2), \, \mu \sim N(0,0.2 \times I_3)$ and $n = 500$. 	
    
    We generate, in turn, $1000$ samples from the prior distributions and then,
    conditional on each sampled parameter vector, high-frequency observations on $[0,50]$ at $5\times 10^5+1$ equidistant time points and keep sub-samples at time points $\{0, 0.1, \ldots,50\}$.  For the $1000$ generated datasets we estimate the corresponding posterior densities using Algorithm \ref{guided proposal algorithm} and test whether they come from the same distribution as the prior as this validates that our algorithm works properly.
    Figure~\ref{fig:AI_prior_reproduction} illustrates that the prior has been successfully replicated while, as expected, the parameters approximation improves as $m$ increases.  

\subsubsection{Higher dimensions } 
        



Since there are $n(n+1)/2$ diffusions to estimate, the computational cost of Algorithm~\ref{guided proposal algorithm} will grow proportionally to $n^2$.  In this subsection, we present a comparison of the performance of the Algorithm in dimensions $n = 2$ and $n = 3$.  We used compiled C++ code in R on a PC running at 2290 MHz with 16 cores. For each dimension, we monitored the times of the central processing unit (CPU) for $10^4$ iterations after suitable burn-in periods and reported the average times per iteration.

Similarly to the case of dimension $n = 2$, we simulated $10^6 + 1$ time points on $[0,100]$ with model parameters 
\[\theta = 0.5\,;\hspace{1cm} M = \begin{pmatrix}
1& 0.7& 0.9\\
0.7& 1.2 &0.9\\
0.9&0.9 & 1
\end{pmatrix}\,;\hspace{1cm} \sigma^2 = 1,\]
and assumed the prior distributions $\log \sigma^2,\, \log \theta \sim N(0,4)$ and $\mu \sim N(0, 4 \times I_6)$ and 
number of imputed points $m$ varying over $10,\, 50$ and $100$.
Traceplots from $ 10^ 4$ iterations after discarding $2 \times 10^3$ iterations of burn-in and the corresponding ACF plots can be found in Figures~\ref{fig:supp:AI_dim_three} and \ref{fig:supp:AI_dim_three_acf} (Supplementary material).

\begin{table}[t]
    \centering
    \begin{tabular}{c|c|c}
    Number of imputed points & Dimension  &  CPU times per iterations (seconds)  \\
    \hline
    \multirow{2}{*}{$m = 10$} & $n = 2$  & 0.0596\\
     & $n = 3$ & 0.408\\
    \hline
    \multirow{2}{*}{$m = 50$} & $n = 2$  & 0.288\\
     & $n = 3$ & 2.08\\
    \hline
    \multirow{2}{*}{$m = 100$} & $n = 2$  & 0.541\\
     & $n = 3$ & 4.27\\
    \end{tabular}
    \caption{Average CPU times per iterations after suitable burn-in period when using Algorithm~\ref{guided proposal algorithm}. }
    \label{tab:CPU_times}
\end{table}


Table~\ref{tab:CPU_times} indicates  that working with dimension $n = 3$ requires more computational resources. In fact, it takes roughly 7 times longer to run for all three cases of $m \in \{10, 50, 100\}$ compared to when working with dimension $n = 2$. However, the MCMC running time is still entirely manageable, even in the case of imputing the highest number of points $m = 100$. Illustrations for the convergence of the MCMC chains for dimension $n = 3$ are shown in Figures~\ref{fig:supp:AI_dim_three} and \ref{fig:supp:AI_dim_three_acf} (Supplementary material).
\section{Application to finance}
\subsection{Introduction}
    The widespread availability of intra-day high-frequency prices of financial assets has enabled the computation of consistent estimates of daily covariation of
	asset prices called realized covariances, introduced in \citep{andersen2001distribution2,andersen2001distribution} and studied in detail by O. E. Barndorff-Nielsen \& N. Shephard \citep{barndorff2004econometric}.  
	The literature in discrete time series involves approaches based on Wishart-based distributions, see  \citep{asai2014stochastic, golosnoy2012conditional, gourieroux2009wishart,jin2013modeling, yu2017generalized, jin2019bayesian}; matrix decompositions to deal with the positive definiteness requirement of the elements of the covariance matrices, see  \citep{bauer2011forecasting, chiriac2011modelling}; or ideas borrowed from the literature of multivariate GARCH models, see  \citep{noureldin2012multivariate, hansen2014realized}.  There has been quite a lot of evidence that direct modelling of realized covariances provides more precise forecasts than GARCH and stochastic volatility multivariate models that assume that the covariance matrices are unobserved latent matrices; see \citep{golosnoy2012conditional, gourieroux2009wishart, chiriac2011modelling, bauer2011forecasting, noureldin2012multivariate, golosnoy2012conditional}.
	
	One important practical question in this framework is how covariance matrices vary at different time scales.  It is well known since the paper by T. W. Epps \citep{epps1979comovements} that correlations between assets decrease with the duration of investment horizons and this necessitates models that are frequency independent.  Modelling realized covariances with diffusions offers a critical advantage over discrete time models because they allow  inference of implied model dynamics and properties as well as forecasting at various frequencies that may differ from the observed data frequency.  Moreover, continuous time models are useful in irregularly spaced observed data and provide a clear advantage when used in pricing derivative instruments.
	
\subsection{Data preprocessing}

    \begin{figure}[t]
		\centering
		\includegraphics[width = \textwidth]{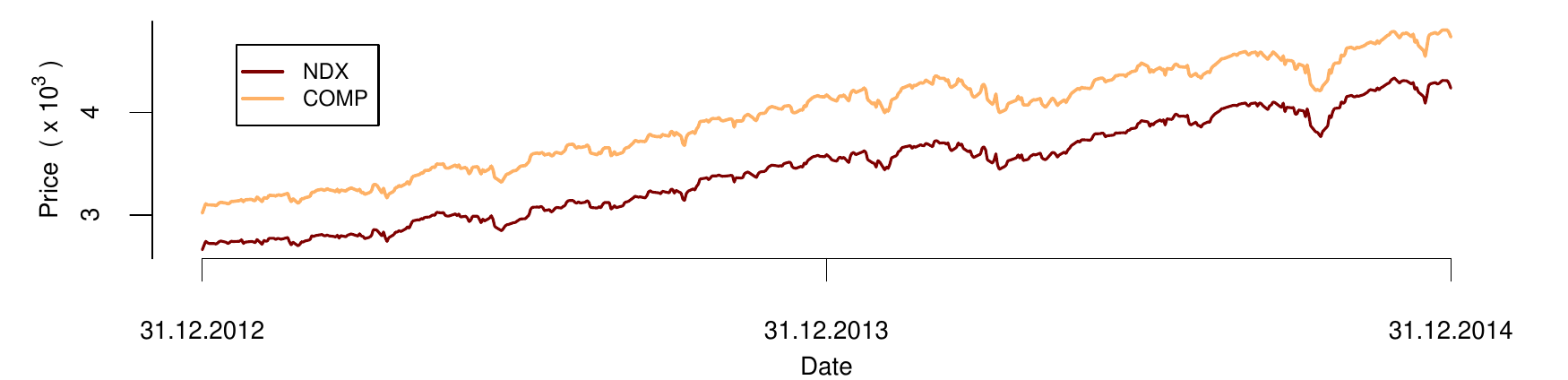}
		\caption{Time series of the closing price from two indices: NASDAQ Composite (COMP) and NASDAQ $100$ (NDX) at the end of each working day from $31.12.2012$ to $31.12.2014$.}
		\label{fig:NDX_COMP}
	\end{figure}

	\begin{figure}[t]
		\centering
		\includegraphics[width = \textwidth]{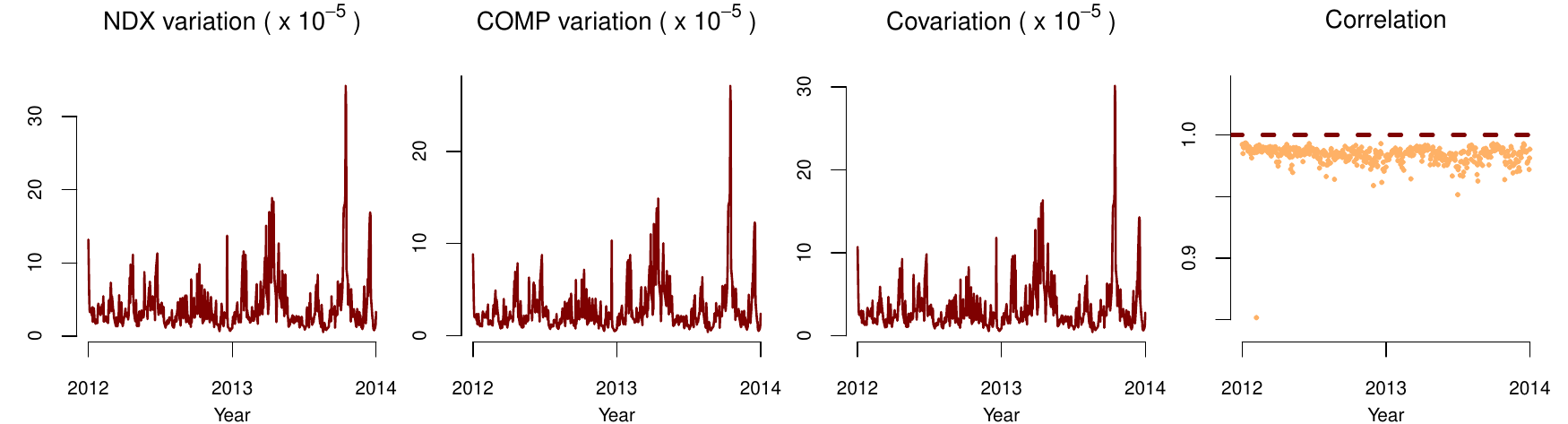}
		\caption{Time series of estimated covariance matrices based on the price of two indices NDX and COMP. Estimated correlations between indices are plotted against time (red dashed line indicates correlation $1$).}
		\label{fig:cov_describe}
	\end{figure}
	
	We estimate $2 \times 2$ daily covariance (volatility) matrices from  NASDAQ Composite (COMP) and NASDAQ $100$ (NDX) indices with data of $504$ working days obtained  from $31.12.2012$ to $31.12.2014$ at $1-$minute intervals from \citep{firstRate}; see Figure~\ref{fig:NDX_COMP}. For the estimation we used quadratic variation/covariation in the logarithm domain and assumed that the microstructure noise does not impact our estimates because the indices are very liquid \citep{zhang2011estimating}.	We verified this assumption by noting that estimates based on $5$-minute inter-observation intervals are very similar.
	
    The pattern of the time series from the entries of covariance matrices in Figure~\ref{fig:cov_describe} indicates a  mean-reverting tendency and many observations lying close to the boundary of $\mathcal{SP}(2)$, making the Euclidean metric inappropriate and the Riemannian structures suitable.

\subsection{Model fitting} \label{Model fitting}

    Our time series has unevenly spaced observations due to weekends and holidays, so the imputed points $m_j$ between the $(j-1)\text{th}$ and $j\text{th}$ consecutive observations are carefully chosen such that $\delta_t = (t_j -t_{j-1})/m_j$ is constant.  We choose vague proper priors 	$\log \theta \sim N(0,4)$ , $\log \sigma^2 \sim N(-1,4)$ and  	$ \mu \sim N\left( \left(\begin{array}{ccc} -12&-12&3 \end{array}\right)^T, 4\, I_3\right)$. We adopt the algorithm by G. O. Roberts \& O. Stramer \citep{roberts2001inference} in the case of the LE metric and Algorithm \ref{guided proposal algorithm} with time change for the case of the AI metric. Trace plots and ACF plots when using either the LE metric with $\delta_t = 0.01$ or the AI metric with $\delta_t = 0.001$ are shown in Figure~\ref{fig:supp:model_convergence} (Supplementary Material).
    MCMC samples based on $10^5$ iterations with varying values of $\delta_t$ were collected after a burn-in period of 4000 iterations and kernel density estimations are depicted in Figure~\ref{fig:AI_LE_posterior_model_coef}. 
    
	\begin{figure}[t]
		\centering
		\includegraphics[width = \textwidth]{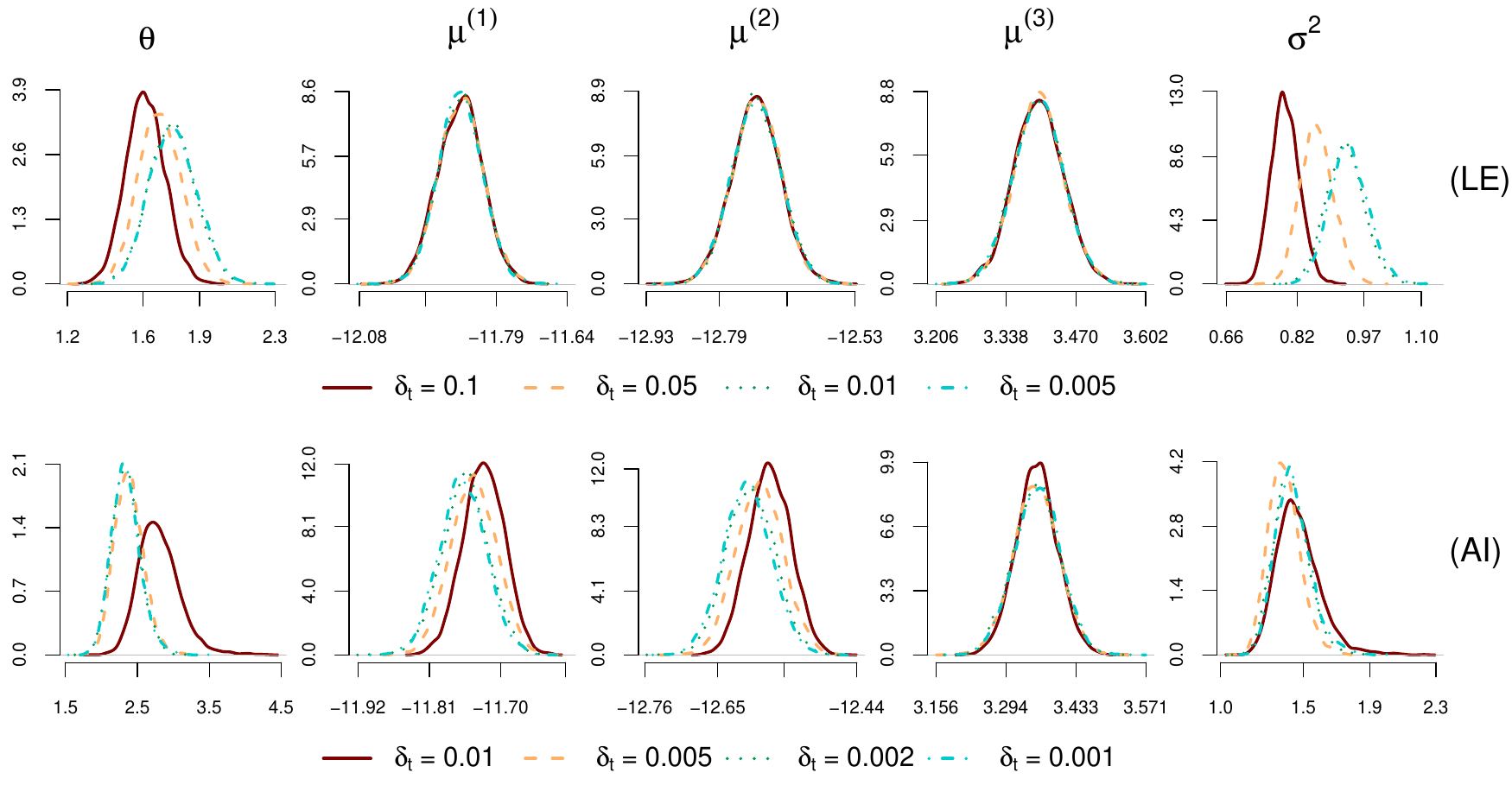}
		\caption{Estimated posterior distribution for $\{\theta,\mu,\sigma^2\}$ from financial data example using either the LE (top row) or the AI metric (bottom row), based on $10^5$ MCMC iterations ($4\times 10^3$ burn-in discarded, thinned by $19$). Moreover, the result for $M = \mathfrak{h}^{-1}(\mu)$ is given in Figure~\ref{fig:supp:AI_LE_posterior_model} (Supplementary material).}
		\label{fig:AI_LE_posterior_model_coef}
	\end{figure}    
	
    Fewer imputed points are required for the LE metric than for the AI metric to adequately approximate the posterior densities which we attribute to different degrees of non-linearity: after transformation through the matrix logarithm, the LE problem is reduced to a linear problem whereas with the AI metric, discretization including approximation of the horizontal lift takes place in the original domain. This domain can be seen to be less linear from the fact that covariance matrices are commutative under the logarithm product which is at the heart of the LE metric whereas exchanging the order of multiplication causes a different result for the AI metric.  This difference is more pronounced near the boundary of the cone which is where the majority of our observations lie.

	\begin{figure}[t]
		\centering
		\includegraphics[ width = \textwidth]{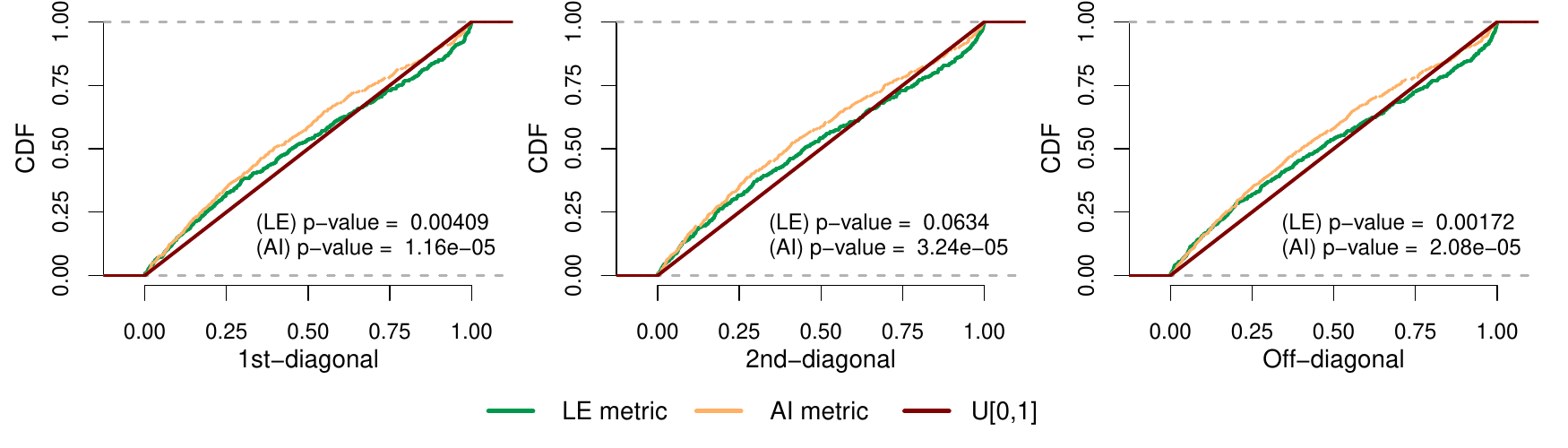}
		\caption{Empirical cumulative distribution of the generalized residuals for the entries from volatilities under two models using either the LE and the AI metric, in which model parameter estimates are the posterior mean. Moreover, p-values are obtained from the K-S tests of comparison to  $U[0,1]$ (red line).}
		\label{fig:AI_LE_gof}
	\end{figure}
	
    We test the fit of the two models using the transition density-based approach by Y. Hong \& H. Li \citep{hong2005nonparametric}. For each model, we choose the posterior mean to estimate $\{\theta,M,\sigma^2\}$ and compute the generalized residuals $Z_j^{(i)}$, for $1 \leq j \leq 504$, as
    	\begin{equation}
            Z_j^{(i)} = \int_{-\infty}^{y_j^{(i)}} p^{(i)}(t_j,v|t_{j-1},y_{j-1}) dv.
        \label{eq:generalized residuals}
    	\end{equation}
    where $\{y_j^{(i)}\}_{j=0,\, i=1}^{504,\, 3}$ denote the observations,  $p^{(i)}$ are the marginal transition densities and  $i=1,2,3$ denotes the two diagonal and the off-diagonal entries respectively. The integral in equation~\eqref{eq:generalized residuals} is estimated by simulating $k = 3000$ points at $t_j$  via equation~\eqref{eq:OU via exp} starting from $y_{j-1}$ at time $t_{j-1}$. Under the null hypothesis that the observations come from the model, the realized generalized residuals $\{ Z_j^{(i)} \}_{j=0}^{504}$ are i.i.d and follow the standard uniform distribution $U[0,1]$ for all $i\in \{1,2,3\}$. The empirical cumulative distribution of these $Z_j^{(i)}$ are shown in Figure~\ref{fig:AI_LE_gof} with p-values from  the Kolmogorov–-Smirnov (K-S) test when comparing to $U[0,1]$. 

Figure~\ref{fig:AI_LE_gof} clearly indicates that the model using the LE metric fits the data better than the one using the AI metric for this particular dataset.

\section{Discussion}

	In summary, it is clear that the Euclidean metric, despite its simplicity, should not be used on $\mathcal{SP}(n)$. We instead suggest using either the LE or the AI metric. While both metrics behave similarly, the difference escalates when moving toward the boundary of $\mathcal{SP}(n)$ and the demonstrated goodness of fit testing can aid model choice. Although the AI metric generally has increased computational cost over the LE metric and exhibits slightly worse fit for our financial data example, it provides an alternative diffusion process on $\mathcal{SP}(n)$ with available Bayesian estimation. Moreover, the AI metric  leads to less anisotropy than the LE metric, which is desirable in some other application areas, e.g. diffusion tensor imaging.  
	
    Cartan-Hadamard manifolds are diffeomorphic to Euclidean space, and the diffeomorphism maps can be obtained from the exponential map at any point, see \citep{carmo1992riemannian,jost2008riemannian}. Furthermore, following the work of H. Karcher \citep{karcher1977riemannian}, we note that $\nabla_{X_t}d^2$ can be expressed in terms of the logarithm map for more general manifolds so that our choice of drift is available for even more general Riemannian manifolds. Thus, our approach of sampling diffusion bridges and our proposed class of OU processes 
	can be extended to any Cartan-–Hadamard manifold on which there exists a suitable diffusion process with an explicit transition density function, e.g. hyperbolic spaces \citep{matsumoto2001closed,nagano2019wrapped}. This opens up potential applications to phylogenetic trees \citep{nye2011principal} and electronic engineering \citep{huckemann2010mobius} where this or closely related geometries are used.  Extension to volatilities observed indirectly through the price processes generalizing the univariate models in \citep{beskos2013advanced} seems both feasible and practically relevant.


\section*{Acknowledgement}
    \hspace{\parindent} Mai Ngoc Bui acknowledges financial support from the UCL Overseas Research Scholarship. The authors would like to thank Stephan Huckemann for helpful discussions.
\bibliographystyle{abbrvnat}
\bibliography{ms}
\newpage
\section*{Supplementary material}
\theoremstyle{plain}
\newtheorem{thm}{Theorem}[subsection]
\renewcommand{\thethm}{S\arabic{subsection}.\arabic{thm}}
\newtheorem{cor}[thm]{Corollary}
\renewcommand{\thecor}{S\arabic{cor}}
\newtheorem{rem}[thm]{Remark}
\renewcommand{\therem}{S\arabic{rem}}
\newtheorem{lem}[thm]{Lemma}
\renewcommand{\thelem}{S\arabic{lem}}
\newtheorem{prop}[thm]{Proposition}
\renewcommand{\theprop}{S\arabic{prop}}

\renewcommand{\theequation}{S\arabic{equation}}
\renewcommand{\thesubsection}{S\arabic{subsection}}
\renewcommand{\thefigure}{S\arabic{figure}}
 \setcounter{equation}{0}
 \setcounter{subsection}{0}
\setcounter{figure}{0}
    \hspace{\parindent}In this section,   we prove one of our main theorem about SDEs on $\mathcal{SP}(n)$ when using the LE metric in subsection~\ref{supp:SDEs with the LE metric}. And, we study the absolute continuity of the measure coming from the guided proposal on $\mathcal{SP}(n)$ equipped with the AI metric in subsection~\ref{supp:absolute continuity}, while provide the calculation of the approximation for $\zeta,\phi$ and $\Phi$. In particular, in this subsection we compute the Riemannian gradient of distance squared and discuss about the stochastic development of smooth curves on $\mathcal{SP}(n)$ for both the LE and AI metrics. Finally,  we discuss stochastic completeness of $\mathcal{SP}(n)$ in subsection~\ref{supp:stoch completeness}  and provide supplementary figures in Section~\ref{supp:Supplementary figures}.

\subsection{Stochastic differential equations on $\mathcal{SP}(n)$ equipped with the LE metric}
\label{supp:SDEs with the LE metric}

		\begin{thm}
    		Suppose the process $X_t$ is the solution of the following SDE on $\mathcal{SP}(n)$ endowed with the LE metric,  for $t \in [0, \tau)$ with  a $\mathfrak{F}_*$-stopping time $\tau$:
    		\begin{equation}\label{supp:eq:solve SDE in LE case}
    			dX_t = A(t, X_t) \, dt + F_{X_t}\big(b(X_t) \, dB_t\big) \hspace{2cm}(X_0 = P),
    		\end{equation}
    		where $A$ assigns smoothly for each $t \in [0,\tau)$ a smooth vector field $A(t,\cdot)$ on $\mathcal{SP}(n)$ and some smooth function $b : \mathcal{SP}(n) \rightarrow \mathbb{R}^{d \times d}$. Moreover, $B_t$ is $\mathbb{R}^d$-valued Brownian motion and the function $F$ is given by
    		
    		\begin{equation*}
		F_Q(e) = \sum_{i=1}^d \epsilon_i\,E_i^{\text{LE}}(Q) \,\,\, \text{ with } \,\,\,e = \sum_{i=1}^d \epsilon_i \,e_i \in \mathbb{R}^d,\, Q \in \mathcal{SP}(n) \text{ and } d = \frac{n(n+1)}{2},
	\end{equation*}
	where the basis $\mathfrak{B}_d^{\text{LE}} = \{E_i^{\text{LE}}\}_{i=1}^d$ is  the orthonormal basis field  on $T\, \mathcal{SP}(n)$ with respect to the LE metric. Then the problem of solving the SDE~\eqref{supp:eq:solve SDE in LE case} on $\mathcal{SP}(n)$ is the same as solving the following SDE on $\mathbb{R}^d$ :
    		\begin{equation}
    			dx_t = a(t,x_t)\,dt + \tilde{b}(x_t)\,dB_t \hspace{2.1cm} (x_0 = p),
    		\label{supp:eq:solve SDE in LE case2}
    		\end{equation}
    		Here, $p = \mathfrak{h}(P)$, $x_t = \mathfrak{h}(X_t)$ hold for all $t \in [0,\tau)$ and smooth function $\tilde{b}$ is given by $\tilde{b} = b\circ \mathfrak{h}^{-1}$. In addition,  smooth function $a = \big(a^{(j)}\big)$  is given by  
    		\[a^{(j)} : [0,\tau) \times \mathbb{R}^d \rightarrow \mathbb{R}, \hspace{1cm} \,(t,x_t) \mapsto \big\langle D_{X_t}\log. A(t,X_t) ,S_j\big\rangle_F \,\,\,\,\,\text{ for all } 1 \leq j \leq d.\]
    	\label{supp:SDE LE}
    	\end{thm}
	\begin{proof}

As $A(t,\cdot)$ is a smooth vector field on $\mathcal{SP}(n)$ for all $t \in [0,\tau)$, there always exist functions $f^{(j)} \in C^{\infty}\big([0,\tau) \times \mathcal{SP}(n)\big)$ for all $ 1 \leq j \leq d$ such that
    		\[
    		    A(t,X_t) = \sum_{j=1}^d f^{(j)}(t,X_t) \, E_j^{\text{LE}}(X_t) = \sum_{j=1}^d f^{(j)}(t,X_t) \, \Big(\big(d\log_{X_t}\big)^{-1} (S_j)\Big).
    		\]
    		Then applying the differential map $d\log$ at $X_t$ on both sides, we get 
    		\[d\log_{X_t}\big(A(t,X_t)\big) = \sum_{j=1}^df^{(j)}(t,X_t)\, S_j,\]
    		which is equivalent to $D_{X_t}\log. A(t,X_t) = \sum_{j=1}^df^{(j)}(X_t)\, S_j$. And since the matrix logarithm function $\log $ is smooth on $\mathcal{SP}(n)$, $D_{X_t}\log. A(t,X_t)$ is smooth with respect to $t$ and $X_t$. Moreover,  we can easily deduce $f^{(j)}(t,X_t) = \big\langle D_{X_t}\log. A(t,X_t) ,S_j\big\rangle_F $ as $\{S_i\}_{i=1}^d$ are orthonormal. Thus, $a^{(j)}(t,x_t)$  simply equals  $f^{(j)}(\mathfrak{h}^{-1}(t,x_t))$, and since $\mathfrak{h}^{-1}$ and $f$ are both smooth, the function $a$ is smooth on $[0,\tau) \times \mathbb{R}^d$.
    		
    		On the other hand, for any $ e = \sum_{j=1}^d \epsilon_j \, e_j \in \mathbb{R}^d$, the definition of $F$  implies
    		\begin{align*}
    		    F_{X_t}(e) &= \sum_{j=1}^d \epsilon_j\, E_j^{\text{LE}}(X_t) = \big(d \log_{X_t}\big)^{-1} \left(\sum_{j=1}^d \epsilon_j\, S_j \right) \\
    		    &\iff d\log_{X_t}\big(F_{X_t}(e)\big) = \sum_{j=1}^d \epsilon_j\, S_j.
    		\end{align*}
    	    We know that $\log:\mathcal{SP}(n) \rightarrow \mathcal{S}(n)$ is a diffeomorphism, and therefore we can apply the result in \cite[Proposition 1.2.4, Page 20]{hsu2002stochastic}, which says if $X_t$ is the solution of the SDE~\eqref{supp:eq:solve SDE in LE case}, the process $\log X_t$ is a solution of the following SDE
    	    \begin{align*}
    	        d\log_{X_t}\big(d X_t\big)  &=  d\log_{X_t} \big(A(t,X_t)\big) \, dt + d\log_{X_t} \Big(F_{X_t}\big\{b(X_t)\, dB_t\big\}\Big)
    	    \end{align*}
    	    \[    	        \iff \sum_{j=1}^d dx_t^{(j)} \,S_j = \sum_{j=1}^da^{(j)}(t,x_t) \, S_j\, dt + \sum_{i,j = 1}^d \tilde{b}_{ij}(x_t)\, S_j\,dB^i_t.\]
    The second order term of the Ito formulation vanishes because $\mathcal{SP}(n)$ equipped with the LE metric has null sectional curvature everywhere. Removing the standard symmetric basis $\mathfrak{B}_d = \{S_i\}_{i=1}^d$, we get the desired result.
	\end{proof}
\subsection{Absolute continuity of the guided proposal on $\mathcal{SP}(n)$ equipped with the Affine-Invariant metric} \label{supp:absolute continuity}
 Let us fix the standard basis $\{e_i\}_{i=1}^d$ for $\mathbb{R}^d$ and  consider the problem of  $X_t$ being the OU process on $\mathcal{SP}(n)$ equipped with the AI metric  with its law $\mathbb{P}_t$  :
	\begin{equation}
		dX_t = -\frac{\theta}{2}\,\nabla_{X_t} \left\{d^2(X_t,M)\right\} \, dt + F_{X_t}( \sigma \, dB_t) \hspace{1 cm } (X_0 = P),
		\label{supp:eq:SDE of OU}
	\end{equation}
	where the smooth function $F: \mathcal{SP}(n) \times \mathbb{R}^d \rightarrow \Gamma(T\,\mathcal{SP}(n))$ defined by:
	\begin{equation*}
		F_Q(e) = \sum_{i=1}^d \epsilon_i\,E_i^{\text{AI}}(Q) \,\,\, \text{ with } \,\,\,e = \sum_{i=1}^d \epsilon_i \,e_i \in \mathbb{R}^d,\, Q \in \mathcal{SP}(n) \text{ and } d = \frac{n(n+1)}{2},
	\end{equation*}

	Our simulation method involves the exponential adapted Euler--Maruyama method, that is
	\begin{equation}
		X_{t+\delta_t} = \text{Exp}_{X_t}\Bigg\{-\frac{\theta}{2} \nabla_{X_t}\{d^2(X_t,M)\}\,\delta_t + \sum_{j = 1}^d (B^{(j)}_{t+\delta_t} - B^{(j)}_t)\, \sigma \, E_j(X_t)\Bigg\} \,\,\text{ for } \delta_t >0.
		\label{supp:eq:OU via exp}
	\end{equation}
	We want to sample from the target diffusion bridge $X_t^{*} = \{X_t, 0 \leq t \leq T \,| \, X_0 = U, X_T = V\}$ with its corresponding law $\mathbb{P}^{*}_t$ by introducing the guided proposal $X_t^{\diamond}$, which is the solution of the following SDE with its law $\mathbb{P}^{\diamond}_t$:
	\begin{equation}
		dX^{\diamond}_t = \left(\theta \, \text{Log}^{\text{AI}}_{X^{\diamond}_t}\, M + \frac{\text{Log}_{X^{\diamond}_t}^{\text{AI}}\,V}{T-t} \right) \,\, dt +F(X^{\diamond}_t)( \sigma \, dB_t) \hspace{1.5 cm } (X^{\diamond}_0 = U).
		\label{supp:eq:AI guided proposal}
	\end{equation}
	We aim to show the absolute continuity among three measures $\mathbb{P}_t,\, \mathbb{P}_t^*$ and $\mathbb{P}_t^{\diamond}$ in Theorem~\ref{supp:AI equivalent laws} below. In preparation for this main theorem, we will compute the Riemannian gradient of distance square in Proposition~\ref{supp:D_LE_AI} and construct the stochastic development of smooth curves in Proposition~\ref{supp:AI_LE_stochastic development} in the case of the LE and the AI metrics. Moreover, since the expression of the horizontal lift in local coordinates in the case of the AI metric does not admit explicit formulae, we approximate them in Corollary~\ref{supp:AI approximation of zeta} and provide the calculation of Christoffel symbol in Lemma~\ref{supp:AI Chris symbol} and the Laplace-Beltrami operator $\nabla_{\mathcal{SP}(n)}$ to the squared Riemannian distance respect to $\mathfrak{B}_d^{\text{AI}}$ in Lemma~\ref{supp:hess AI}.
	
	\begin{prop}[Riemannian gradient of distance squared] \hspace{1cm}
	\begin{enumerate}
	    \item[(i)](LE metric). 	The set $\mathfrak{B}^{\text{LE}}_d = \{E_i^{\text{LE}}\}_{i=1}^d$ is an orthonormal frame on the tangent bundle $T\mathcal{SP}(n)$, where for any $P \in \mathcal{SP}(n)$:
	    \begin{equation}
		    E_i^{\text{LE}}(P) = (d\log_P)^{-1}(S_i) = D_{\log P}\exp.S_i   \hspace{2cm}  (1 \leq i \leq d).
		\label{supp:eq:LE orthonormal basis}
	    \end{equation}
        Moreover, the Riemannian gradient of distance squared for any fixed point $Q \in \mathcal{SP}(n)$ is   
    	$
    		\Big(\nabla \text{d}^2_{\text{LE}}(P,Q)\Big)_P  =  -2\,D_{\log P}\exp.(\log Q - \log P) = -2\, \text{Log}_P^{\text{LE}}(Q).
    	$
    	\item[(ii)](AI metric). The set $\mathfrak{B}^{\text{AI}}_d = \{E_i^{\text{AI}}\}_{i=1}^d$ is an orthonormal frame on the tangent bundle $T\mathcal{SP}(n)$, where for any $P \in \mathcal{SP}(n)$:
	    \begin{equation}
		    E_i^{\text{AI}}(P) =  P^{1/2} \star S_i  \,\,\,\, \text{ for } 1 \leq i \leq d.
		\label{supp:eq:AI orthonormal basis}
	    \end{equation}
	    Moreover, the Riemannian gradient of distance squared for any fixed point $Q\in \mathcal{SP}(n)$ is 
	    $
		    \Big(\nabla \text{d}^2_{\text{AI}}(P,Q)\Big)_P = -2\,\sum_{i=1}^d \left\langle \log (P^{-1/2}  \star Q),S_i\right\rangle_F\,  E_i^{\text{AI}}(P)= -2 \,\text{Log}^{\text{AI}}_P(Q).
		$
	\end{enumerate}
	\label{supp:D_LE_AI}
	\end{prop}
	
	\begin{proof}
    \begingroup
	\allowdisplaybreaks
    	\begin{enumerate}
    	    \item[(i)] Notice that the function $\log : \mathcal{SP}(n) \rightarrow \mathcal{S}(n)$ is bijective, thus there always exists uniquely $Q_i \in \mathcal{SP}(n)$ such that $\log Q_i - \log P = S_i$ for all $ 1\leq i \leq d$. And $E_i^{\text{LE}}(P)$ simply equals to $\text{Log}_P^{\text{LE}}(Q_i) \in T_P\mathcal{SP}(n)$. Moreover, since matrix exponential and logarithm are smooth on $\mathcal{SP}(n)$, $E_i^{\text{LE}}(P)$ are smooth for all $1\leq i \leq d$. Moreover, for any $1\leq i,j\leq d$, we have 
    		\[
    		    g_P^{\text{LE}}\left(E_i^{\text{LE}}(P),E_j^{\text{LE}}(P)\right) = \left\langle d_P\log E_i^{\text{LE}}(P), d_P\log  E_j^{\text{LE}}(P)\right\rangle_F = \langle S_i,S_j\rangle_F = \delta_{ij},
    		\]
    		where $\delta_{ij}$ is the Kronecker delta. Thus, the first argument is proved.
    		

    	On the other hand,  we calculate the Riemannian gradient of function  $f(P) = d^2_{\text{LE}}(P,Q)$ as seen in \cite{hsu2002stochastic}, and use Levi-Civita theorem (see \cite{boothby1986introduction})  with the help from the formula of the logarithm map associated with the LE metric and Equation~\eqref{supp:eq:LE orthonormal basis}. We get 
    	\begin{align*}
    	  &g_P^{\text{LE}}(\nabla f,E_i^{\text{LE}}) = (E_i^{LE})_Pf =  2g_P\Big(\nabla_{E_i^{\text{LE}}}\text{Log}^{\text{LE}}_PQ,\text{Log}^{\text{LE}}_PQ)\Big)\\
    	  &= 2 \,\big\langle D_{\log P}(\log Q - \log P).S_i, \log Q - \log P \big \rangle_F = -2\, \langle S_i, \log Q - \log P  \rangle_F\\
    	  &\Rightarrow (\nabla f)_P = \sum_{i=1}^d g_P^{\text{LE}}(\nabla f,E_i^{\text{LE}})  \, E_i^{\text{LE}}(P) =  -2\,D_{\log P } \exp.(\log Q - \log P).
    	\end{align*}
		
    		\item[(ii)] Since the function $\log : \mathcal{SP}(n) \rightarrow \mathcal{S}(n)$ is bijective, there always exists uniquely $Q_i \in \mathcal{SP}(n)$ such that $\log (P^{-1/2} \star  Q_i) = S_i$ for all $ 1\leq i \leq d$. And $E_i(P) = \text{Log}_P^{\text{AI}}(Q_i) \in T_P\mathcal{SP}(n)$. Moreover, for any $1\leq i,j\leq d$, we have 
    		\[
    		g_P^{\text{AI}}\left(E_i^{\text{AI}}(P),E_j^{\text{AI}}(P)\right) = \left\langle P^{-1/2}  \star E_i^{\text{AI}}(P) , P^{-1/2} \star E_j^{\text{AI}}(P) \right\rangle_F = \langle S_i,S_j\rangle_F = \delta_{ij},
    		\]
    		where $\delta_{ij}$ is the Kronecker delta. We hence obtain a global orthonormal basis field $\mathfrak{B}_d^{\text{AI }} $ on  $\mathcal{SP}(n)$, by the congruent transformation of $P^{1/2}$.

    	On the other hand, we calculate the Riemannian gradient of the $f(P) = d^2_{\text{AI}}(P,Q)$ as seen in \cite{hsu2002stochastic}, for all $1 \leq i\leq d$: $g_P^{\text{AI}}(\nabla f,E_i^{\text{AI}}) = df_P\big(E_i^{\text{AI}}(P)\big)$.
        Moreover, using properties of the matrix logarithm function, we get 
        \begin{align*}
            \log(P^{-1/2} \star Q) &= \log\left(P^{-1/2} QP^{-1/2}\right) = \log\left(P^{1/2} P^{-1} Q P^{-1/2}\right)\\
            &= P^{1/2} \log\left(P^{-1} Q\right)P^{-1/2} = - P^{1/2} \log\left(Q^{-1} P\right)P^{-1/2}\\
            \Rightarrow f(P) &= \text{tr}\Big(\log^T(P^{-1/2} \star Q) \log(P^{-1/2} \star Q) \Big) = \text{tr}\Big(\log^2\left(Q^{-1} P\right)\Big)
        \end{align*}
    	Let us fix $i \in \{1,\ldots, d\}$ and  consider geodesics $\gamma(t) = P^{1/2}\star \exp (tS_i)$, that satisfy $\gamma(0) = P$, $\gamma'(0) = P^{1/2}\star S_i = E_i(P)$. Then $(f\circ \gamma)(t) = \text{tr}\Big\{\log^2\big[Q^{-1} \big(P^{1/2} \star \exp(tS_i)\big) \Big]\Big\}$. We set $\phi(t) = Q^{-1} \big(P^{1/2}\star \exp(tS_i)\big)$ and apply  the result by M. Moakher \citep{moakher2005differential} to $\phi(t)$:
    	\begin{align*}
    		df_P(E_i(P)) &= \left.\frac{d}{dt} (f\circ \gamma)(t)\right|_{t=0} = \left.\frac{d}{dt} \text{tr}\Big\{\log^2\big(\phi(t)\big)\Big\}\right|_{t=0} \\
    		&= 2\, \text{tr}\Big\{\log\big[\phi(0)\big]\,\big[\phi(0)\big]^{-1}\, \phi'(0)\Big\}\\
    		&= 2\, \text{tr}\Big\{\log \big[Q^{-1} P\big]\,\big[Q^{-1}\big(P^{1/2}\star I_n\big)\big]^{-1}\, Q^{-1} \big(P^{1/2} \star S_i\big) \Big\}\\
    		&= 2\, \text{tr}\Big\{\log \big[P^{-1/2} P^{1/2} Q^{-1} P^{1/2}P^{1/2}\big]\,P^{-1/2}  S_i P^{1/2} \Big\}\\
    		&= -2\, \text{tr}\Big\{ \log\big(P^{-1/2}\star Q\big)   S_i\Big\} = - 2\big\langle \log\big(P^{-1/2}\star Q\big)  ,S_i\big\rangle_F
    	\end{align*}
    	\begin{align*}
    		 \Rightarrow (\nabla f)_P &= \sum_{i=1}^d g^{\text{AI}}\big(\nabla f, E_i^{\text{AI}}\big) \, E_i^{\text{AI}}(P) \\
    		 &=  -2 \sum_{i=1}^d\big\langle \log(P^{-1/2} \star Q),S_i\big\rangle_F\, E_i^{\text{AI}}(P)\\
    		 &= -2\, \text{Log}_P^{\text{AI}}(Q).
    	\end{align*}

    	\end{enumerate}
		\endgroup
	\end{proof}
	\begin{lem}
		The Christoffel symbols at any $P \in \mathcal{SP}(n)$ with respect to the basis $\mathfrak{B}_d^{\text{AI}} $ are
		$
		    \Gamma_{ij}^k(P) = -\left\langle (S_iS_j +S_jS_i)\,,\,S_k\right\rangle_F / 2 
		$, which are constant (i.e. they do not depend on $P$).
		\label{supp:AI Chris symbol} 
	\end{lem}
	\begin{proof}
	    \begingroup
	    \allowdisplaybreaks
		We take the result by M. Moakher et al. \citep{moakher2011riemannian} that the Levi-Civita connection of $\mathcal{SP}(n)$ equipped with the AI metric are given as follows:
		\[
		    (\nabla_XY)_P = -\frac{1}{2}(X_P P^{-1}  Y_P+Y_P  P^{-1} X_P) \hspace{1cm} \text{ for } X, Y \in \Gamma\big(T\mathcal{SP}(n)\big).
		\]
		Replacing $X= E_i^{\text{AI}}(P)$ and $Y = E_j^{\text{AI}}(P)$, we get
	    \[
			(\nabla_{E_i^{\text{AI}}}E_j^{\text{AI}})_P 
			= P^{1/2} \star \left\{-\frac{1}{2} (S_iS_j +S_jS_i)\right\} 
			= \sum_{k=1}^d \left\langle -\frac{1}{2} (S_iS_j +S_jS_i)\, , \, S_k \right\rangle_F E_k^{\text{AI}}(P)
		\]
		\[
		    \Rightarrow \Gamma_{ij}^k(P) = -\frac{1}{2}\Big\langle  (S_iS_j +S_jS_i)\,,\,S_k\Big\rangle_F
		\]
		and clearly, Christoffel symbols do not depend on $P$.
	\endgroup
	\end{proof}

	    \begin{lem}The Laplace-Beltrami operator $\Delta_{\mathcal{SP}(n)}$ and the Hessian to the squared Riemannian distance  with respect to $\mathfrak{B}_d^{\text{AI}}$ equal to $2\, I_d$, that is 
            \[(\text{Hess}f)_P\big(E_i^{\text{AI}}, E_j^{\text{AI}}\big) = 2\,\delta_{ij} \hspace{0.3cm}\& \hspace{0.3cm} \Delta_{\mathcal{SP}(n)}f\big(E_i^{\text{AI}},E_j^{\text{AI}}\big) = 2\, \delta_{ij} \hspace{1cm} \text{ for all } 1\leq i,j \leq d.\]
            where $f(P) = d^2_{\text{AI}}(P, Q)$ and $\delta_{ij}$ is the Kronecker delta. 
            \label{supp:hess AI}
    \end{lem}
\begin{proof}
             Definition of the Hessian given in \cite{hsu2002stochastic} gives 
             \[(\text{Hess}f)_P\big(E_i^{\text{AI}}, E_j^{\text{AI}}\big) = E_i^{\text{AI}}(E_j^{\text{AI}}f)_P - \big(\nabla_{E_i^{\text{AI}}}E_j^{\text{AI}}f\big)_P,\]
            while from the proof of Proposition~\ref{supp:D_LE_AI}, we know that for any $1 \leq i \leq d$, 
            \[(E_i^{\text{AI}}f)_P  = - 2\big\langle \log\big(P^{-1/2}\star Q\big)  ,S_i\big\rangle_F = - 2g_P^{\text{AI}}\big( E_i^{\text{AI}}(P), \text{Log}_P^{\text{AI}}(Q)\big)\]
            Moreover, Levi-Civita theorem in \cite{boothby1986introduction} implies $(E_i^{\text{AI}}f)_P = 2 g_P^{\text{AI}}\big(\nabla_{E_i^{\text{AI}}(P)}\text{Log}_P^{\text{AI}}(Q), \text{Log}_P^{\text{AI}}(Q)\big)$. Since $E^{\text{AI}}_i(P),\, \text{Log}_P^{\text{AI}}(Q)$ and $P $ are invertible, we can conclude $\nabla_{E_i^{\text{AI}}(P)}\text{Log}_P^{\text{AI}}(Q) = - E_i^{\text{AI}}(P) $. Therefore, using Lemma~\ref{supp:AI Chris symbol} we have 
            \begin{align*}
             \big(\nabla_{E_i^{\text{AI}}}E_j^{\text{AI}}&f\big)_P = 2g_P^{\text{AI}} \Big(\nabla_{(\nabla_{E_i^{\text{AI}}}E_j^{\text{AI}})_P} \text{Log}^{\text{AI}}_P(Q),\text{Log}^{\text{AI}}_P(Q) \Big)\\
                &= 2g_P^{\text{AI}} \Big(\sum_{k=1}^d\nabla_{\Gamma_{ij}^k(P) E_k^{\text{AI}}(P)} \text{Log}^{\text{AI}}_P(Q),\text{Log}^{\text{AI}}_P(Q) \Big)\\
                &= 2g_P^{\text{AI}} \Big(\sum_{k=1}^d\Gamma_{ij}^k\,\nabla_{ E_k^{\text{AI}}(P)} \text{Log}^{\text{AI}}_P(Q),\text{Log}^{\text{AI}}_P(Q) \Big) \\
                &= -2g_P^{\text{AI}} \Big(\big(\nabla_{E_i^{\text{AI}}}E_j^{\text{AI}}\big)_P ,\text{Log}^{\text{AI}}_P(Q) \Big).\\
                E_i^{\text{AI}}(E_j^{\text{AI}}&f)_P = - 2\, E_i^{\text{AI}}\,g_P^{\text{AI}}\big( E_j^{\text{AI}}(P), \text{Log}_P^{\text{AI}}(Q)\big)\\
                &= - 2\, g_P^{\text{AI}}\Big( \big(\nabla_{E_i^{\text{AI}}}E_j^{\text{AI}}\big)_P, \text{Log}_P^{\text{AI}}(Q)\Big) - 2\, g_P\Big( E_j^{\text{AI}}(P), \nabla_{E_i^{\text{AI}}(P)}\text{Log}_P^{\text{AI}}(Q)\Big) \\
                &= - 2\, g_P^{\text{AI}}\Big( \big(\nabla_{E_i^{\text{AI}}}E_j^{\text{AI}}\big)_P, \text{Log}_P^{\text{AI}}(Q)\Big) + 2\, g_P\Big( E_j^{\text{AI}}(P), E_i^{\text{AI}}(P) \Big)\\ 
                 &= - 2\, g_P^{\text{AI}}\Big( \big(\nabla_{E_i^{\text{AI}}}E_j^{\text{AI}}\big)_P, \text{Log}_P^{\text{AI}}(Q)\Big) + 2\, \delta_{ij}.       
            \end{align*}
                Thus, $\big(\text{Hess}f\big)_{ij} = 2\,\delta_{ij}$. Since $\mathfrak{B}^d_{\text{AI}}$ is an orthonormal basis with respect to $g^{\text{AI}}$, thus the matrix form $G_{\text{AI}}$ of the metric $g^{\text{AI}}$ equals to $I_d$. Therefore, we also achieve  $\Delta_{\mathcal{SP}(n)}f\big(E_i^{\text{AI}},E_j^{\text{AI}}\big) = 2 \,\delta_{ij}$
        \end{proof}	
        
	\begin{prop}[Horizontal lift of smooth curves] 		 
        Suppose that $\gamma(t)$ is a smooth curve on $\mathcal{SP}(n)$ with $\gamma(0)= P$  and some smooth curve $u_t$ on the frame bundle $\mathscr{F}\big(\mathcal{SP}(n)\big)$ such that  $\pi(u_t) = \gamma(t)$ for all $t > 0$, where $\pi: \mathscr{F}\big(\mathcal{SP}(n)\big) \rightarrow \mathcal{SP}(n)$ is the canonical projection map previously discussed, see \cite{hsu2002stochastic}.
    	\begin{enumerate}
    	    \item[(i)] (LE metric) If $u_t$ is the unique horizontal lift of $\gamma(t)$ from an initial frame $u_0$, where $u_0(e_i) = E_i^{\text{LE}}(P) \,\, $ for all $ 1 \leq i \leq d$, the expression of $u_t$ in local coordinates with respect to $\big\{E_i^{\text{LE}}\big(\gamma(t)\big),e_i\big\}_{i=1}^d$ is   $(\gamma(t),\delta)$, where $\delta = \big(\delta_{ij}\big)$, the Kronecker delta.
    	    
    	    \item[(ii)](AI metric) Suppose the expression of $u_t$ in local coordinates with respect to $\big\{E_i^{\text{AI}}\big(\gamma(t)\big),e_i\big\}_{i=1}^d$  is  $u_t = (\gamma(t),\zeta)$ 
    	with $\zeta = \big(\zeta^i_j\big)$ and $\zeta_j^i : \mathcal{SP}(n) \rightarrow \mathbb{R}$ are differentiable functions, then $u_t$ is the unique horizontal lift of $\gamma(t)$ from an initial frame $u_0$, where $u_0(e_i) = E_i^{\text{AI}}(P) $ for all $ 1 \leq i \leq d$ if and only if functions $\zeta_j^i $ exist uniquely and satisfy for all $1 \leq i,j \leq d$:
    	\begin{equation}
    	    \sum_{r=1}^d \alpha_r\big(\gamma(t)\big) \left\{\big(E_r^{\text{AI}}\zeta_j^i\big)_{\gamma(t)} + \sum_{k=1}^d \zeta_k^i\big(\gamma(t)\big) \, \Gamma_{rk}^j\big(\gamma(t)\big)\right\} = 0,\label{supp:eq:AI horizontal lift coefficient}
    	\end{equation}
    	where the Christoffel symbols $\Gamma_{rk}^j\big(\gamma(t)\big) = \Gamma_{rk}^j $  are given in Lemma~\ref{supp:AI Chris symbol} and smooth functions $\alpha_r : \mathcal{SP}(n) \rightarrow \mathbb{R}$ satisfy $\gamma'(t) = \sum_{r=1}^d \alpha_r\big(\gamma(t)\big) \, E_r^{\text{AI}}\big(\gamma(t)\big) $.
    
    	Moreover, if  $\alpha_r\big(\gamma(t)\big) \neq  0$ for only one $ r \in \{1, \ldots, d\} $ , then functions $\zeta_j^i\big(\gamma(t)\big)$ must satisfy
    	\[
    	    \big(E_r^{\text{AI}}\zeta_j^i\big)_{\gamma(t)} = - \sum_{k=1}^d \zeta_k^i\big(\gamma(t)\big) \, \Gamma_{rk}^j \hspace{1.5cm} \text{ for all }  \, 1 \leq i,j \leq d.
    	\]
    	\end{enumerate}
    	\label{supp:AI_LE_stochastic development}
    	\end{prop}
    	
    	\begin{proof}
    		\begingroup
    	    \allowdisplaybreaks
    	    We are given $\pi(u_t) = \gamma(t)$ with a fixed initial value $\gamma(0) = P$ and an initial frame $u_0 \in \mathscr{F}\big(\mathcal{SP}(n)\big)_P$, thus it is sufficient to show that $u_t(e)$ is parallel along  the curve $\gamma(t)$ for any $e \in \mathbb{R}^d$, i.e. $\nabla_{\gamma'(t)} u_t(e) = 0$. Consider an arbitrary $e \in \mathbb{R}^d$, then $e = \sum_{i=1}^d \epsilon_i e_i$ for some $\epsilon_i \in \mathbb{R}$ and we aim to use  the definition of an affine connection for  $\nabla$, see \cite{boothby1986introduction}.
    	    \begin{enumerate} 
    	        \item[(i)]
    	         We express $\gamma'(t)$ with respect to the basis $\mathfrak{B}_d^{\text{LE}}$, i.e. \[\gamma'(t) = \sum_{i=1}^d \alpha_i\big(\gamma(t)\big) \, E_i^{\text{LE}}\big(\gamma(t)\big),\] where functions $\alpha_i \in C^{\infty}\big(\mathcal{SP}(n)\big)$. Since  $u_t(e) = \sum_{i=1}^d \epsilon_i \, E_i^{\text{LE}}(X_t)$ and $\epsilon_i$ does not depend on $\gamma(t)$ for all $1 \leq i \leq d$,
        		\begin{align*}
        		    \nabla_{\gamma'(t)}u_t(e) &= \nabla_{\big\{\sum_{i=1}^d \alpha_i(\gamma(t)) \, E_i^{\text{LE}}\big(\gamma(t)\big)\big\}} \Big\{\sum_{j=1}^d \epsilon_j \, E_j^{\text{LE}}\big(\gamma(t)\big) \Big\}\\
        		    &= \sum_{i,j= 1}^d \alpha_i\big(\gamma(t)\big)\,\epsilon_j \,\,\Big\{\nabla_{E_i^{\text{LE}}\big(\gamma(t)\big)}E_j^{\text{LE}}\big(\gamma(t)\big)\Big\} = 0.
        		\end{align*}
        		Here we use the result that $\mathcal{SP}(n)$ equipped with the LE metric has null curvature everywhere, i.e. $\nabla_{E_i^{\text{LE}}}E_j^{\text{LE}} = 0\,\, (\forall 1\leq i,j\leq d)$.
        		
        		\item[(ii)] Firstly, let us suppose that there exist such functions $\zeta_j^i$, which satisfy Equation~\eqref{supp:eq:AI horizontal lift coefficient}. Since
        		    $u_t(e) = \sum_{i,j=1}^d \epsilon_i\, \zeta_j^i\big(\gamma(t)\big) \, E_j^{\text{AI}}\big(\gamma(t)\big)$, we use the given condition in Equation~\eqref{supp:eq:AI horizontal lift coefficient} and get
        		\begin{align*}
        			\nabla_{\gamma'(t)}u_t(e) &= \sum_{r=1}^d \alpha_r\big(\gamma(t)\big)\, \left( \nabla_{E_r^{\text{AI}}\big(\gamma(t)\big)}\left\{\sum_{i,j=1}^d e_i\,\zeta_j^i\big(\gamma(t)\big)\,E_j^{\text{AI}}\big(\gamma(t)\big)\right\}\right)\\
        			&=\sum_{r=1}^d \alpha_r\big(\gamma(t)\big) \Bigg(\sum_{i,j=1}^d e_i \Big\{ \big(E_r^{\text{AI}}\zeta_j^i\big)_{\gamma(t)}\, E_j^{\text{AI}}\big(\gamma(t)\big)\\ 
        			& \hspace{0.5cm}+ \, \zeta_j^i(X_t)\,\nabla_{E_r^{\text{AI}}\big(\gamma(t)\big)} E_j^{\text{AI}}\big(\gamma(t)\big)\Big\}\Bigg)\\
        			&=\sum_{r,i,j=1}^d  \alpha_r\big(\gamma(t)\big)\, e_i \, \big(E_r^{\text{AI}}\zeta_j^i\big)_{\gamma(t)}\, E_j^{\text{AI}}\big(\gamma(t)\big)\\
        			& \hspace{0.5cm}+ \sum_{r=1}^d\alpha_r\big(\gamma(t)\big)\sum_{i,j=1}^d e_i \, \zeta_j^i\big(\gamma(t)\big) \left(\sum_{k=1}^d \Gamma_{rj}^k\big(\gamma(t)\big) E_k^{\text{AI}}\big(\gamma(t)\big)\right)\\
        			&=\sum_{i,j=1}^d e_i\,E_j^{\text{AI}}\big(\gamma(t)\big) \\
        			& \hspace{1cm} \left(\sum_{r=1}^d \alpha_r\big(\gamma(t)\big) \left\{\big(E_r^{\text{AI}}\zeta_j^i\big)_{\gamma(t)} + \sum_{k=1}^d \zeta_k^i\big(\gamma(t)\big) \, \Gamma_{rk}^j\right\} \right)
        			 = 0.
        		\end{align*}
        		
        		On the other hand, if $u_t$ is the horizontal lift of $\gamma(t)$ starting from the initial frame $u_0$,   $\nabla_{\gamma'(t)} u_t(e) = 0$. So, for all   $1 \leq i,j \leq d$:
        		\[
        		    \sum_{r=1}^d \alpha_r\big(\gamma(t)\big) \left\{\big(E_r^{\text{AI}}\zeta_j^i\big)_{\gamma(t)} + \sum_{k=1}^d \zeta_k^i\big(\gamma(t)\big) \, \Gamma_{rk}^j\big(\gamma(t)\big)\right\} = 0.
        		\]
        
        		If only one $ r \in \{1, \ldots, d\} $ is such that $\alpha_r\big(\gamma(t)\big) \neq  0$, then clearly functions $\zeta_j^i\big(\gamma(t)\big)$ must satisfy
        		\[\big(E_r^{\text{AI}}\zeta_j^i\big)_{\gamma(t)} = - \sum_{k=1}^d \zeta_k^i\big(\gamma(t)\big) \, \Gamma_{rk}^j \hspace{1.5cm} \text{ for all } 1 \leq i,j \leq d.\]
        		Uniqueness and existence  of $u_t$ result from the fact that $\zeta_j^i$ are the solution of a system of first order linear ordinary differential equations with initial conditions $\zeta_j^i(P) = \delta_{ij}$.
    	    \end{enumerate}
    		\endgroup
    	\end{proof}

	\begin{thm} 
	    \begingroup
	    \allowdisplaybreaks
		For $t \in [0,T)$ the laws $\mathbb{P}_t$,$\mathbb{P}_t^{\diamond}$ and $\mathbb{P}_t^{*}$ are absolutely continuous. Let $p(t,X_t;T,V)$ be the true (unknown) transition density of moving from $X_t$ at time $t$ to $V$ at time $T$  and let $X^{\diamond}_{[0:t]}$ be the path of $X_t^{\diamond}$ from time $0$ to $t$. Then
\begin{align}
	\frac{d\mathbb{P}_t}{d\mathbb{P}_t^{\diamond}}\big(X^{\diamond}_{[0:t]}\big) &= \frac{\exp\big(f(X^{\diamond}_0;\sigma^2)\big)}{\exp\big(f(X^{\diamond}_t;\sigma^2)\big)}  \exp\Big\{\Phi(t,X^{\diamond}_{[0:t]})+ \phi\big(t,X^{\diamond}_{[0:t]}\big)\Big\},
    			\label{supp:eq:derivative1}\\
    			\frac{d\mathbb{P}_t^{*}}{d\mathbb{P}_t^{\diamond}}\big(X^{\diamond}_{[0:t]}\big) &= \frac{p(t,X_t^{\diamond};T,V)}{\exp\big(f(X^{\diamond}_t;\sigma^2)\big)} \, \frac{\exp\big(f(X^{\diamond}_0;\sigma^2)\big)}{p(0,U;T,V)}  \exp\Big\{\Phi\big(t,X^{\diamond}_{[0:t]}\big)+ \phi\big(t,X^{\diamond}_{[0:t]}\big)\Big\},
    			\label{supp:eq:derivative2}
    		\end{align} 
    		where  the functions $f,\, \phi$ and $\Phi$ are defined by
    		\begin{align}
    			&f(X^{\diamond}_t;\sigma^2) =  -\frac{d^2_{\text{AI}}(X^{\diamond}_t,V)}{2\sigma^2(T-t)} = -\frac{\left|\left| \log\left((X^{\diamond}_t)^{-1/2} V (X^{\diamond}_t)^{-1/2}\right)\right|\right|_F^2}{2\sigma^2(T-t)}, \label{supp:eq:AI f}\\
    			 &\phi\big(t, X^{\diamond}_{[0:t]}\big)=  \sum_{i,j=1}^d \int_0^t\frac{ \big(\zeta^i_j(X_s^{\diamond};\Theta)\big)^2\, }{ 2(T-s)} \,ds, \label{supp:eq:AI phi}\\
    			&\Phi \big( t, X^{\diamond}_{[0:t]}\big) = \int\limits_0^t \Bigg( \frac{\theta \,g_{X_s}^{\text{AI}}\left(\text{Log}^{\text{AI}}_{X^{\diamond}_s}M ,\text{Log}^{\text{AI}}_{X_s^{\diamond}}V\right) }{\sigma^2 (T-s)} \nonumber\\
    			 & + \sum_{i,j,r=1}^d  \,  \frac{g_{X_s^{\diamond}}^{\text{AI}} \Big(  \zeta_j^i(X_s^{\diamond};\Theta)\big(\sum_{l=1}^d \zeta_l^i(X_s^{\diamond};\Theta) \Gamma_{jl}^r + \big(E_j^{\text{AI}}\zeta_r^i(\cdot;\Theta)\big)_{X_s^{\diamond}}\big) E_r^{\text{AI}}(X_s^{\diamond})\,,\,\text{Log}^{\text{AI}}_{X_s^{\diamond}}V \Big)}{ 2(T -s)} \Bigg) ds,
    			\label{supp:eq:AI Phi} 
    		\end{align}	
    		with $\Gamma_{jl}^r$ given in Lemma~\ref{supp:AI Chris symbol}. The functions $ \zeta(X_s^{\diamond};\Theta) = \big(\zeta_j^i(X_s^{\diamond};\Theta)\big)$ are the coefficients with respect to the basis $\mathfrak{B}^{\text{AI}}_d$ in the local expression of the horizontal lift, see Proposition~\ref{supp:AI_LE_stochastic development}.
		\endgroup
		\label{supp:AI equivalent laws}
	\end{thm}	

        \begin{proof} For notational simplicity, in this proof we write function $f(X_t^{\diamond})$ instead of $f(X_t^{\diamond};\sigma^2)$.
        
            \begingroup
	        \allowdisplaybreaks
        		Using the Girsanov-Cameron-Martin Theorem in  \citep[Theorem 11C, Page 263]{elworthy1982stochastic}, the measures $\mathbb{P}_t$ and $\mathbb{P}_t^{\diamond}$ are absolutely continuous, and the Radon-Nikodym derivative is given as
        	\begin{align}
        	\frac{d\mathbb{P}_t}{d\mathbb{P}_t^{\diamond}}\big(X^{\diamond}_{[0:t]}\big) &= \exp \Bigg\{-\int_0^t \frac{g^{\text{AI}}_{X^{\diamond}_s}\left(\text{Log}_{X_s^{\diamond}}^{\text{AI}} V\,,\,U_s^{\diamond}( dB_s)\right)}{\sigma(T-s)}   \nonumber\\
        	& \hspace{0.5cm} - \frac{1}{2} \,\int_0^t\frac{\big|\big|\text{Log}_{X_s^{\diamond}}^{\text{AI}}V\big|\big|_{g^{\text{AI}}_{X^{\diamond}_s}}^2}{\sigma^2\,(T-s)^2}\, ds\Bigg\}
        	\label{supp:eq:Girsanov theorem}
        	\end{align}
        	where $U_t^{\diamond}$ is the horizontal lift of the guided proposal process $X_t^{\diamond}$.

        	Moreover, Proposition~\ref{supp:D_LE_AI} and Lemma~\ref{supp:hess AI} imply the following results:
		\begin{align*}
		   & (\nabla f)_{X_t^{\diamond}}  = \frac{\text{Log}^{\text{AI}}_{X_t^{\diamond}}V}{\sigma^2(T-t)}, \\
		  &  \frac{\partial f}{\partial t} = -\frac{d^2_{\text{AI}}(V,X^{\diamond}_t)}{2\sigma^2(T-t)^2}= -\frac{\big|\big|\text{Log}_{X_t^{\diamond}}^{\text{AI}}V\big|\big|_{g^{\text{AI}}_{X^{\diamond}_t}}^2}{2\sigma^2(T-t)^2},\\
		  &  \Delta_{\mathcal{SP}(n)} f\big(\,E_i^{\text{AI}}(X^{\diamond}_t)\,, \, E_j^{\text{AI}}(X^{\diamond}_t)\,\big) = \frac{\delta_{ij}}{\sigma^2(T-t)}.	
		 \end{align*}
		    Since    $U_t^{\diamond}(e_i) = \sum_{j=1}^d\zeta_j^i(X_t^{\diamond};\Theta) \, E_j^{\text{AI}}(X_t^{\diamond})$, we get
		    \[\Rightarrow \nabla_{U_t^{\diamond}(e_i)} U_t^{\diamond}(e_k) = \sum_{j,l}^d\zeta^i_j(X_t^{\diamond};\Theta) \Big(\sum_{r=1}^d\zeta_l^k(X_t^{\diamond};\Theta) \Gamma_{jl}^r  E_r^{\text{AI}}(X_t^{\diamond}) +  \big(E_j^{\text{AI}}\zeta_l^k(\cdot;\Theta)\big)_{X_t^{\diamond}}E_l^{\text{AI}}(X_t^{\diamond})\Big) .\]
		 We then apply It\^o's formula from \citep[Lemma 9B, Page 145]{elworthy1982stochastic} to the smooth function $f$ while using some results:
		\begin{align*}
			 &f(X_t^{\diamond}) - f(X_0^{\diamond}) = \int\limits_0^t \left(\frac{\partial f}{\partial s} + g_{X_s}^{\text{AI}}\left(\theta\,\text{Log}^{\text{AI}}_{X^{\diamond}_s}M + \frac{\text{Log}_{X_s^{\diamond}}^{\text{AI}}V}{T-s}, \nabla f\right) \right)ds\\
			&\hspace{0.2cm}+\int\limits_0^t g_{X_s}^{\text{AI}}\big( U_s^{\diamond}(\sigma \,dB_s)\,,\, \nabla  f\big) +   \frac{1}{2} \int\limits_0^t g_{X_s}^{\text{AI}}\left( \nabla_{U_s^{\diamond}(\sigma\, dB_s)} U_s^{\diamond}(\sigma \,dB_s)\,,\, \nabla
			f\right) \\
			&\hspace{0.2cm}+ \frac{1}{2} \int\limits_0^t \Delta_{\mathcal{SP}(n)} f \big(U_s^{\diamond}(\sigma\,dB_s)\,,\, U_s^{\diamond}(\sigma \,dB_s)\big)\\
			&= - \frac{1}{2} \,\int\limits_0^t \frac{\big|\big|\text{Log}_{X_s^{\diamond}}^{\text{AI}}\big|\big|_{g_{X^{\diamond}_s}^{\text{AI}}}^2}{\sigma^2(T-s)^2} \,ds + \int\limits_0^t   \,\frac{\theta\,g_{X_s}^{\text{AI}}\Big(\text{Log}^{\text{AI}}_{X^{\diamond}_s}M ,\text{Log}_{X_s^{\diamond}}^{\text{AI}}V\Big)}{\sigma^2(T-s)}  \,ds\\
			&\hspace{0.2cm}+ \int\limits_0^t \frac{\big|\big|\text{Log}_{X_s^{\diamond}}^{\text{AI}}V\big|\big|_{g_{X^{\diamond}_s}^{\text{AI}}}^2}{\sigma^2(T-s)^2} \,ds +  \int_0^t \frac{g_{X^{\diamond}_s}^{\text{AI}}\Big(\text{Log}_{X_s^{\diamond}}^{\text{AI}} V\,,\, U_s^{\diamond}(\sigma \, dB_s)\Big)}{\sigma^2\,(T-s)} \\
			&\hspace{0.2cm}+\frac{1}{2\sigma^2\, (T-s)}\int_{0}^t \sum_{i,k=1}^d \, d\big[ B^i,B^k\big]_s \,\Bigg\{ \sum_{j,l =1}^d \delta_{jl}\, \zeta^i_j(X_s^{\diamond};\Theta)\, \zeta_l^k(X_s^{\diamond};\Theta)\,\sigma^2  \\
			&\hspace{0.2cm} + \sum_{ j,r=1}^d \, g_{X_s^{\diamond}}^{\text{AI}} \left( \sigma^2 \, \zeta_j^i(X_s^{\diamond};\Theta)  \left(\sum_{l=1}^d \zeta_l^k(X_s^{\diamond}) \Gamma_{jr}^l + \big(E_j^{\text{AI}}\zeta_r^k(\cdot;\Theta)\big)_{X_s^{\diamond}}\right) E_l^{\text{AI}}(X_s^{\diamond})\,,\, \text{Log}^{\text{AI}}_{X_s^{\diamond}}V \right) \Bigg\}  \\
			&=  \frac{1}{2} \,\int\limits_0^t \frac{\big|\big|\text{Log}_{X_s^{\diamond}}^{\text{AI}}V\big|\big|_{g_{X^{\diamond}_s}^{\text{AI}}}^2}{\sigma^2(T-s)^2} \,ds  + \int_0^t \frac{g_{X^{\diamond}_s}^{\text{AI}}\left(\text{Log}_{X_s^{\diamond}}^{\text{AI}} V\,,\,U_s^{\diamond}( dB_s)\right)}{\sigma(T-s)}  \\
			&\hspace{0.2cm}+ \int\limits_0^t \theta\,g_{X_s}^{\text{AI}}\left(\text{Log}^{\text{AI}}_{X^{\diamond}_s}M ,  \frac{\text{Log}_{X_s^{\diamond}}^{\text{AI}}V}{\sigma^2(T-s)} \right) ds + \frac{1}{2}\sum_{i,j,r=1}^d \int_{0}^t  \zeta_j^i(X_s^{\diamond};\Theta)\\
			& \hspace{1.2cm} \, g_{X_s^{\diamond}}^{\text{AI}} \left(  \Big(\sum_{l=1}^d \zeta_l^i(X_s^{\diamond};\Theta) \Gamma_{jl}^r + \big(E_j^{\text{AI}}\zeta_r^i(\cdot;\Theta)\big)_{X_s^{\diamond}}\Big) E_r^{\text{AI}}(X_s^{\diamond})\,,\, \frac{\text{Log}^{\text{AI}}_{X_s^{\diamond}}V}{ T -s} \right)\, ds\\
			&\hspace{0.2cm} +\sum_{i=1}^d \int_0^t\frac{\sum_{j=1}^d \big(\zeta^i_j(X_s^{\diamond};\Theta)\big)^2\, }{2 (T-s)} \,ds.
		\end{align*}
		Substituting into Equation~\eqref{supp:eq:Girsanov theorem}, we get
		\begin{align*}
			& \frac{d\mathbb{P}_t}{d\mathbb{P}_t^{\diamond}}\big(X^{\diamond}_{[0:t]}\big) = \exp \Bigg[-\big(f(X_t^{\diamond}) - f(X_0^{\diamond})\big) \\
			&\hspace{0.2cm}+  \int\limits_0^t \frac{ds}{2\sigma^2(T-s)} \Bigg\{\theta\,g_{X_s}^{\text{AI}}\left(\text{Log}^{\text{AI}}_{X^{\diamond}_s}M ,  \text{Log}^{\text{AI}}_{X_s^{\diamond}}V \right)  + \sigma^2 \sum_{i,j=1}^d\big(\zeta^i_j(X_s^{\diamond};\Theta)\big)^2 \\
			&\hspace{0.2cm}  + \sigma^2\sum_{i,j,r = 1}^d \, g_{X_s^{\diamond}}^{\text{AI}} \Big(  \zeta_j^i(X_s^{\diamond};\Theta) \Big(\sum_{l=1}^d \zeta_l^i(X_s^{\diamond};\Theta) \Gamma_{jl}^r + \big(E_j^{\text{AI}}\zeta_r^i(\cdot;\Theta)\big)_{X_s^{\diamond}}\Big) E_r^{\text{AI}}(X_s^{\diamond})\,,\,  \text{Log}^{\text{AI}}_{X_s^{\diamond}}V\Big)\Bigg\} \Bigg]\\
			&
			= \frac{\exp\big(f(X^{\diamond}_0)\big)} {\exp\big(f(X^{\diamond}_t)\big)} \exp\Big\{\Phi\big(t,X^{\diamond}_{[0:t]}\big) + \phi(t,X^{\diamond}_{[0:t]})  \Big\}.
		\end{align*}	
		Using the result by M. Schauer \& F. Van Der Meulen et al. \cite{schauer2017guided}, that is
		\[ 
		     \frac{d\mathbb{P}_t^{*}}{d\mathbb{P}_t}(X^{\diamond}_{[0:t]})  = \frac{p(t,X^{\diamond}_t;T,V)}{p(0,U;T,V)},
		\]
		we can easily get the desired result in Equation~\eqref{supp:eq:derivative2}.
		\endgroup
    \end{proof} 
    Since the function $\zeta(X_s^{\diamond};\Theta)$, the coefficients with respect to the basis $\mathfrak{B}^{\text{AI}}_d$ in  the expression of the horizontal lift in local coordinates, cannot be expressed explicitly (as mentioned in Proposition~\ref{supp:AI_LE_stochastic development}), their approximation will be discussed below. For notational simplicity, we write function $ \zeta(X_s^{\diamond})$ instead of $\zeta(X_s^{\diamond};\Theta)$ in Corollary~\ref{supp:AI approximation of zeta} and \ref{supp:AI equivalent law1 approx}.

	\begin{cor}[Approximation of functions $\zeta$]
            \begingroup
    	    \allowdisplaybreaks
    	    Consider the geodesic $\gamma(t)$ on $\mathcal{SP}(n)$ equipped with the AI metric with $\gamma(0) = P$. We  approximate $\zeta_j^i\big(\gamma(t)\big)$ defined in Proposition~\ref{supp:AI_LE_stochastic development} when $t > 0$ is infinitesimally small under the special scenarios for the initial tangent vector $\gamma'(0) \in T_P\mathcal{SP}(n)$ as follows:
    		\begin{enumerate}
    			\item Choose an integer $ l \in \{1, ,\ldots, d\}$ and set $\gamma'(0) = P^{1/2} \star  S_l$, we therefore get  $\gamma(t) = P^{1/2} \star \exp (tS_l)  $. Suppose $u_t$ is the unique horizontal lift of $\gamma(t)$ defined in Proposition~\ref{supp:AI_LE_stochastic development}. For all $ 1 \leq r \leq d$ we have
    			\[
    			    \Rightarrow \big(E_r^{\text{AI}}\zeta_j^i\big)_{P} = \big(d\zeta_j^i\big)_{P}\big(E_r^{\text{AI}}(P)\big) = \left.\frac{d}{dt}(\zeta_j^i \circ \gamma)(t)\right|_{t=0} .
    			\] 
    		    Since $\alpha_r(P) = \delta_{rl}$ and $\zeta_i^j(P) = \delta_{ij}$ (i.e. using information of the initial frame), Proposition~\ref{supp:AI_LE_stochastic development} implies 
    			\begin{align*}
    			 \big(E_l^{\text{AI}}\zeta_j^i\big)_{P} = - \sum_{k=1}^d \zeta_k^i(P)\, \Gamma_{lk}^j &= - \sum_{k=1}^d \delta_{ik}\Gamma_{lk}^j = -\Gamma_{li}^j \,\,\,\,\,\, (\forall 1 \leq i,j \leq d)\\
    			    \Rightarrow \left.\frac{d}{dt}(\zeta_j^i \circ \gamma)(t)\right|_{t=0} &= -\Gamma_{li}^j 
    			    \hspace{0.5cm} \Rightarrow\underset{t \rightarrow 0}{\lim}\frac{(\zeta_j^i \circ \gamma)(t) - \delta_{ij}}{t}  = -\Gamma_{li}^j,
    			\end{align*}
                as $(\zeta_{j}^i\circ \gamma)(0) =  \delta_{ij}$, so
    			$\zeta_j^i\big(\gamma(t)\big) \approx \delta_{ij} - t \Gamma_{li}^j$,  given that $t$ is close to zero.
    			Thus, for all $ 1 \leq i,j\leq d $: 
    			$\zeta_j^i\big(P^{1/2}  \star \exp(t\,S_l) \big) \approx \delta_{ij} - t \Gamma_{li}^j \,\,\hspace{0.1cm} ( 0 \leq t \ll 1).$
    			
    			\item We extend case 1 by considering an arbitrary $e = \sum_{l=1}^d \alpha_l \, e_l \in \mathbb{R}^d$  such that $\alpha_l  \in \mathbb{R}\setminus\{0\}$ for $l \in \mathfrak{I} \subseteq \{1,\ldots,d\}$ and set $\gamma'(0) = P^{1/2} \star \big(u_0(e))$, i.e.  
    			$\gamma(t) = \text{Exp}^{\text{AI}}_P\big(t\,u_0(e)\big)$ with $\alpha_r(X_0)$ in Proposition~~\ref{supp:AI_LE_stochastic development} equals to  $\alpha_r \neq 0$ for $r \in \mathfrak{I}$. Therefore, Proposition~\ref{supp:AI_LE_stochastic development} implies
    			\[
    			   \sum_{r \in \mathfrak{I}} \alpha_r \big(E_r^{\text{AI}}\zeta_j^i\big)_{P} = - \sum_{r \in \mathfrak{I}}\alpha_r \, \sum_{k=1}^d \zeta_k^i(P)\, \Gamma_{rk}^j = - \sum_{r \in \mathfrak{I}} \alpha_r \, \Gamma_{ri}^j.
    			\]
    			On the other hand,  we denote $V = \sum_{l=1}^d \alpha_l \,E_l^{\text{AI}} = \sum_{r \in \mathfrak{I}} \alpha_r \,E_r^{\text{AI}}$ which is a vector field on $\mathcal{SP}(n)$, then 
    			\[
    			    \big(V\zeta_j^i\big)_{\gamma(t)} = \sum_{r \in \mathfrak{I}} \alpha_r \, \big(E_r^{\text{AI}}\zeta_j^i\big)_{\gamma(t)} \, \text{ and  } \,\big(V\zeta_j^i\big)_{P} = \left.\frac{d}{dt} (\zeta_j^i\circ \gamma)(t)\right|_{t=0}
    			\]
    			\[
    			    \Rightarrow \underset{t \rightarrow 0}{\lim}\frac{(\zeta_j^i \circ \gamma)(t) - \delta_{ij}}{t}  =\sum_{r \in \mathfrak{I}} \alpha_r \, (E_r^{\text{AI}}\zeta_j^i)_{P} =  - \sum_{r \in \mathfrak{I}} \alpha_r \, \Gamma_{ri}^j.
    			\]
    			Thus, for all $1 \leq i,j\leq d$:
    			\begin{equation*}
    				\zeta_j^i\left(P^{1/2}  \exp(t\sum_{l =1}^d \alpha_l\, S_l) P^{1/2}\right) \approx \delta_{ij} - t \sum_{l =1}^d   \alpha_l \, \Gamma_{li}^j\, \mathbb{I}_{\{l \in \mathfrak{I}\}},
    			\end{equation*}
    			 for  infinitesimally small $ t > 0$, where $\mathbb{I}$ stands for the indicator function.
    		\end{enumerate}
    		\endgroup
    		\label{supp:AI approximation of zeta} 
        \end{cor}
        
        Using Corollary \ref{supp:AI approximation of zeta}, we are ready to approximate $\phi$ and $\Phi$ in Theorem~\ref{supp:AI equivalent laws}. 

 \begin{cor}
        \begingroup
	    \allowdisplaybreaks
	    Suppose that we have an $\mathcal{SP}(n)$-valued path $\{X_{t_k}^{\diamond} = y^{\diamond}_{t_k}\}_{k=0}^{m+1}$ of the OU process $X_t$ when equipping $\mathcal{SP}(n)$ with the AI metric, in which it is simulated from the exponential adapted Euler-Maruyama method,  where $\max\{t_{k+1}-t_k\}_{k=0}^m$ is sufficiently small and $0 = t_0 < \ldots <t_{m+1} = t$. Then the functions $\phi$ in Equation~\eqref{supp:eq:AI phi} and $\Phi$ in Equation~\eqref{supp:eq:AI Phi} can be approximated as follows:
		\begin{align}
			\phi(t, X^{\diamond}_{[0:t]}) &\approx  \frac{d}{2} \log \frac{T-t}{T}, \label{supp:eq:AI phi approx}\\
			\Phi(t,X^{\diamond}_{[0:t]}) &\approx \sum_{k=0}^m \frac{t_{k+1}- t_k}{T- t_k} \,\Bigg\{  \frac{\theta\,\left\langle \log\big((y_{t_k}^{\diamond})^{-1/2} \star M \big)\,,\, \log\big((y_{t_k}^{\diamond})^{-1/2} \star V  \big)\right\rangle_F}{\sigma^2} \nonumber  \\
    	    & \hspace{1cm} +     \frac{\left\langle  \Gamma\,,\, \log\big((y_{t_k}^{\diamond})^{-1/2} \star V\big) \right\rangle_F}{2}\Bigg\},  \label{supp:eq:AI Phi approx}
		\end{align}	
		with $\Gamma = \sum_{i,r=1}^d\Gamma_{ii}^r  S_r $ and the Christoffel symbols $\Gamma_{ii}^r$ are given in Lemma~\ref{supp:AI Chris symbol}.
		\endgroup
		\label{supp:AI equivalent law1 approx}        
    \end{cor}
    	
	\begin{proof}
    \begingroup
	\allowdisplaybreaks	
		Using the  approximation for $\zeta_j^i$ in Corollary~\ref{supp:AI approximation of zeta}, we get for $0 \leq k \leq m$ and $1 \leq i,j \leq d$:  
		$
		    \zeta^i_j(X_{t_k}^{\diamond}) \approx \delta_{ij}$, 
		    \[\Rightarrow E_l^{\text{AI}} \zeta_j^i = 0 ,\,\, \forall 1 \leq l \leq d \hspace{0.9cm}\& \hspace{0.9cm} U^{\diamond}_{t_k}(e_i) \approx E_i^{\text{AI}}(X_{t_k}^{\diamond}).\] 
		Thus, we get
		\begin{align*}
		   & \phi(t, X^{\diamond}_{[0:t]}) = \sum_{i,j=1}^d \int_0^t\frac{ \big(\zeta^i_j(X_s^{\diamond})\big)^2\, }{ 2(T-s)} \,ds \approx  \sum_{i=1}^d \int_0^t\frac{ 1 }{ 2(T-s)} \,ds=\frac{d}{2}\, \log\frac{T-t}{T},\\
		   & \Phi \big( t, X^{\diamond}_{[0:t]}\big) = \int\limits_0^t \Bigg( \frac{\theta \,g_{X_s}^{\text{AI}}\left(\text{Log}^{\text{AI}}_{X^{\diamond}_s}M ,\text{Log}^{\text{AI}}_{X_s^{\diamond}}V\right) }{\sigma^2 (T-s)} \nonumber\\
    		& \hspace{0.5cm}+ \sum_{i,j,r=1}^d  \,  \frac{g_{X_s^{\diamond}}^{\text{AI}} \Big(  \zeta_j^i(X_s^{\diamond})\big(\sum_{l=1}^d \zeta_l^i(X_s^{\diamond}) \Gamma_{jl}^r + (E_j^{\text{AI}}\zeta_r^i)_{X_s^{\diamond}}\big) E_r^{\text{AI}}(X_s^{\diamond})\,,\,\text{Log}^{\text{AI}}_{X_s^{\diamond}}V \Big)}{ 2(T -s)} \Bigg) ds,\\
			&\hspace{0.2cm} \approx \int\limits_0^t  \frac{\theta\, \left\langle \log\big((X_s^{\diamond})^{-1/2}\star M \big)\,,\, \log\big((X_s^{\diamond})^{-1/2} \star V \big)\right\rangle_F}{\sigma^2(T-s)} \,ds \\
			&\hspace{0.5cm}+ \sum_{i,j,r =1}^d \sum_{k=0}^m \frac{\Big\langle \delta_{ij}\left(\sum_{l=1}^d \delta_{il} \Gamma_{jl}^r \right) S_r\,,\,\log\big(y^{\diamond}_{t_k})^{-1/2} \star V \big)\Big\rangle_F}{2(T-t_k)}\, (t_{k+1}-t_k)\\
			&\hspace{0.2cm}\approx \sum_{k=0}^m \frac{t_{k+1}- t_k}{T- t_k} \,\Bigg\{  \frac{\theta\,\left\langle \log\big((y_{t_k}^{\diamond})^{-1/2} \star M \big)\,,\, \log\big((y_{t_k}^{\diamond})^{-1/2} \star V  \big)\right\rangle_F}{\sigma^2} \nonumber  \\
    	    & \hspace{0.5cm} +     \frac{\left\langle  \Gamma\,,\, \log\big((y_{t_k}^{\diamond})^{-1/2} \star V\big) \right\rangle_F}{2}\Bigg\}.
		\end{align*}
	\endgroup
	\end{proof}
	
\subsection{Stochastic completeness on $\mathcal{SP}(n)$}\label{supp:stoch completeness}
	Before proving stochastic completeness, let us state the calculation of the Ricci curvature in \citep{pennec2020manifold} in the case of the AI metric as well as the theorem in \citep{hsu2002stochastic} showing that manifolds having a lower bound on the Ricci curvature are stochastically complete. 
	\begin{lem} \citep{pennec2020manifold}
		For any $P \in \mathcal{SP}(n)$, the Ricci curvature is given in terms of the basis $\mathfrak{B}_d^{\text{AI}}(P)$, as defined in equation~\eqref{supp:eq:AI orthonormal basis} :
		\[
		    \text{Ric}_P = -\frac{n}{4}\begin{pmatrix}
			I_n -\frac{1}{n}\mathbbm{1}_{n,n}  &0\\
			0& I_{n(n-1)/2}\end{pmatrix},
		\]
		where $\mathbbm{1}_{n,n}$ is an $n\times n $ matrix, that has all entries equal to one.
		\label{supp:AI Ricci curvature}
	\end{lem}
	\begin{thm}\citep{hsu2002stochastic}
		Consider a complete Riemannian manifold $\mathcal{M}$ of dimension $d$, a fixed point $p \in \mathcal{M}$ and denote $\text{d}(x,p)$ as the distance between $x\in \mathcal{M}$ and $p$. Suppose that
		a negative, non-decreasing, continuous function  $ \kappa : [0,\infty) \rightarrow \mathbb{R}_{<0}$ satisfies
		\[\kappa(r) \leq \frac{1}{d-1} \,\,\underset{x\in \mathcal{M}}{\,\inf\,}\,\,\{\text{Ric}_{\mathcal{M}}(x): d(x,p) = r\}\]
		where $\text{Ric}_{\mathcal{M}}(x) = \{\text{Ric}(X,X) : X \in T_x\mathcal{M} \text{ and } |X|= 1\}$. If 
		\[\int_c^{\infty} \frac{1}{\sqrt{-\kappa(r)}} \, dr = \infty \]
		for some constant $c$ then $\mathcal{M}$ is stochastically complete.
		\label{supp:Ricci curvature-stochastically complete}
	\end{thm}

    	\begin{prop}[Stochastic completeness]\label{stochastic completeness} The Riemannian manifold $\mathcal{SP}(n)$ is stochastically complete when it is equipped with either
        	\[\begin{array}{ll}
        	 \text{(i)}    \hspace{1cm}&  \text{the LE metric,}\hspace{9cm}\\
        	  \text{(ii)}  \hspace{1cm} & \text{or the AI metric.\hspace{9cm}}
        	\end{array}\]
    	\end{prop}
    	
    	\begin{proof}
        	\begin{enumerate}
        	    \item[(i)]	Firstly, we show that $\mathcal{SP}(n)$ equipped with the LE metric has null sectional curvature everywhere, i.e. for all $ 1 \leq i,j \leq d$ and $ P \in \mathcal{SP}(n)$ we have $\nabla_{E_i^{\text{LE}}(P)} E_j^{\text{LE}}(P) = 0$, where $E_i^{\text{LE}},E_j^{\text{LE}} \in \mathfrak{B}_d^{\text{LE}}$, as defined in equation~\eqref{supp:eq:LE orthonormal basis}.
		
		We have the result that $\mathcal{SP}(n)$ endowed with the LE metric is isometric to the $\mathcal{S}(n)$ endowed with a Euclidean metric (Frobenius inner product) through the matrix logarithm function, that is $\log : \mathcal{SP}(n) \rightarrow \mathcal{S}(n)$ is a diffeomorphism and for all $P \in \mathcal{SP}(n)$ we have 
		\[
		    g^{\text{LE}}_P\left( S_1,S_2\right) =  \langle d\log_P S_1, d\log_P S_2\rangle_F\hspace{1cm} S_1,S_2 \in T_P\mathcal{SP}(n).
		\]
		So  $\nabla^{\mathcal{SP}(n)}$  is the pull-back connection of $\nabla^{\mathcal{S}(n)}$ by the matrix logarithm function, and we have for any $P \in \mathcal{SP}(n)$ and $1\leq i,j \leq d$ : 
		\[
		    \Rightarrow \nabla^{\mathcal{SP}(n)}_{E_i^{\text{LE}}(P)}E_j^{\text{LE}}(P) =\nabla^{\mathcal{SP}(n)}_{\{(d \log_P)^{-1}S_i\}}\{(d \log_P)^{-1}S_j\}=  (d \log_P)^{-1} \nabla^{\mathcal{S}(n)}_{S_i}S_j = 0.
		\]
		Thus, the Ricci curvature tensor also vanishes everywhere, and the required result is a direct consequence of Theorem~\ref{supp:Ricci curvature-stochastically complete}.
    		\item[(ii)] 
    		Let us fix a point $P\in \mathcal{SP}(n)$, and vary some point $Q\in \mathcal{SP}(n)$ such that $d_{\text{AI}}^2(P,Q) = r$ for some $r >0$. Consider a tangent vector $v \in T_Q\mathcal{SP}(n)$ such that it has unit length, that is $v = \sum_{i=1}^d \nu_i E_i^{\text{AI}}(Q) \in T_Q\mathcal{SP}(n)$ and $||\nu||_2 = 1$.
    		
    		Using Lemma~\ref{supp:AI Ricci curvature}, and denoting the $(i,j)\text{th}$ entry of the Ricci curvature tensor at $Q$ in matrix form as  $\text{Ric}_Q^{(i,j)}$, we have
    		\begin{align*}
    			\text{Ric}(v,v) &= \sum_{i,j= 1}^d \nu_i \text{Ric}_Q^{(i,j)} \nu_j \\
    			&=  -\frac{n-1}{4} \sum_{i=j=1}^n \nu_i^2 -\frac{n}{4} \sum_{i=j=n+1}^d \nu_i^2 + \frac{1}{4}\sum_{i\neq j}^n \nu_i\nu_j\\
    			&= -\frac{n}{4} +  \frac{1}{4}\sum_{i,j=1}^n \nu_i \nu_j \geq -\frac{n}{4} -\frac{n^2}{4} = -\frac{d}{2}
    		\end{align*}
    		The last inequality holds due to the fact that $||\nu||_2 = 1$. Therefore, using Theorem~\ref{supp:Ricci curvature-stochastically complete}, we can set $\kappa(r) = -d/\big(2(d-1)\big)$, which is clearly a negative, non-decreasing, continuous function.
    		\[
    		    \Rightarrow \int_{c}^{\infty}\frac{1}{\sqrt{-\kappa(r)}}\,dr = \sqrt{\frac{2(d-1)}{d}} \int_{c}^{\infty} 1 \,dr = \infty \,\,\,\,\,\,\text{ for some constant } c.
    		 \]
    	\end{enumerate}
    
    	\end{proof}
 
\subsection{Figures}
\label{supp:Supplementary figures}

    We include the kernel density estimations for the marginal posterior distribution of the model parameters $\{\theta,M,\sigma^2\}$ in the simulation study (Fig.~\ref{fig:supp:AI_posterior_simulation}) and the application in finance (Fig.~\ref{fig:supp:AI_LE_posterior_model}), which correspond to the estimated posterior distribution for $\{\theta,\mu,\sigma^2\}$ in Fig.~\ref{fig:AI_posterior_simulation_coef} and Fig.~\ref{fig:AI_LE_posterior_model_coef}, respectively. Furthermore, trace plots and ACF plots when using either the LE metric with $\delta_t = 0.01$ or the AI metric with $\delta_t = 0.001$ are shown in Figure~\ref{fig:supp:model_convergence}.
    
    	\begin{figure}[H]
    		\centering
    		\includegraphics[width = \textwidth]{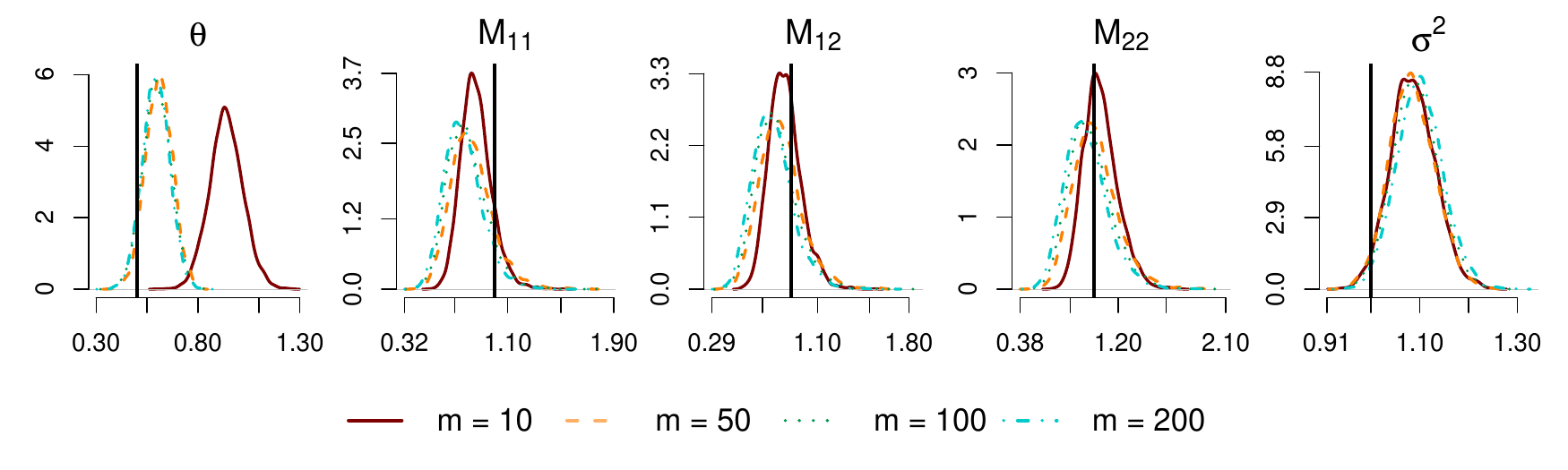}
    		\caption{(Simulation study on $\mathcal{SP}(2)$). Estimated posterior distribution of $\{\theta,M,\sigma^2\}$ using 
    		 $5 \times 10^4$ MCMC iterations ($2\times 10^3$ burn-in discarded, thinned by $12$). True values are indicated by solid vertical black lines.}
    		\label{fig:supp:AI_posterior_simulation}
    	\end{figure}
    	
    	\begin{figure}[H]
    		\centering
    		\includegraphics[width = \textwidth, height = 11cm]{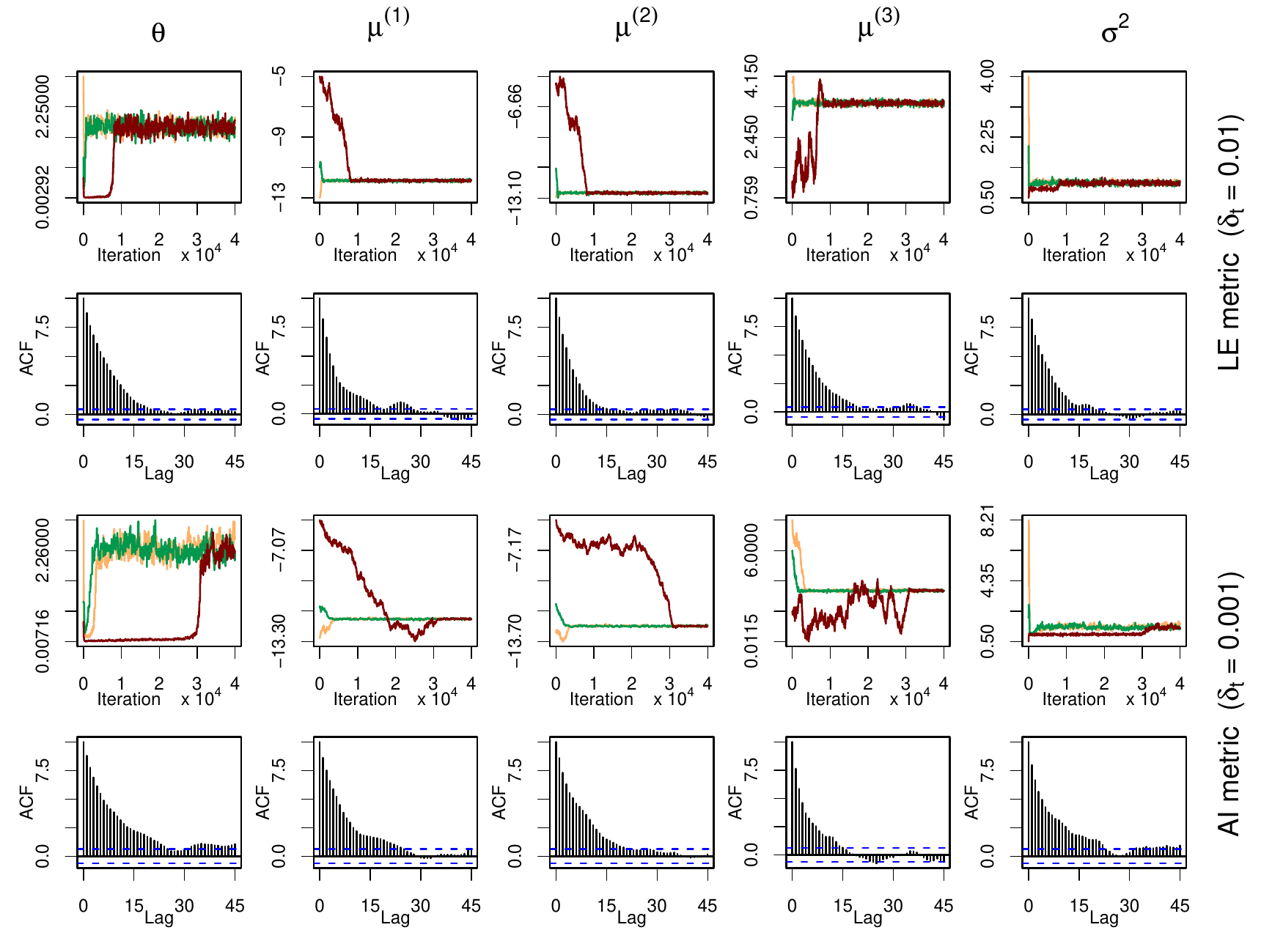}
    		\caption{(Application in finance). MCMC trace plots of 4000 iterations using different starting points (orange, green and red) and ACF plots based on iterates $1000-4000$ of the green chain  with $\delta_t = 0.01$ in the case of the LE metric and $\delta_t = 0.001$ in the case of the AI metric .}
    		\label{fig:supp:model_convergence}
    	\end{figure}
    	
    		\begin{figure}[H]
    		\centering
    		\includegraphics[ width = \textwidth, height = 7.5 cm]{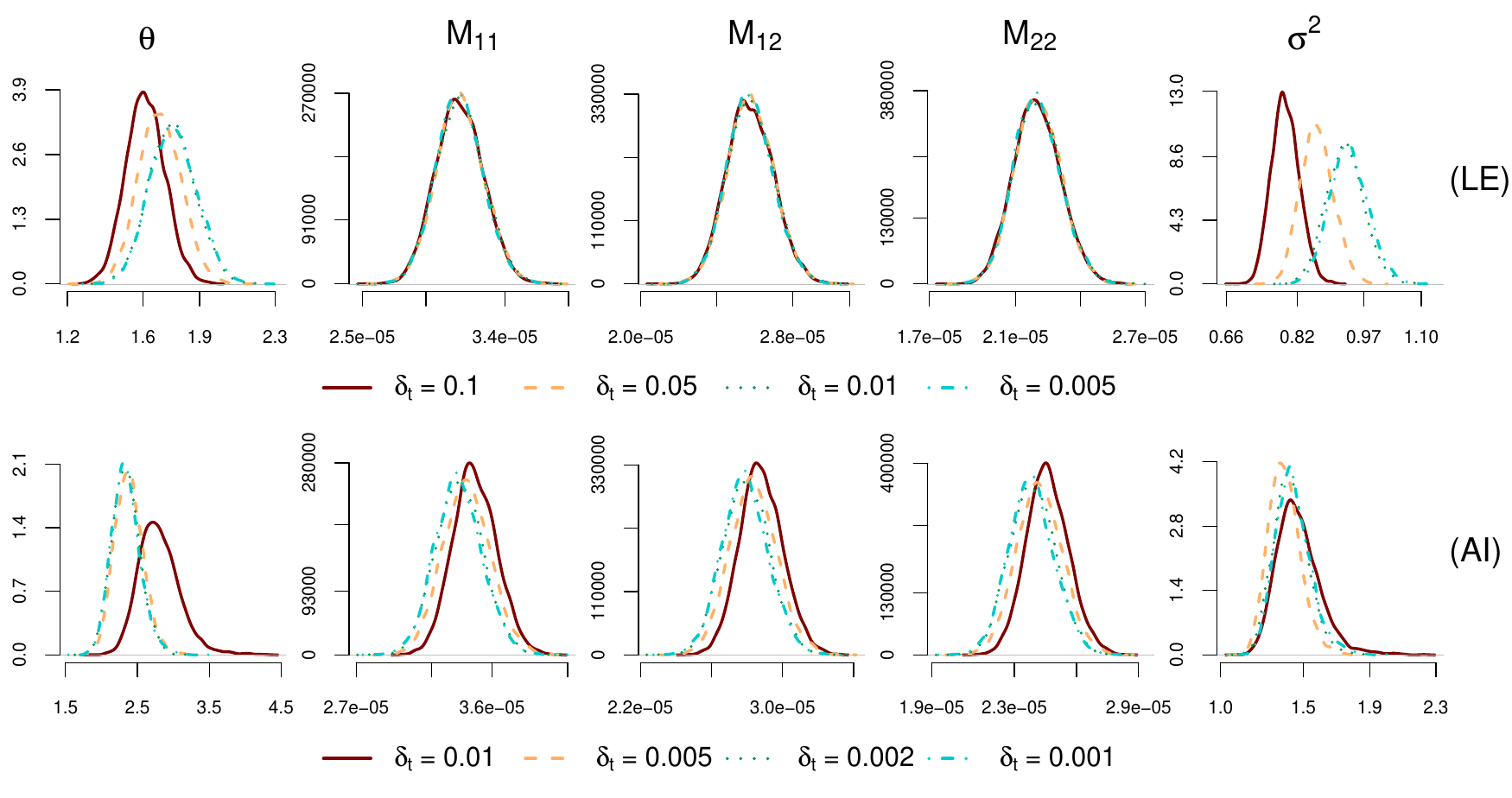}
    		\caption{(Application in finance). Estimated posterior distribution for $\{\theta,M,\sigma^2\}$ using either the LE (top row) or the AI metric (bottom row), based on $10^5$ MCMC iterations ($4\times 10^3$ burn-in discarded, thinned by $19$).}
    		\label{fig:supp:AI_LE_posterior_model}
    	\end{figure}

    		\begin{figure}[H]
    		\centering
    		\includegraphics[ width = \textwidth]{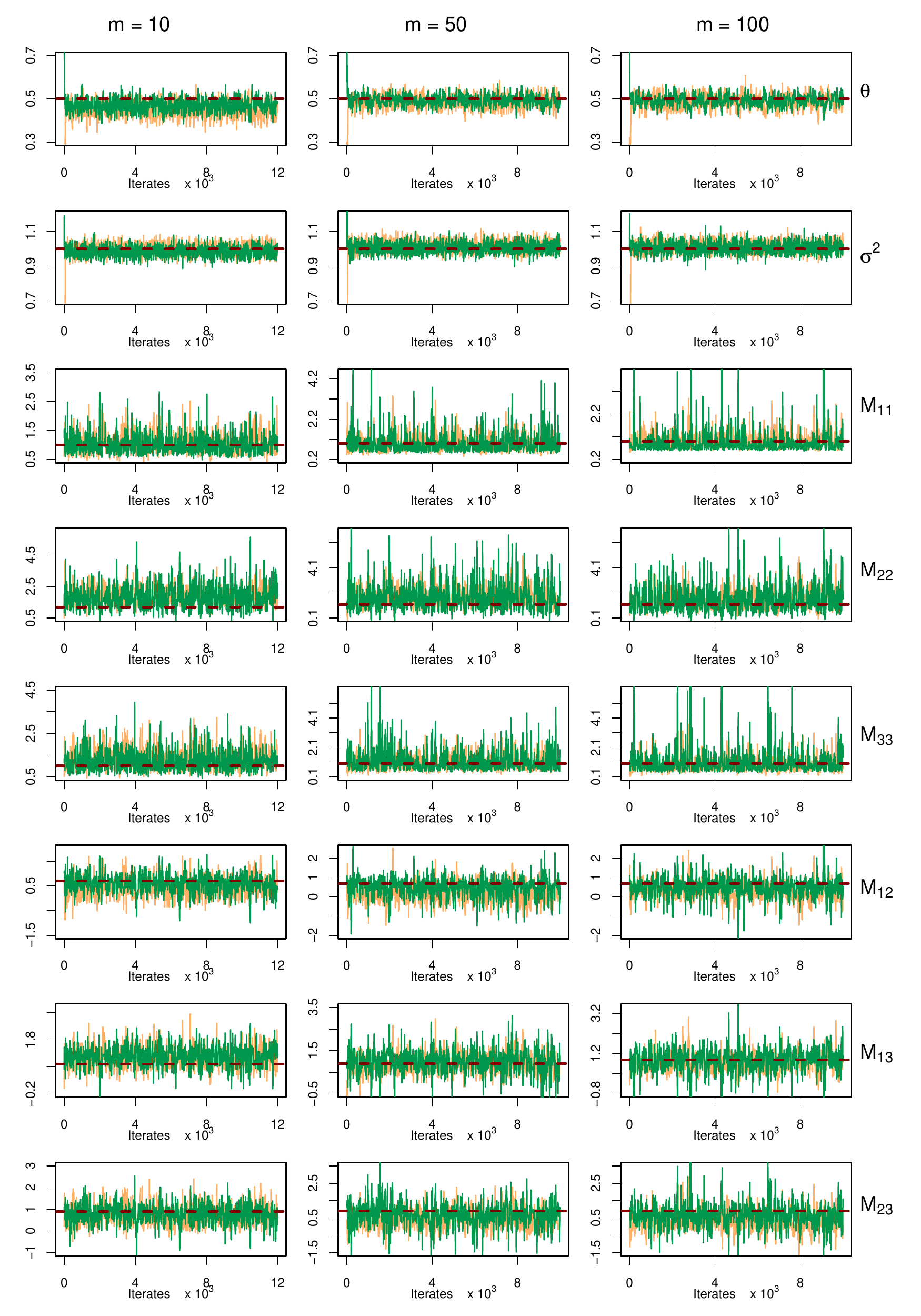}
    		\caption{(Simulation study on $\mathcal{SP}(3)$). MCMC traceplots of $12 \times 10^3$ iterations using different start points (orange, green) when number of imputed points $m$ are varied over $10,\, 50$ and $100$.}
    		\label{fig:supp:AI_dim_three}
    	\end{figure}
    		\begin{figure}[H]
    		\centering
    		\includegraphics[ width = \textwidth]{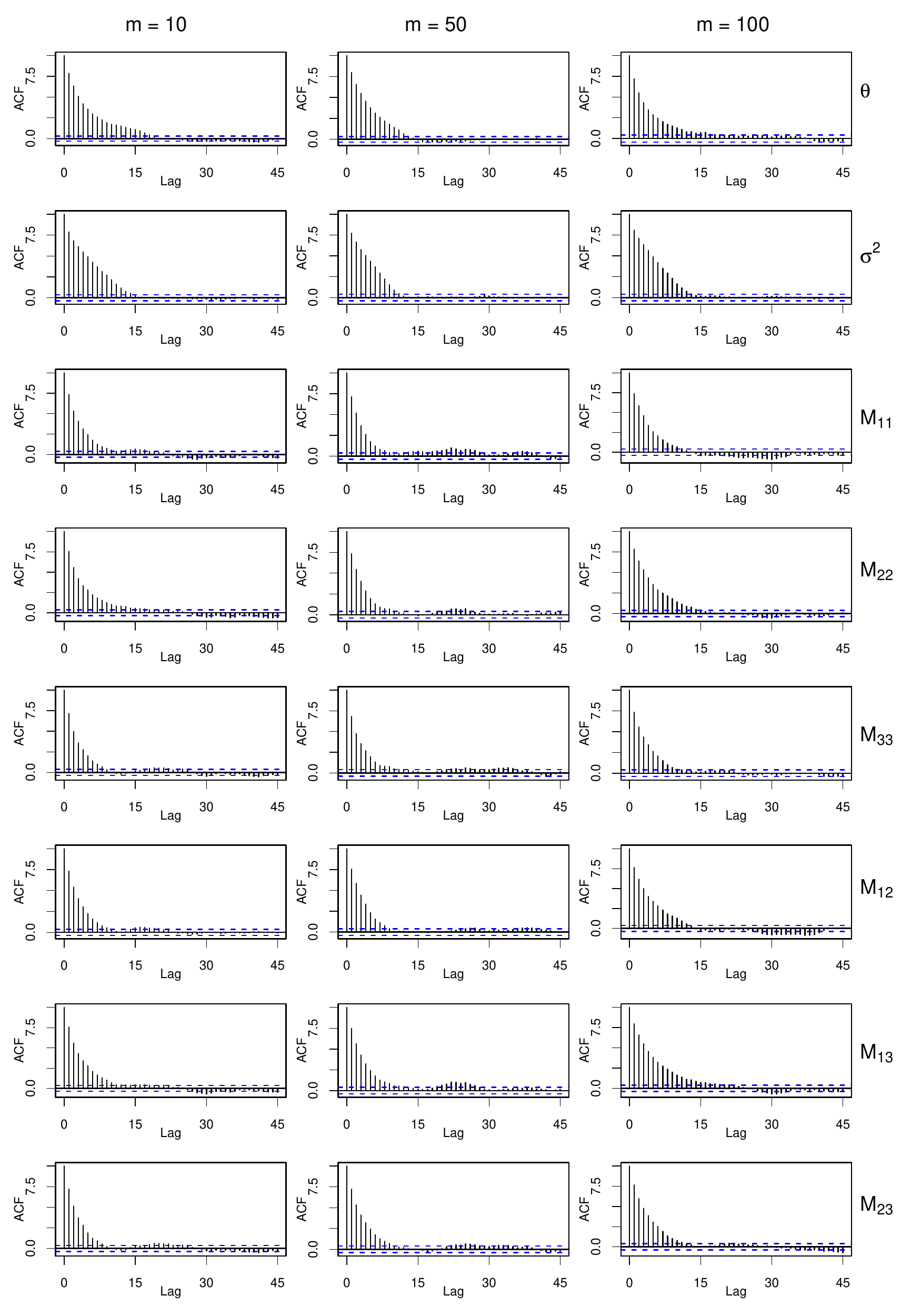}
    		\caption{(Simulation study on $\mathcal{SP}(3)$). ACF plots based on iterates $10^3 - 12 \times 10^3$ of the orange chain in Figure~\ref{fig:supp:AI_dim_three} when number of imputed points $m$ are varied over $10,\, 50$ and $100$.}
    		\label{fig:supp:AI_dim_three_acf}
    	\end{figure}

\end{document}